\documentclass[preprint,a4paper,10pt,dvipsnames]{elsarticle}
\usepackage[T1]{fontenc}
\usepackage[utf8]{inputenc}
\usepackage[english]{babel}
\usepackage{amsmath,amsthm,enumerate}
\usepackage{amssymb}
\usepackage{amscd}
\usepackage{vmargin}
\usepackage{url}
\usepackage{amsmath}
\usepackage{amssymb}
\usepackage{mathrsfs}
\usepackage{stmaryrd}
\usepackage{tikz,pgf}
\usepackage{subfig}
\usepackage{cancel}
\usepackage{comment}
\usetikzlibrary{positioning}
\usetikzlibrary {patterns,patterns.meta}
\usepackage{subfig}
\usepackage{xspace}

\newcommand{\si}{\sigma}

\newcommand{\np}{\textsf{NP}\xspace}
\newcommand{\SAN}{\textsc{Signed Achromatic Number}\xspace}
\newcommand{\kSAN}{\textsc{Signed Achromatic Number~$k$}\xspace}
\newcommand{\AN}{\textsc{Achromatic Number}\xspace}
\newcommand{\kAN}{\textsc{Achromatic Number~$k$}\xspace}

\begin{document}
\begin{frontmatter}

\title{Generalising the achromatic number to \\ Zaslavsky's colourings of signed graphs\tnoteref{thanks}}
\tnotetext[thanks]{The authors were supported by ANR project HOSIGRA (ANR-17-CE40-0022), and by IFCAM project ``Applications of graph homomorphisms'' (MA/IFCAM/18/39).}

\author[nice]{Julien Bensmail}
\author[bordeaux]{François Dross}
\author[lyon,lyon2]{Nacim Oijid}
\author[bordeaux]{\'Eric Sopena}

\address[nice]{Universit\'e C\^ote d'Azur, CNRS, Inria, I3S, France}
\address[bordeaux]{Univ. Bordeaux, CNRS,  Bordeaux INP, LaBRI, UMR 5800, F-33400, Talence, France}
\address[lyon]{\'Ecole Normale Sup\'erieure de Lyon, 69364 Lyon Cedex 07, France }
\address[lyon2]{Univ. Lyon, Universit\'e Lyon 1, LIRIS UMR CNRS 5205, F-69621, Lyon, France}

\journal{...}

\begin{abstract}
The chromatic number, which refers to the minimum number of colours required to colour the vertices of graphs properly,
is one of the most central notions of the graph chromatic theory.
Several of its aspects of interest have been investigated in the literature,
including variants for modifications of proper colourings.
These variants include, notably, the achromatic number of graphs, which is the maximum number of colours required to colour the vertices of graphs properly
 so that each possible combination of distinct colours is assigned along some edge.
The behaviours of this parameter have led to many investigations of interest,
bringing to light both similarities and discrepancies with the chromatic number.

This work takes place in a recent trend aiming at extending the chromatic theory of graphs to the realm of signed graphs,
and, in particular, at investigating how classic results adapt to the signed context.
Most of the works done in that line to date are with respect to two main generalisations of proper colourings of signed graphs,
attributed to Zaslavsky and Guenin.
Generalising the achromatic number to signed graphs was initiated recently by Lajou,
his investigations being related to Guenin's colourings.
We here pursue this line of research, but with taking Zaslavsky's colourings as our notion of proper colourings.
We study the general behaviour of our resulting variant of the achromatic number, 
mainly by investigating how known results on the classic achromatic number  generalise to our context.
Our results cover, notably, bounds, standard operations on graphs, and complexity aspects.

\end{abstract}

 
\end{frontmatter}

\newtheorem{theorem}{Theorem}[section]
\newtheorem{lemma}[theorem]{Lemma}
\newtheorem{conjecture}[theorem]{Conjecture}
\newtheorem{observation}[theorem]{Observation}
\newtheorem{claim}{Claim}
\newtheorem{corollary}[theorem]{Corollary}
\newtheorem{proposition}[theorem]{Proposition}
\newtheorem{question}[theorem]{Question}

%%%%%%%%%%%%%%%%%%%%%%%%%%%%%%%%%%%%%%%%%%%%%%%%%%%%%%%%%%%%%%%%%%%%
%%%%%%%%%%%%%%%%%%%%%%%%%%%%%%%%%%%%%%%%%%%%%%%%%%%%%%%%%%%%%%%%%%%%
%%%%%%%%%%%%%%%%%%%%%%%%%%%%%%%%%%%%%%%%%%%%%%%%%%%%%%%%%%%%%%%%%%%%
%%%%%%%%%%%%%%%%%%%%%%%%%%%%%%%%%%%%%%%%%%%%%%%%%%%%%%%%%%%%%%%%%%%%
%%%%%%%%%%%%%%%%%%%%%%%%%%%%%%%%%%%%%%%%%%%%%%%%%%%%%%%%%%%%%%%%%%%%
%%%%%%%%%%%%%%%%%%%%%%%%%%%%%%%%%%%%%%%%%%%%%%%%%%%%%%%%%%%%%%%%%%%%
%%%%%%%%%%%%%%%%%%%%%%%%%%%%%%%%%%%%%%%%%%%%%%%%%%%%%%%%%%%%%%%%%%%%
%%%%%%%%%%%%%%%%%%%%%%%%%%%%%%%%%%%%%%%%%%%%%%%%%%%%%%%%%%%%%%%%%%%%
%%%%%%%%%%%%%%%%%%%%%%%%%%%%%%%%%%%%%%%%%%%%%%%%%%%%%%%%%%%%%%%%%%%%
%%%%%%%%%%%%%%%%%%%%%%%%%%%%%%%%%%%%%%%%%%%%%%%%%%%%%%%%%%%%%%%%%%%%
%%%%%%%%%%%%%%%%%%%%%%%%%%%%%%%%%%%%%%%%%%%%%%%%%%%%%%%%%%%%%%%%%%%%
%%%%%%%%%%%%%%%%%%%%%%%%%%%%%%%%%%%%%%%%%%%%%%%%%%%%%%%%%%%%%%%%%%%%
%%%%%%%%%%%%%%%%%%%%%%%%%%%%%%%%%%%%%%%%%%%%%%%%%%%%%%%%%%%%%%%%%%%%
%%%%%%%%%%%%%%%%%%%%%%%%%%%%%%%%%%%%%%%%%%%%%%%%%%%%%%%%%%%%%%%%%%%%
%%%%%%%%%%%%%%%%%%%%%%%%%%%%%%%%%%%%%%%%%%%%%%%%%%%%%%%%%%%%%%%%%%%%
%%%%%%%%%%%%%%%%%%%%%%%%%%%%%%%%%%%%%%%%%%%%%%%%%%%%%%%%%%%%%%%%%%%%
%%%%%%%%%%%%%%%%%%%%%%%%%%%%%%%%%%%%%%%%%%%%%%%%%%%%%%%%%%%%%%%%%%%%
%%%%%%%%%%%%%%%%%%%%%%%%%%%%%%%%%%%%%%%%%%%%%%%%%%%%%%%%%%%%%%%%%%%%
%%%%%%%%%%%%%%%%%%%%%%%%%%%%%%%%%%%%%%%%%%%%%%%%%%%%%%%%%%%%%%%%%%%%
%%%%%%%%%%%%%%%%%%%%%%%%%%%%%%%%%%%%%%%%%%%%%%%%%%%%%%%%%%%%%%%%%%%%

\section{Introduction}

In this work, we introduce and study a generalisation of \textbf{complete colourings} and of the \textbf{achromatic number} usually defined for undirected unsigned graphs
to \textbf{signed graphs}, in particular to \textbf{colourings of signed graphs} as originally defined by Zaslavsky in the 1980's.
Before entering further into the details of our contribution, let us first give a general reminder about all the notions and concepts involved.

\medskip

In the context of (undirected and unsigned) graphs, the notion of (vertex-)colouring usually refers to \textit{proper colourings},
where a $k$-colouring of a graph $G$, i.e., an assignment $V(G) \rightarrow \{1,\dots,k\}$ of colours $1,\dots,k$ to the vertices, is proper if no two adjacent vertices are assigned the same colour.
Alternatively, a proper $k$-colouring of $G$ can be seen as a \textit{homomorphism} from $G$ to $K_k$, the complete graph on $k$ vertices,
where, recall, a homomorphism $h$ from $G$ to a graph $H$ is a vertex-mapping $h:V(G) \rightarrow V(H)$ preserving the edges (that is, $h(u)h(v)$ is an edge in $H$ whenever $uv$ is an edge in $G$).

These two definitions of proper colourings lead to several interesting questions.
Such questions include finding colourings that are not only proper but also somewhat optimal for some notion of optimality.
In this line, surely the most investigated question, motivated notably by numerous practical applications,
is, given a graph $G$, to determine its \textit{chromatic number} $\chi(G)$, which is defined as the smallest $k$ such that proper $k$-colourings of $G$ exist.
Other possible ways to consider that a colouring is optimal do not involve the number of colours only
but also its quality regarding some criteria.
This is well illustrated through \textit{complete colourings}, where a proper colouring of $G$ is complete if for any two distinct colours $i$ and $j$,
there exist two adjacent vertices being assigned colour~$i$ and~$j$, respectively.
From the homomorphism point of view, note that complete colourings are equivalent to homomorphisms to complete graphs that are edge-surjective.
Regarding this notion, given a graph $G$, we are usually interested in determining its \textit{achromatic number} $\psi(G)$, which is defined as the largest $k$ such that complete $k$-colourings of $G$ exist.

\medskip

Whenever notions and problems are originally defined for graphs, 
an interesting direction for research is to wonder about possible generalisations to more general structures, such as directed graphs, hypergraphs, signed graphs, etc.
In the current work, we are mainly interested in \textit{signed graphs}.
A signed graph $(G,\sigma)$ is defined as a graph $G$ equipped with a \textit{signature} $\sigma: E(G) \rightarrow \{-,+\}$ through which every edge is either \textit{negative} (assigned sign $-$) or \textit{positive} (otherwise, assigned sign $+$).
Graph signatures can freely be transposed from a graph $G$ to any of its subgraphs $H$, in the sense that if $(G,\sigma)$ is a signed graph, then $(H,\sigma)$ should be understood as the signed graph defined over $H$ by applying the signature $\sigma$ to the edges of $H$.

When dealing with signed graphs, we often take into consideration a peculiar graph operation, called the \textit{switching operation},
where switching a set of vertices $S$ of $(G,\sigma)$ means changing the signs of all edges in the cut $(S,V(G)\setminus S)$, resulting in another signed graph $(G,\sigma')$ defined over $G$. 
In case where $S=\{v\}$, i.e., $S$ is a singleton, by switching the vertex $v$ we mean switching $S$.
The main point behind this switching operation is that it yields equivalence classes over the possible signed graphs defined over a given graph,
which, for some problems, showed up to be particularly judicious to consider.
Problems and notions for signed graphs have actually been, sometimes, defined accordingly to these equivalence classes.
The introduction of the switching operation itself was also justified by the need of a dynamic model to study particular real-life problems,
issued notably from the study of social networks.

The common chromatic theory of graphs being probably one of the most investigated fields of graph theory, 
a legitimate wonder is about its possible generalisations to signed graphs.
In particular, for a possible generalisation, it is interesting to investigate the possible discrepancies with the original theory.
As such, two main parallel lines of generalisation for the chromatic theory to signed graphs have been studied throughout the years,
leading to many interesting questions. 
These two lines involve different ways of adapting proper colourings to signed graphs,
attributed usually to Zaslavsky and Guenin, respectively.
One significant difference between these two variants lies within the switching operation.
In Guenin's variation, one downside is that proper colourings are not preserved by the switching operation,
while they are in Zaslavsky's variation.

\begin{itemize}
	\item The most recent of these two variations, introduced in the 2000's by Guenin~\cite{Gue05},
is with respect to the homomorphism point of view.
We will not elaborate further on this variant as it is rather distant from the second variation below, which is actually the main point of focus in this work.
Let us mention, however, that this branch has been developing into a rich and nice theory,
which already includes one way, introduced and studied by Lajou~\cite{Laj19}, of generalising complete colourings and the achromatic number to signed graphs.
Again, due to the quite different definitions involved, this variant hardly compares to the one we consider in the current work.
Nevertheless, we refer the interested reader to~\cite{RST21} for a recent survey on this interesting topic.

	\item The second, oldest, variation was introduced by Zaslavsky as early as in the 1980's~\cite{Zas82},
and has been regaining a lot of attention in the recent years, due notably to the work~\cite{MRS16} of M\'a\v{c}ajov\'a, Raspaud and \v{S}koviera.
The main idea, here, is to have a notion of proper colouring of signed graphs that 1) is reminiscent of the notion of proper colouring of unsigned graphs,
and 2) takes into account the very peculiarity of signed graphs, being the switching operation.
Particularly, Zaslavsky's colourings of signed graphs, which will be defined thoroughly in next Section~\ref{section:preliminaries}, 
are preserved under switching vertices.
We refer the reader to~\cite{SV21} for a recent survey on this field.
\end{itemize}

In this work, as a counterpart to Lajou's work, we introduce and study a generalisation of complete colourings and the achromatic number of graphs to signed graphs with respect to Zaslavsky's colourings.
Our main way to proceed is by going through the survey~\cite{HM97} by Hughes and MacGillivray 
and Chapter~12 of the monograph~\cite{CZ09} by Chartrand and Zhang
on the original achromatic number,
and to wonder about more or less natural ways to generalise known results to our context.
The main results we come up with are about the influence, on our modified definition of achromatic number, of modifying the structure of a given signed graph,
and about complexity aspects. 
We also prove that several results on unsigned graphs do not generalise to our context, at least not in an obvious and natural way.

This paper is organised as follows.
In Section~2, we make more formal some of the notions and definitions mentioned earlier,
and we state first properties and results on our notion of achromatic number of signed graphs.
In Section~3, we investigate the consequences on our parameter of modifying the structure of a given signed graph.
Section~4 is dedicated to complexity aspects.
In Sections~3 and~4, the results we exhibit show that our variant of the achromatic number and the original one behave in a very close way,
although there are some differences.
In Section~5, we point out results on the original achromatic number that do not extend, or only very partially, to our context.
We finish off in Section~6 with some directions for further work on the topic.

%%%%%%%%%%%%%%%%%%%%%%%%%%%%%%%%%%%%%%%%%%%%%%%%%%%%%%
%%%%%%%%%%%%%%%%%%%%%%%%%%%%%%%%%%%%%%%%%%%%%%%%%%%%%%
%%%%%%%%%%%%%%%%%%%%%%%%%%%%%%%%%%%%%%%%%%%%%%%%%%%%%%
%%%%%%%%%%%%%%%%%%%%%%%%%%%%%%%%%%%%%%%%%%%%%%%%%%%%%%
%%%%%%%%%%%%%%%%%%%%%%%%%%%%%%%%%%%%%%%%%%%%%%%%%%%%%%
%%%%%%%%%%%%%%%%%%%%%%%%%%%%%%%%%%%%%%%%%%%%%%%%%%%%%%
%%%%%%%%%%%%%%%%%%%%%%%%%%%%%%%%%%%%%%%%%%%%%%%%%%%%%%
%%%%%%%%%%%%%%%%%%%%%%%%%%%%%%%%%%%%%%%%%%%%%%%%%%%%%%
%%%%%%%%%%%%%%%%%%%%%%%%%%%%%%%%%%%%%%%%%%%%%%%%%%%%%%
%%%%%%%%%%%%%%%%%%%%%%%%%%%%%%%%%%%%%%%%%%%%%%%%%%%%%%
%%%%%%%%%%%%%%%%%%%%%%%%%%%%%%%%%%%%%%%%%%%%%%%%%%%%%%
%%%%%%%%%%%%%%%%%%%%%%%%%%%%%%%%%%%%%%%%%%%%%%%%%%%%%%
%%%%%%%%%%%%%%%%%%%%%%%%%%%%%%%%%%%%%%%%%%%%%%%%%%%%%%
%%%%%%%%%%%%%%%%%%%%%%%%%%%%%%%%%%%%%%%%%%%%%%%%%%%%%%

\section{Preliminaries} \label{section:preliminaries}

This preliminary section is organised as follows.
In Subsection~\ref{subsection:formal-definitions}, we start by recalling how Zaslavsky  defined proper colourings of signed graphs,
before defining formally the notion of complete colourings of signed graphs that we investigate throughout this paper.
Along the way, we also recall and point out some useful properties of both types of colourings.
In Subsection~\ref{subsection:first-results-signed}, we then prove first elementary results and properties of our variant of the achromatic number.
Lastly, we study, in Subsections~\ref{subsection:achromatic-cliques} and~\ref{subsection:achromatic-paths-cycles}, our notion of achromatic number for particular signed graphs, 
namely signed complete graphs, paths, and cycles. We determine its exact value in some cases.

\subsection{From proper colourings to complete colourings of signed graphs} \label{subsection:formal-definitions}

\subsubsection*{Proper colourings and chromatic number of signed graphs}

First things first, we need to make formal all notions that revolve around Zaslavsky's proper colourings of signed graphs~\cite{Zas82}.
Let $(G,\sigma)$ be a signed graph.
A \textit{proper colouring} $\phi: V(G) \rightarrow \mathbb{Z}$ of $(G,\sigma)$ is an assignment of colours to the vertices such that,
for every edge $uv$ of $G$, we have $\phi(u) \neq \sigma(uv)\phi(v)$ (where, naturally, $\sigma(uv)\phi(v)=\phi(v)$ if $\sigma(uv)=+$, and $\sigma(uv)\phi(v)=-\phi(v)$ otherwise).
As mentioned earlier in the introductory section, note that this definition meets several properties.
First, note that, by a proper colouring of a signed graph, any two adjacent vertices joined by a positive edge must not be assigned the same colour,
while any two adjacent vertices joined by a negative edge must not be assigned opposite colours.
This is reminiscent of the main property of proper colourings of graphs, which is that adjacent vertices must get assigned distinct colours.
Second, note, by carefully checking the distinction property, that proper colourings of signed graphs are preserved under switching vertices.
That is, if switching a vertex $v$ in a signed graph $(G,\sigma)$ results in a different signed graph $(G,\sigma')$ defined over $G$,
then, for every proper colouring $\phi$ of $(G,\sigma)$, it can be checked from the definitions that the colouring $\phi'$ of $(G,\sigma')$, 
defined as $\phi'(v)=-\phi(v)$ and $\phi'(u)=\phi(u)$ for every $u \in V(G) \setminus \{v\}$, is also proper.
This is a property that will be used throughout this work;
thus, whenever switching vertices in a signed graph that is already $k$-coloured,
implicitly it should be understood that their colours are negated.

\begin{figure}[!t]
 	\centering
 	
 	\subfloat[$K_1^*$]{
    \scalebox{0.8}{
	\begin{tikzpicture}[inner sep=0.7mm]	
	
	\node[draw, circle, black, line width=1pt](v1) at (0,0.7)[]{$0$};
	\node[draw, white](c1) at (-0.5,0)[]{};
	\node[draw, white](c2) at (0.5,0)[]{};

	\end{tikzpicture}
    }
    }
    \hspace{20pt} 	 	
 	\subfloat[$K_2^*$]{
    \scalebox{0.8}{
	\begin{tikzpicture}[inner sep=0.7mm]	
	
	\node[draw, circle, black, line width=1pt](v1) at (0,0)[]{$1$};
     		
    \draw [-, line width=1.5pt, Red, densely dashdotted] (v1) to[out=150,in=210,looseness=12] (v1);
	\end{tikzpicture}
    }
    }
    \hspace{20pt}
 	\subfloat[$K_3^*$]{
    \scalebox{0.8}{
	\begin{tikzpicture}[inner sep=0.7mm]	
	
	\node[draw, circle, black, line width=1pt](v1) at (0,0)[]{$1$};
	\node[draw, circle, black, line width=1pt](v0) at (3,0)[]{$0$};
     		
    \draw [-, line width=1.5pt, Blue] (v1) to[out=0,in=180,bend left=20] (v0);
    \draw [-, line width=1.5pt, Red, densely dashdotted] (v1) to[out=0,in=180,bend right=20] (v0);
    \draw [-, line width=1.5pt, Red, densely dashdotted] (v1) to[out=150,in=210,looseness=12] (v1);
	\end{tikzpicture}
    }
    }
    \hspace{20pt}
 	\subfloat[$K_4^*$]{
    \scalebox{0.8}{
	\begin{tikzpicture}[inner sep=0.7mm]	
	
	\node[draw, circle, black, line width=1pt](v1) at (0,0)[]{$1$};
	\node[draw, circle, black, line width=1pt](v2) at (3,0)[]{$2$};
     		
    \draw [-, line width=1.5pt, Blue] (v1) to[out=0,in=180,bend left=20] (v2);
    \draw [-, line width=1.5pt, Red, densely dashdotted] (v1) to[out=0,in=180,bend right=20] (v2);
    \draw [-, line width=1.5pt, Red, densely dashdotted] (v1) to[out=150,in=210,looseness=12] (v1);
    \draw [-, line width=1.5pt, Red, densely dashdotted] (v2) to[out=30,in=-30,looseness=12] (v2);
	\end{tikzpicture}
    }
    }
    \hspace{20pt}
 	\subfloat[$K_5^*$]{
    \scalebox{0.8}{
	\begin{tikzpicture}[inner sep=0.7mm]	
	\node[draw, circle, black, line width=1pt](v0) at (1.5,2)[]{$0$};
	\node[draw, circle, black, line width=1pt](v1) at (0,0)[]{$1$};
	\node[draw, circle, black, line width=1pt](v2) at (3,0)[]{$2$};
     		
    \draw [-, line width=1.5pt, Blue] (v1) to[out=0,in=180,bend left=20] (v2);
    \draw [-, line width=1.5pt, Red, densely dashdotted] (v1) to[out=0,in=180,bend right=20] (v2);
    \draw [-, line width=1.5pt, Red, densely dashdotted] (v1) to[out=150,in=210,looseness=12] (v1);
    \draw [-, line width=1.5pt, Red, densely dashdotted] (v2) to[out=30,in=-30,looseness=12] (v2);
    
    \draw [-, line width=1.5pt, Red, densely dashdotted] (v1) to[out=90,in=180,bend left=20] (v0);
    \draw [-, line width=1.5pt, Blue] (v1) to[out=90,in=180,bend right=20] (v0);
    \draw [-, line width=1.5pt, Red, densely dashdotted] (v2) to[out=90,in=180,bend right=20] (v0);
    \draw [-, line width=1.5pt, Blue] (v2) to[out=90,in=180,bend left=20] (v0);
	\end{tikzpicture}
    }
    }
    \hspace{20pt}
 	\subfloat[$K_6^*$]{
    \scalebox{0.8}{
	\begin{tikzpicture}[inner sep=0.7mm]	
	\node[draw, circle, black, line width=1pt](v0) at (1.5,2)[]{$3$};
	\node[draw, circle, black, line width=1pt](v1) at (0,0)[]{$1$};
	\node[draw, circle, black, line width=1pt](v2) at (3,0)[]{$2$};
     		
    \draw [-, line width=1.5pt, Blue] (v1) to[out=0,in=180,bend left=20] (v2);
    \draw [-, line width=1.5pt, Red, densely dashdotted] (v1) to[out=0,in=180,bend right=20] (v2);
    \draw [-, line width=1.5pt, Red, densely dashdotted] (v1) to[out=150,in=210,looseness=12] (v1);
    \draw [-, line width=1.5pt, Red, densely dashdotted] (v2) to[out=30,in=-30,looseness=12] (v2);
    \draw [-, line width=1.5pt, Red, densely dashdotted] (v0) to[out=120,in=60,looseness=12] (v0);
    
    \draw [-, line width=1.5pt, Red, densely dashdotted] (v1) to[out=90,in=180,bend left=20] (v0);
    \draw [-, line width=1.5pt, Blue] (v1) to[out=90,in=180,bend right=20] (v0);
    \draw [-, line width=1.5pt, Red, densely dashdotted] (v2) to[out=90,in=180,bend right=20] (v0);
    \draw [-, line width=1.5pt, Blue] (v2) to[out=90,in=180,bend left=20] (v0);
	\end{tikzpicture}
    }
    }

\caption{Examples of signed complete multigraphs $K_k^*$.
Dashed red edges are negative edges, while solid blue edges are positive edges. 
\label{figure:kns}}
\end{figure}
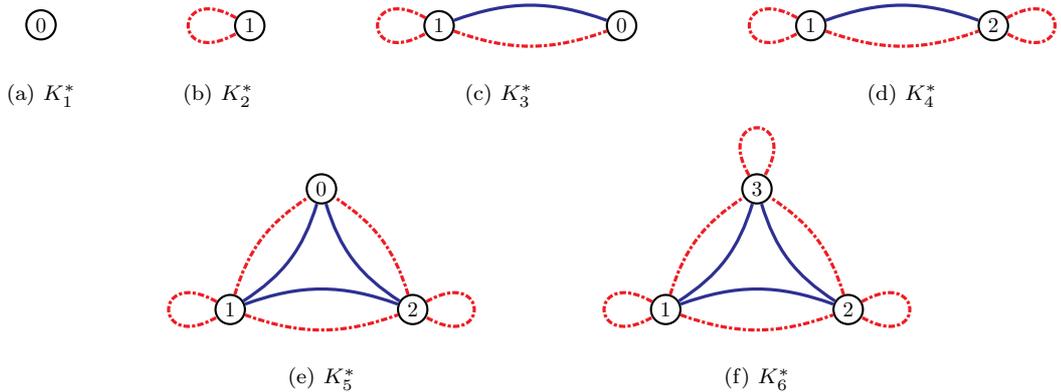

With the notion of proper colouring of signed graphs above in hand, 
the next step is to wonder about a natural way to define a dedicated parameter corresponding to the usual chromatic number of graphs.
From here on, we stick to some terminology introduced by M\'a\v{c}ajov\'a, Raspaud and \v{S}koviera,
which, for technicalities described in~\cite{MRS16} (dealing, notably, with the oddities that colour $\pm 0$ brings), 
fits better to the unsigned context than the terminology originally introduced by Zaslavsky in~\cite{Zas82}.
For every $k \geq 1$, we define a set $M_k \subset \mathbb{Z}$ being $M_k=\{-n, -(n-1), \dots, -1, +1, \dots, +(n-1), +n\}$ if $k=2n$,
and $M_k=\{-n, -(n-1), \dots, -1, \pm 0, +1, \dots, +(n-1), +n\}$ if $k=2n+1$ (where, note, $-0$ and $+0$ are regarded as a single colour $\pm 0$),
so that the set $M_k$ has size exactly $k$.
From here on, a \textit{$k$-colouring} of a signed graph will refer to a colouring assigning colours from $M_k$.
We finally get to generalising the chromatic number to the signed context,
the \textit{chromatic number} $\chi(G,\sigma)$ of a signed graph $(G,\sigma)$ being the smallest $k \geq 1$ such that $(G,\sigma)$ admits proper $k$-colourings.
As described earlier, it is important to emphasise that, although a signed graph $(G,\sigma)$ might have many \textit{equivalent signed graphs} (i.e., other signed graphs defined over $G$ reached through switching vertices in $(G,\sigma)$),
they are actually all equivalent in terms of chromatic number, as a proper $k$-colouring of any of them yields a proper $k$-colouring of any other equivalent signed graph.
Thus, in some sense, this parameter is actually defined for equivalence classes of signed graphs, rather than for individual signed graphs only.

Similarly to Guenin's way of defining proper colourings of signed graphs, 
the definitions above can also be seen through the lens of homomorphisms between signed graphs.
Particularly, a proper $k$-colouring of a signed graph $(G,\sigma)$ can be seen as a homomorphism of $(G,\sigma)$ to $K_k^*$,
i.e., preserving both edges and edge signs,
where $K_k^*$ is the signed complete multigraph where (see Figure~\ref{figure:kns} for an illustration):

\begin{itemize}
	\item $V(K_k^*)=\{i : +i \in M_k {\rm ~and~} i \geq 0\}$, 
	\item every two distinct vertices $i$ and $j$ are joined by a negative edge and a positive edge,
	\item there is a negative loop at every vertex $i \neq 0$.
\end{itemize}

\noindent Note that $K_k^*$ has essentially two possible structures depending on the parity of $k$,
which illustrates one subtlety of proper colourings of signed graphs,
being that the number of colours itself has a central place.
Nevertheless, from this, we can also define $\chi(G,\sigma)$ as the smallest $k \geq 1$ such that $(G,\sigma)$ admits a homomorphism to $K_k^*$.

\subsubsection*{Complete colourings and achromatic number of signed graphs}

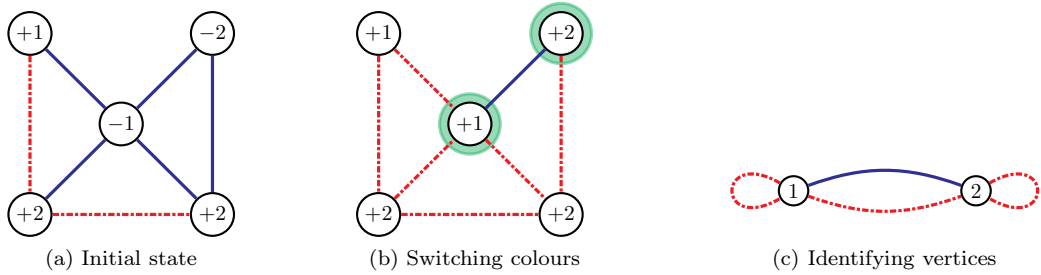
\begin{figure}[!t]
 	\centering
 	\subfloat[Initial state]{
    \scalebox{0.8}{
	\begin{tikzpicture}[inner sep=0.7mm]	
	
	\node[draw, circle, black, line width=1pt](v1) at (0,0)[]{$+1$};
	\node[draw, circle, black, line width=1pt](v2) at (0,-3)[]{$+2$};
	\node[draw, circle, black, line width=1pt](v3) at (1.5,-1.5)[]{$-1$};
	\node[draw, circle, black, line width=1pt](v4) at (3,0)[]{$-2$};
	\node[draw, circle, black, line width=1pt](v5) at (3,-3)[]{$+2$};
	
    \draw [-, line width=1.5pt, Red, densely dashdotted] (v1) -- (v2);
    \draw [-, line width=1.5pt, Red, densely dashdotted] (v2) -- (v5);
    \draw [-, line width=1.5pt, Blue] (v4) -- (v5);
    \draw [-, line width=1.5pt, Blue] (v1) -- (v3);
    \draw [-, line width=1.5pt, Blue] (v2) -- (v3);
    \draw [-, line width=1.5pt, Blue] (v3) -- (v4);
    \draw [-, line width=1.5pt, Blue] (v3) -- (v5);
	\end{tikzpicture}
    }
    }
    \hspace{30pt}
 	\subfloat[Switching colours]{
    \scalebox{0.8}{
	\begin{tikzpicture}[inner sep=0.7mm]	
	
	\draw[Green, line width=1.5pt,opacity=0.4, fill] (v4) circle (0.5cm);
	\draw[Green, line width=1.5pt,opacity=0.4, fill] (v3) circle (0.5cm);
	
	\node[draw, circle, black, line width=1pt](v1) at (0,0)[]{$+1$};
	\node[draw, circle, black, line width=1pt](v2) at (0,-3)[]{$+2$};
	\node[draw, circle, black, line width=1pt,fill=white](v3) at (1.5,-1.5)[]{$+1$};
	\node[draw, circle, black, line width=1pt,fill=white](v4) at (3,0)[]{$+2$};
	\node[draw, circle, black, line width=1pt](v5) at (3,-3)[]{$+2$};
	
    \draw [-, line width=1.5pt, Red, densely dashdotted] (v1) -- (v2);
    \draw [-, line width=1.5pt, Red, densely dashdotted] (v2) -- (v5);
    \draw [-, line width=1.5pt, Red, densely dashdotted] (v4) -- (v5);
    \draw [-, line width=1.5pt, Red, densely dashdotted] (v1) -- (v3);
    \draw [-, line width=1.5pt, Red, densely dashdotted] (v2) -- (v3);
    \draw [-, line width=1.5pt, Blue] (v3) -- (v4);
    \draw [-, line width=1.5pt, Red, densely dashdotted] (v3) -- (v5);
	\end{tikzpicture}
    }
    }
    \hspace{30pt}
 	\subfloat[Identifying vertices]{
    \scalebox{0.8}{
	\begin{tikzpicture}[inner sep=0.7mm]	
	\node[draw, circle, black, line width=1pt](v1) at (0,0)[]{$1$};
	\node[draw, circle, black, line width=1pt](v2) at (3,0)[]{$2$};
     		
    \draw [-, line width=1.5pt, Blue] (v1) to[out=0,in=180,bend left=20] (v2);
    \draw [-, line width=1.5pt, Red, densely dashdotted] (v1) to[out=0,in=180,bend right=20] (v2);
    \draw [-, line width=1.5pt, Red, densely dashdotted] (v1) to[out=150,in=210,looseness=12] (v1);
    \draw [-, line width=1.5pt, Red, densely dashdotted] (v2) to[out=30,in=-30,looseness=12] (v2);
	\end{tikzpicture}
    }
    }

\caption{Illustration of the reduction process. Given a signed graph with a given proper $4$-colouring (a), 
we start by switching all vertices (highlighted in green) with negative colours (b), before identifying vertices with the same colour and removing parallel edges (c). 
The reduced signed graph in (c) being $K_4^*$, the proper $4$-colouring in (a) is thus also complete.
Dashed red edges are negative edges, while solid blue edges are positive edges. 
\label{figure:reduced}}
\end{figure}

We now get to defining the notion of complete colouring (and of achromatic number) of signed graphs that we explore throughout this work.
As will be seen later on, the upcoming definitions admit other interpretations, which, depending on the context, 
might be more or less convenient to adopt.
For now, we need an auxiliary notion, illustrated in Figure~\ref{figure:reduced}.
Let $(G,\sigma)$ be a signed graph, and $\phi$ be a $k$-colouring of $(G,\sigma)$.
By \textit{reducing} $(G,\sigma)$ (according to $\phi$), we mean applying the following successive operations to $(G,\sigma)$:

\begin{enumerate}
	\item switching all vertices assigned negative colours by $\phi$;
	\item identifying all vertices assigned a given colour $+i$ by $\phi$;
	\item for every two vertices joined by an edge (loop or not), keeping at most one positive edge and at most one negative edge joining them;
\end{enumerate}

\noindent which result in a \textit{reduced signed graph} $R(G,\sigma,\phi)$.
Note that the vertices of $R(G,\sigma,\phi)$ get assigned non-negative colours of $M_k$ inherited from $\phi$;
for convenience, we also denote by $\phi$ the resulting $k$-colouring of $R(G,\sigma,\phi)$.
Similarly as in the unsigned context, in which applying similar operations from a complete colouring must end up in a complete graph,
we here say that $\phi$ is \textit{complete} if $R(G,\sigma,\phi)$ is $K_k^*$.
Note that, actually, this is similar to requiring that, after having applied the first operation above, the resulting signed graph admits an edge-surjective homomorphism to $K_k^*$,
which, in some sense, means we want all possible valid configurations of two colours and an edge sign by a proper $k$-colouring to appear along at least one edge each.
To be clear, the third operation above means making the resulting signed graph simple,
that is, having no two loops with the same sign at any vertex, neither two edges with the same sign linking any two vertices.

\begin{figure}[!t]
 	\centering
 	\subfloat[$(G,\sigma)$]{
    \scalebox{0.8}{
	\begin{tikzpicture}[inner sep=0.7mm]	
	
	\node[draw, circle, black, line width=1pt](v1) at (0,0)[]{$+1$};
	\node[draw, circle, black, line width=1pt](v2) at (0,-2)[]{$+1$};
	\node[draw, circle, black, line width=1pt](v3) at (-2,-4)[]{$\pm 0$};
	\node[draw, circle, black, line width=1pt](v4) at (2,-4)[]{$\pm 0$};
	
    \draw [-, line width=1.5pt, Red, densely dashdotted] (v1) -- (v2);
    \draw [-, line width=1.5pt, Red, densely dashdotted] (v2) -- (v3);
    \draw [-, line width=1.5pt, Blue] (v2) -- (v4);
	\end{tikzpicture}
    }
    }
    \hspace{20pt}
 	\subfloat[$(G,\sigma')$]{
    \scalebox{0.8}{
	\begin{tikzpicture}[inner sep=0.7mm]	
	
	\draw[Green, line width=1.5pt,opacity=0.4, fill] (v1) circle (0.5cm);	
	
	\node[draw, circle, black, line width=1pt,fill=white](v1) at (0,0)[]{$-1$};
	\node[draw, circle, black, line width=1pt](v2) at (0,-2)[]{$+1$};
	\node[draw, circle, black, line width=1pt](v3) at (-2,-4)[]{$\pm 0$};
	\node[draw, circle, black, line width=1pt](v4) at (2,-4)[]{$\pm 0$};
	
    \draw [-, line width=1.5pt, Blue] (v1) -- (v2);
    \draw [-, line width=1.5pt, Red, densely dashdotted] (v2) -- (v3);
    \draw [-, line width=1.5pt, Blue] (v2) -- (v4);
	\end{tikzpicture}
    }
    }
    \hspace{20pt}
 	\subfloat[$(G,\sigma'')$]{
    \scalebox{0.8}{
	\begin{tikzpicture}[inner sep=0.7mm]	
	
	\draw[Green, line width=1.5pt,opacity=0.4, fill] (v4) circle (0.5cm);	
	
	\node[draw, circle, black, line width=1pt](v1) at (0,0)[]{$-1$};
	\node[draw, circle, black, line width=1pt](v2) at (0,-2)[]{$+1$};
	\node[draw, circle, black, line width=1pt](v3) at (-2,-4)[]{$\pm 0$};
	\node[draw, circle, black, line width=1pt,fill=white](v4) at (2,-4)[]{$\pm 0$};
	
    \draw [-, line width=1.5pt, Blue] (v1) -- (v2);
    \draw [-, line width=1.5pt, Red, densely dashdotted] (v2) -- (v3);
    \draw [-, line width=1.5pt, Red, densely dashdotted] (v2) -- (v4);
	\end{tikzpicture}
    }
    }

\caption{Illustration of the fact that complete colourings are not necessarily preserved under switching vertices.
(a) depicts a signed graph $(G,\sigma)$ together with a complete $3$-colouring.
(b) depicts another signed graph $(G,\sigma')$ obtained from $(G,\sigma)$ by switching the vertex highlighted in green; it can be checked that the resulting $3$-colouring remains complete.
(c) depicts a third signed graph $(G,\sigma'')$ obtained from $(G,\sigma')$ by switching the vertex highlighted in green;
this time, the resulting $3$-colouring is not complete since the reduced graph has no positive edge joining its vertices $0$ and~$1$.
Dashed red edges are negative edges, while solid blue edges are positive edges. 
\label{figure:preservation}}
\end{figure}
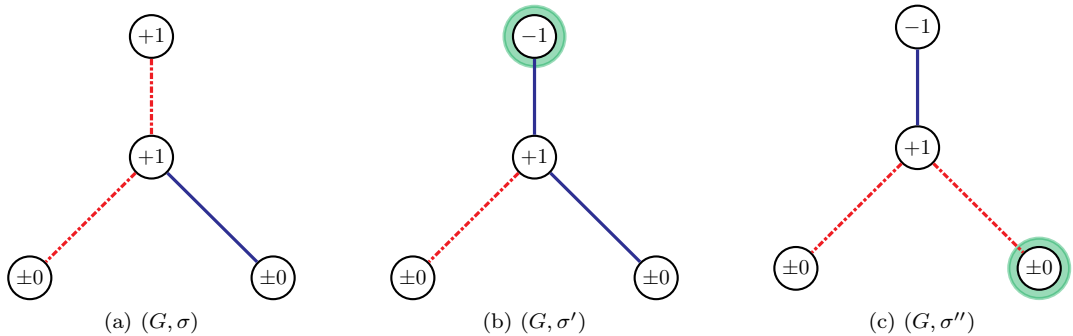

It is worth pointing out that complete colourings of signed graphs, as defined above, are quite different from proper colourings of signed graphs,
in that they might be not preserved by the switching operation (see Figure~\ref{figure:preservation} for an illustration of the following arguments).
More precisely, from a complete colouring $\phi$ of a signed graph $(G,\sigma)$, by switching vertices that have been assigned colour $\pm 0$ by $\phi$, we end up with an equivalent signed graph $(G,\sigma')$ and a proper colouring $\phi'$ that does not have to be complete. 
To make this point clear, it is worth emphasising that only switching vertices assigned colour $\pm 0$ can make a complete colouring lose its completeness.
A consequence is that complete colourings are preserved under switching vertices being assigned any other colour (essentially because this does not alter the resulting edges in the reduced graph; this will be made clearer through later Observation~\ref{observation:types-edges}), 
and, more generally speaking, complete $k$-colourings with $k$ even are consistent under switching vertices.
Nevertheless, in general all equivalent signed graphs defined over a given graph do not have to be all equivalent in terms of complete colourings.
Through the results we will establish in later sections (those in Subsection~\ref{subsection:achromatic-cliques} and Section~\ref{section:operations} being particularly significant),
we will establish that signed graphs that are close in terms of signatures
can actually be quite different in terms of achromatic number.

To take those concerns into consideration, we choose to define the \textit{achromatic number} of a signed graph $(G,\sigma)$, denoted by $\psi(G,\sigma)$, 
as the largest $k \geq 1$ such that there exists a signed graph $(G,\sigma')$ equivalent to $(G,\sigma)$ that admits complete $k$-colourings.
In other words, starting from $(G,\sigma)$, our goal is to switch vertices to get a signed graph that maps in an edge-surjective way to the largest $K_k^*$ possible (following the concepts above).
In cases where $H$ refers to the signed graph $(G,\sigma)$ (i.e., the signature of $H$ is implicit),
abusing the notation we will simply denote by $\psi(H)$ its achromatic number.
Note that a complete colouring is by definition proper, and, thus, we always have $\chi(G, \si) \le \psi(G, \si)$.

\subsubsection*{Some remarks and tools}

Let us now give a bit more insight into complete colourings and the achromatic number (as introduced above for signed graphs)
by introducing some more working terminology and pointing out some tools and results to be used in the main parts of this paper.

To make sure a given colouring of a signed graph is complete, the existence of edges with certain signs and colours must be verified.
To ease this part, we will make use of the following terminology.
Given a signed graph $(G,\sigma)$ and a $k$-colouring $\phi$ of $(G,\sigma)$,
for any two (possibly equal) colours $i,  j \in M_k$ and any edge sign $s \in \{-,+\}$,
we say that an edge $uv$ of $(G,\sigma)$ is of \textit{type $(i,j,s)$} if $\{\phi(u),\phi(v)\}=\{i,j\}$
and $\sigma(uv)=s$.
Assume now $i,j \in V(K_k^*)$.
If $i \neq j$, then, by a \textit{p-edge of type $(i,j)$} of $(G,\sigma)$, we refer to an edge of $(G,\sigma)$ that, essentially, gets mapped, through the reduction process,
to the positive edge joining $i$ and $j$ in $K_k^*$ (note that this edge exists in $K_k^*$ by definition).
Analogously, if $\{i,j\}\neq\{0\}$, then we define an \textit{n-edge of type $(i,j)$} of $(G,\sigma)$ as an edge of $(G,\sigma)$ that gets mapped to the negative edge joining $i$ and $j$ in $K_k^*$ (again, that edge exists).
Note here that if $\phi$ is a proper $k$-colouring of $(G,\sigma)$,
then $(G,\sigma)$ cannot contain what would be a ``p-edge of type $(i,i)$'', 
nor a ``p-edge or n-edge of type $(0,0)$''.

Thus, if $\phi$ is complete, then, in $(G,\sigma)$, there must be several types of p-edges and n-edges,
namely (1) a p-edge and an n-edge of type $(i,j)$ for every two distinct $i,j \in V(K_k^*)$, and (2) an n-edge of type $(i,i)$ for every strictly positive $i \in V(K_k^*)$.
By carefully checking the different types of edges of $(G,\sigma)$ giving a given edge type in the reduced signed graph $R(G,\sigma,\phi)$,
the following classification arises:

\begin{observation}\label{observation:types-edges}
Let $(G,\sigma)$ be a signed graph, and $\phi$ be a 
$k$-colouring of $(G,\sigma)$.
Assume $R(G,\sigma,\phi)$ has an edge of type $(i,j,s)$; then:
\begin{itemize}
	\item if $i \neq j$ and $s=+$, then $(G,\sigma)$ has an edge of type $(i,j,+)$, $(-i,-j, +)$, $(-i,j,-)$ or $(i,-j,-)$;
	\item if $i \neq j$ and $s=-$, then $(G,\sigma)$ has an edge of type $(-i,j,+)$, $(i,-j,+)$, $(i,j,-)$ or $(-i,-j,-)$;
	\item if $i=j$, $i \neq 0$ and $s=-$, then $(G,\sigma)$ has an edge of type $(i,-i,+)$, $(i,i,-)$ or $(-i,-i,-)$.
\end{itemize}
\end{observation}

Through playing with types of edges as just introduced,
it can be checked that the following two claims,
on ways to obtain complete colourings from existing ones,
hold easily.

\begin{observation}\label{observation:switch-negated-colours}
Let $(G,\sigma)$ be a signed graph, and $\phi$ be a complete $k$-colouring of $(G,\sigma)$.
By swapping any two opposite colours $+i$ and $-i$,
i.e., changing to $-i$ the colour by $\phi$ of every vertex assigned colour $+i$,
and \textit{vice versa},
we get a complete $k$-colouring of $(G,\sigma)$.
\end{observation}

\begin{observation}\label{observation:remove-class}
Let $(G,\sigma)$ be a signed graph, and $\phi$ be a complete $k$-colouring of $(G,\sigma)$.
For every $i \in M_k$, removing, from $(G,\sigma)$, all vertices assigned colour $\pm i$ by $\phi$,
results in a signed graph $(G',\sigma)$ in which the restriction of $\phi$ is a complete $(k-2)$-colouring (if $i \neq 0$) or a complete $(k-1)$-colouring (otherwise).
\end{observation}

At first glance, the switching operation might seem a bit hard to comprehend.
Particularly, the unfamiliar reader might think that, given two signed graphs $(G,\sigma)$ and $(G,\sigma')$ defined over a same graph $G$,
determining whether $(G,\sigma)$ and $(G,\sigma')$ are equivalent is not obvious.
Fortunately, this is something that can be decided in quadratic time (see~\cite{RST21}),
due to the following notions that lie behind an invariant shared by any two equivalent signed graphs.
Recall that, in a graph, a \textit{closed walk} $W$ is a cycle in which vertices and edges can be repeated.
Assuming the graph is signed, the \textit{balance} of $W$ refers to the parity of the number of negative edges it traverses:
The balance of $W$ is either \textit{odd} if it traverses an odd number of negative edges, or \textit{even} otherwise,
i.e., if it traverses an even number of negative edges. If $W$ has an odd balance, then the closed walk is said to be \textit{negative}, while it is said to be \textit{positive} otherwise.
The important property is that, in signed graphs, 
the balance of closed walks is preserved upon switching vertices.

\begin{lemma}[Zaslavsky~\cite{Zas82b}]\label{lemma:switching-balance}
Two signed graphs $(G,\sigma)$ and $(G,\sigma')$ are equivalent if and only if the sets of positive (and, thus, negative) cycles
of $(G,\sigma)$ and $(G,\sigma')$ are the same.
\end{lemma}

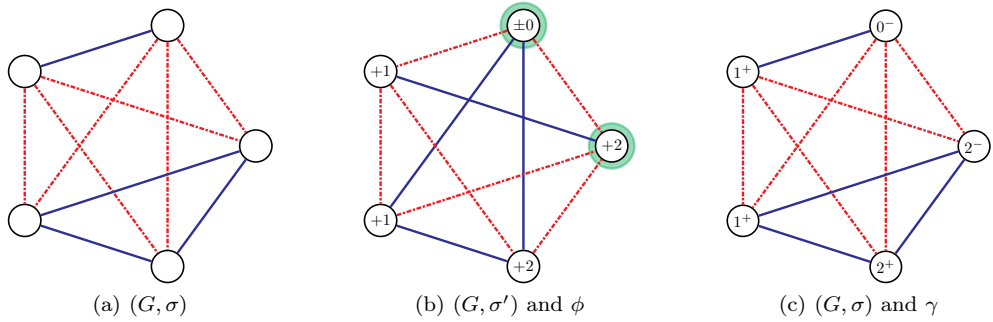
\begin{figure}[!t]
 	\centering
 	\subfloat[$(G,\sigma)$]{
    \scalebox{0.6}{
	\begin{tikzpicture}[inner sep=0.7mm]	
		\node[draw, circle, black, line width=1pt](v1) at (1*360/5:2.8cm){\textcolor{white}{+1}};	
		\node[draw, circle, black, line width=1pt](v2) at (2*360/5:2.8cm){\textcolor{white}{+1}};	
		\node[draw, circle, black, line width=1pt](v3) at (3*360/5:2.8cm){\textcolor{white}{+1}};	
		\node[draw, circle, black, line width=1pt](v4) at (4*360/5:2.8cm){\textcolor{white}{+1}};	
		\node[draw, circle, black, line width=1pt](v5) at (5*360/5:2.8cm){\textcolor{white}{+1}};	
		
    	\draw [-, line width=1.5pt, Blue] (v1) -- (v2);
    	\draw [-, line width=1.5pt, Red, densely dashdotted] (v1) -- (v3);
    	\draw [-, line width=1.5pt, Red, densely dashdotted] (v1) -- (v4);
    	\draw [-, line width=1.5pt, Red, densely dashdotted] (v1) -- (v5);
    	\draw [-, line width=1.5pt, Red, densely dashdotted] (v2) -- (v3);
    	\draw [-, line width=1.5pt, Red, densely dashdotted] (v2) -- (v4);
    	\draw [-, line width=1.5pt, Red, densely dashdotted] (v2) -- (v5);
    	\draw [-, line width=1.5pt, Blue] (v3) -- (v4);
    	\draw [-, line width=1.5pt, Blue] (v3) -- (v5);
    	\draw [-, line width=1.5pt, Blue] (v4) -- (v5);
	\end{tikzpicture}
    }
    }
    \hspace{20pt}
 	\subfloat[$(G,\sigma')$ and $\phi$]{
    \scalebox{0.6}{
	\begin{tikzpicture}[inner sep=0.7mm]	
		\draw[Green, line width=1.5pt,opacity=0.4, fill] (v1) circle (0.5cm);
		\draw[Green, line width=1.5pt,opacity=0.4, fill] (v5) circle (0.5cm);
		
		\node[draw, circle, black, line width=1pt,fill=white](v1) at (1*360/5:2.8cm){$\pm 0$};	
		\node[draw, circle, black, line width=1pt,fill=white](v2) at (2*360/5:2.8cm){$+1$};	
		\node[draw, circle, black, line width=1pt,fill=white](v3) at (3*360/5:2.8cm){$+1$};	
		\node[draw, circle, black, line width=1pt](v4) at (4*360/5:2.8cm){$+2$};	
		\node[draw, circle, black, line width=1pt,fill=white](v5) at (5*360/5:2.8cm){$+2$};	
		
    	\draw [-, line width=1.5pt, Red, densely dashdotted] (v1) -- (v2);
    	\draw [-, line width=1.5pt, Blue] (v1) -- (v3);
    	\draw [-, line width=1.5pt, Blue] (v1) -- (v4);
    	\draw [-, line width=1.5pt, Red, densely dashdotted] (v1) -- (v5);
    	\draw [-, line width=1.5pt, Red, densely dashdotted] (v2) -- (v3);
    	\draw [-, line width=1.5pt, Red, densely dashdotted] (v2) -- (v4);
    	\draw [-, line width=1.5pt, Blue] (v2) -- (v5);
    	\draw [-, line width=1.5pt, Blue] (v3) -- (v4);
    	\draw [-, line width=1.5pt, Red, densely dashdotted] (v4) -- (v5);
    	\draw [-, line width=1.5pt, Red, densely dashdotted] (v3) -- (v5);
	\end{tikzpicture}
    }
    }
    \hspace{20pt}
 	\subfloat[$(G,\sigma)$ and $\gamma$]{
    \scalebox{0.6}{
	\begin{tikzpicture}[inner sep=0.7mm]	
		\node[draw, circle, black, line width=1pt,fill=white](v1) at (1*360/5:2.8cm){$0^-$};	
		\node[draw, circle, black, line width=1pt,fill=white](v2) at (2*360/5:2.8cm){$1^+$};	
		\node[draw, circle, black, line width=1pt,fill=white](v3) at (3*360/5:2.8cm){$1^+$};	
		\node[draw, circle, black, line width=1pt](v4) at (4*360/5:2.8cm){$2^+$};	
		\node[draw, circle, black, line width=1pt](v5) at (5*360/5:2.8cm){$2^-$};	
		
    	\draw [-, line width=1.5pt, Blue] (v1) -- (v2);
    	\draw [-, line width=1.5pt, Red, densely dashdotted] (v1) -- (v3);
    	\draw [-, line width=1.5pt, Red, densely dashdotted] (v1) -- (v4);
    	\draw [-, line width=1.5pt, Red, densely dashdotted] (v1) -- (v5);
    	\draw [-, line width=1.5pt, Red, densely dashdotted] (v2) -- (v3);
    	\draw [-, line width=1.5pt, Red, densely dashdotted] (v2) -- (v4);
    	\draw [-, line width=1.5pt, Red, densely dashdotted] (v2) -- (v5);
    	\draw [-, line width=1.5pt, Blue] (v3) -- (v4);
    	\draw [-, line width=1.5pt, Blue] (v3) -- (v5);
    	\draw [-, line width=1.5pt, Blue] (v4) -- (v5);
	\end{tikzpicture}
    }
    }

\caption{Deducing an inferred complete $k$-colouring $\gamma$ of a signed graph $(G,\sigma)$ from a complete $k$-colouring $\phi$.
(a) depicts $(G,\sigma)$. (b) depicts an equivalent signed graph $(G,\sigma')$ obtained by switching two vertices (highlighted in green) of $(G,\sigma)$,
together with a complete $5$-colouring $\phi$. (c) shows the inferred complete $5$-colouring $\gamma$ of $(G,\sigma)$ that can be obtained from $\phi$.
Dashed red edges are negative edges, while solid blue edges are positive edges. 
\label{figure:inferred}}
\end{figure}

Regarding our notion of achromatic number of signed graphs,
it is more convenient, in some contexts we will experience through some of our proofs,
to manipulate complete $k$-colourings using a slightly different terminology.
This is based on the following observations (illustrated in Figure~\ref{figure:inferred}).
Let $(G,\sigma)$ be a signed graph,
and assume the vertices of $(G,\sigma)$ can be switched to reach an equivalent signed graph $(G,\sigma')$ admitting a complete $k$-colouring $\phi$.
To avoid having to deal with the switching operation, one can infer colours in $(G,\sigma)$ from $\phi$, taking into account that some vertices must be switched for $\phi$ to be complete.
Precisely, we denote by $M_k'$ the set of \textit{signed colours} being $\{n^-, (n-1)^-, \dots, 1^-,1^+, \dots, (n-1)^+,n^+\}$ if $k=2n$ is even,
and $\{n^-, (n-1)^-, \dots, 1^-, 0^-, 0^+, 1^+, \dots, (n-1)^+,n^+\}$ if $k=2n+1$ is odd.
The complete $k$-colouring $\phi$ of $(G,\sigma')$ now \textit{infers} a $k$-colouring $\gamma$ of $(G,\sigma)$ assigning colours from $M_k'$,
where assigning colour $i^+$ to a vertex $v$ means that we assign colour $+i$ to $v$ without switching it,
while assigning colour $i^-$ means we first switch $v$ before assigning colour $+i$ to it.
The notion of completeness now extends to inferred $k$-colourings; that is, we say that an inferred $k$-colouring is \textit{complete} if it verifies that no two adjacent vertices are
\begin{itemize}
	\item assigned colours in $\{0^-,0^+\}$,
	\item joined by a positive edge and assigned colour $i^+$ for $i>0$,
	\item joined by a negative edge and assigned colours $i^-$ and $i^+$ for $i>0$,
\end{itemize}
while there exist edges realising any other combination of an edge sign and two colours.
In particular, note that an inferred $k$-colouring yields a complete $k$-colouring (in the usual sense), after switching vertices and assigning colours as described in the previous paragraph.
One could expect to have colours of the form $(-i)^-$ and $(-i)^+$ in the definition above;
such colours are actually not required, as explained below.

Let us state a few remarks about these notions.
First, we note that the notions of p-edges and n-edges of type $(i,j)$ apply to inferred colourings,
and, particularly, that Observation~\ref{observation:types-edges} adapts naturally to inferred complete colourings.
Next, we also note that, for $k$ odd, it is important to have both colours $0^-$ and $0^+$ to take into account that, in a complete $k$-colouring,
assigning $\pm 0$ to a vertex might require that vertex to be switched beforehand.
Thus, while colours $-0$ and $+0$ are equivalent for complete colourings,
for inferred complete colourings there is real distinction between $0^-$ and $0^+$.
We note also that we have not included to the sets $M_k'$ colours of the form $(-i)^-$ and $(-i)^+$,
which would correspond, for a complete $k$-colouring, to assigning colour $-i$ to a vertex after possibly switching it.
This is because these configurations are covered by other colours of $M_k'$:

\begin{observation}\label{observation:not-minus-in-signed-colours}
Assume $(G,\sigma)$ is a signed graph, and $\phi$ is a complete $k$-colouring of a signed graph $(G,\sigma')$ equivalent to $(G,\sigma)$.
Let $v$ be a vertex of $G$. Then:
\begin{itemize}
	\item if $v$ was switched (as going from $(G,\sigma)$ to $(G,\sigma')$) and $\phi(v)=-i \neq 0$,
	then, when switching $v$ (in $(G,\sigma')$) and changing the colour of $v$ to $+i$ (and keeping the same colours by $\phi$ for all other vertices),
	we get a complete $k$-colouring of the resulting signed graph.
	
	\item if $v$ was not switched and $\phi(v)=-i \neq 0$,
	then, when switching $v$ (in $(G,\sigma')$) and changing the colour of $v$ to $+i$,
	we get a complete $k$-colouring of the resulting signed graph.
\end{itemize}
\end{observation}

\begin{proof}
This follows from the fact that every edge incident to $v$ impacted by the said modifications
remains of the same type, i.e., remains a p-edge or n-edge of a given type.
\end{proof}

The next result stands as an analogue to Observation~\ref{observation:switch-negated-colours} for complete inferred colourings, and will be particularly useful when dealing with such colourings.

\begin{lemma}\label{lemma:inferred-force-colours}
Let $(G,\sigma)$ be a signed graph, $\phi$ be a complete $k$-colouring of $(G,\sigma)$, and $\gamma$ be the (complete) $k$-colouring of $(G,\sigma)$ inferred from $\phi$. By swapping any two ``opposite'' colours $i^+$ and $i^-$, \textit{i.e.}, changing to $i^-$ the colour by $\gamma$ of every vertex assigned colour $i^+$, and \textit{vice versa}, we get an inferred complete $k$-colouring of $(G,\sigma)$.
\end{lemma}

\begin{proof}
This follows from the fact that an n-edge of type $(i,j)$ by $\gamma$ becomes, after swapping $i^+$ and $i^-$, a p-edge of type $(i,j)$, and \textit{vice versa}, except in the case where $i=j$, as n-edges and p-edges of type $(i,i)$ remain of the same type after swapping colours $i^+$ and $i^-$.
\end{proof}

To avoid ambiguities, throughout this work we will stick to the notion of complete $k$-colourings we have originally introduced, 
i.e., without using signed colours, as much as we can.
We will, however, allow ourselves to consider signed colours instead in restricted contexts,
in which their use permits to simplify proofs and arguments in a significant way.
Particularly, signed colours will be useful when dealing with a fixed signature of a graph, in which context they can be seen as the colours of a classic colouring, with the exception that colour $0^-$ corresponds to vertices with colour~$\pm 0$ that must be switched. 
In any case, the use of signed colours will be mentioned explicitly, every time they are employed.

\medskip
\noindent {\bf Remark.}
Still to avoid any possible confusion, let us emphasise that, throughout this work,
colours of the form $+i$ or $-i$ are assigned by $k$-colourings,
while colours of the form $i^-$ and $i^+$ are assigned by inferred $k$-colourings.
Colours of the form $i$ refer mainly to the vertices of $K_k^*$ and are mostly employed to deal with edge types.

\subsection{First bounds on the achromatic number of signed graphs}\label{subsection:first-results-signed}

We here initiate the study of the achromatic number of signed graphs (with respect to our conventions) by providing first upper bounds on $\psi(G,\sigma)$ for every signed graph $(G,\sigma)$. 

We start off with the following easy observation:

\begin{observation}\label{observation:large-matching}
Let $k \geq 2$. If a signed graph $(G, \sigma)$ admits complete $k$-colourings, then $G$ has a matching of size at least $\lfloor k/2 \rfloor$.
Consequently, if the maximum matching of $G$ has size $\alpha$, then $\psi(G,\sigma) \leq 2\alpha+1$.
\end{observation}

\begin{proof}
In a complete $k$-colouring of $(G,\sigma)$, there must be an n-edge of type $(i,i)$ for every $i \in \{1,\dots,\lfloor k/2 \rfloor\}$.
Thus, for every $i$, there must be an edge of $(G,\sigma)$ with ends assigned colours $\pm i$, and $G$ must thus have a set of $\lfloor k/2 \rfloor$ disjoint edges.
\end{proof}

Partly due to Observation~\ref{observation:large-matching}, we note now that the achromatic number of a signed graph cannot be too large.

\begin{proposition} \label{proposition:upper-bound-n}
For every signed graph $(G,\sigma)$, we have $\psi(G, \sigma) \leq |V(G)|$.
\end{proposition}

\begin{proof}
Set $n=|V(G)|$.
The claim follows from the fact that every assignment of colours from $M_k$ to the vertices of any signed graph equivalent to $(G,\sigma)$ cannot form a complete colouring, given that $k>n$.
Indeed, assume this is wrong, and suppose $\phi$ is a complete $k$-colouring of $(G,\sigma')$ with $k>n$, where $(G,\sigma')$ is a signed graph equivalent to $(G,\sigma)$.
Recall that, as seen in Observation~\ref{observation:large-matching}, in every complete colouring, every non-zero colour $\pm i$ must be assigned to at least two vertices, 
since the corresponding reduced signed graph must have an n-edge of type~$(i,i)$. 
In the current context, clearly this condition cannot be guaranteed if $k$ is even, since $k>n$.
Thus, $\phi$ must be a complete $k$-colouring of $(G,\sigma')$ with $k$ odd.
So that $\phi$ is complete, note that colour~$\pm 0$ must be assigned to at least one vertex.
But then, again, we deduce that, because $k>n$, there is a non-zero colour assigned to at most one vertex of $(G,\sigma')$,
contradicting that $\phi$ is complete. 
\end{proof}

Another easy, yet better, upper bound on $\psi(G,\sigma)$ can be deduced through another naive approach, 
which is based on the fact that, in a complete $k$-colouring of $(G,\sigma)$, we must have edges meeting nearly all possible combinations of a sign and of a pair of colours assigned to their ends, which is possible only if $G$ has sufficiently many edges.

\begin{proposition} \label{proposition:upper-bound-size}
Let $k \geq 1$. If $(G,\sigma)$ is a signed graph, then:
\begin{itemize}
	\item if $|E(G)| < k^2 -1$, then $\psi(G, \si) < 2k-1$;
	\item if $|E(G)| < k^2 $,  then $\psi(G, \si) < 2k$. 
\end{itemize}
\end{proposition}

\begin{proof}
This follows directly from the definitions.
Indeed, if $\phi$ is a complete $k$-colouring of $(G,\sigma')$, a signed graph equivalent to $(G,\sigma)$, 
then, for every edge of $K_k^*$ with sign $s$ joining two vertices $i$ and $j$ (being possibly the same),
we must have in the reduced signed graph $R(G,\sigma',\phi)$ an edge with sign $s$ joining two vertices being assigned colour $+i$ and $+j$, respectively, by $\phi$ 
(i.e., corresponding, in $(G,\sigma')$, to either a p-edge of type $(i,j)$ or an n-edge of type $(i,j))$.
From this, we deduce that, for a complete $k$-colouring of $(G,\sigma')$ to exist, the size of $G$ must be at least that of $K_k^*$.
The result now follows, since $K^*_{2k}$ and $K^*_{2k-1}$ have size $k^2$ and $k^2-1$, respectively.
\end{proof}

\subsection{Achromatic number of some signed complete graphs}\label{subsection:achromatic-cliques}

For an $n \geq 1$, we denote by $(K_n,+)$ the signed graph defined over $K_n$, the complete graph of order $n$, with all edges positive.
We denote by $(K_n,-)$ the signed graph defined over $K_n$ with all edges negative.
For a given matching $M$ of $K_n$, we also denote by $(K_n-M,-)$ the signed graph obtained from $(K_n,-)$ by removing the edges of $M$.

Complete graphs tend to be an interesting class of graphs to investigate when generalising a problem from graphs to signed graphs.
Indeed, complete graphs have a very symmetrical structure which is not preserved by most of their signatures.
Partly for this reason, signed complete graphs are promising candidates to look at, especially when looking at significant discrepancies between a graph problem and its signed counterparts.

Regarding complete colourings and the achromatic number, unsigned complete graphs behave in a quite expected way,
as, for every complete graph $K_n$ with order $n$, clearly $\psi(K_n)=\chi(K_n)=n$.
Through the next two results, we show that signing the edges of a complete graph can have different consequences,
from preserving the achromatic number to the highest  possible  value (that indicated in Proposition~\ref{proposition:upper-bound-n}),
to lowering it down to the smallest possible value (namely~2, since only edgeless signed graphs have achromatic number~$1$).
This also illustrates that, for a given graph, two signed graphs defined over it might behave quite differently in terms of the achromatic number.
 
\begin{theorem}\label{theorem:kn-positive}
For every $n \geq 1$, we have $\psi(K_n,+) = n$.
\end{theorem}
 
\begin{proof}
Set $K=(K_n,+)$.
Let us consider any $n$-colouring $\phi$ of $K$ obtained by assigning a distinct colour from $M_n$ to every vertex.
Note that $\phi$ is clearly proper.
We claim that $\phi$ is actually also complete.
To see this is true, consider the reduced signed graph $R$ obtained from $K$ and $\phi$.
Consider an edge of $K_n^*$ with sign $s$ joining two vertices $i$ and $j$.
We need to show that $R$ has an edge of type $(i,j,s)$.

\begin{itemize}
	\item If $i=j$ and $i \neq 0$ (and thus $s=-$), then $R$ has the claimed edge due to $K$ having a positive edge $uv$ with $\phi(u)=+i$ and $\phi(v)=-i$.
	\item If $i \neq j$ and $s=-$ (where, possibly, $i=0$), then $R$ has the claimed edge due to $K$ having a positive edge $uv$ with $\phi(u)=+i$ and $\phi(v)=-j$.
	\item If $i \neq j$ and $s=+$ (where, possibly, $i=0$), then $R$ has the claimed edge due to $K$ having a positive edge $uv$ with $\phi(u)=+i$ and $\phi(v)=+j$.
\end{itemize}

Thus, $\phi$ is complete and $\psi(K) \geq n$.
Proposition~\ref{proposition:upper-bound-n} now implies the exact statement of the theorem.
 \end{proof}
 
\begin{theorem}\label{theorem:kn-negative}
For every $n \geq 5$, we have $\psi(K_n,-) = 2$.
Furthermore, for every non-empty (not necessarily maximal) matching $M$ of $K_n$, we have $\psi(K_n-M,-)=3$.
\end{theorem}

\begin{proof}
Let $K$ and $K'$ denote the signed graphs $(K_n,-)$ and $(K_n-M,-)$, respectively,
where $n$ and $M$ fulfill the conditions in the statement.
We start by showing that $\psi(K) < 4$ and $\psi(K') < 4$.
Assume this is wrong, and, first (the case of $K'$ will be discussed later), 
that %$K$ admits a complete $k$-colouring $\phi$ with $k \geq 4$.
some vertices of $K$ can be switched to attain an equivalent signed graph admitting a complete $k$-colouring $\phi$ with $k \geq 4$.
We consider signed colours instead, and consider $\gamma$ the complete $k$-colouring of $K$ inferred from $\phi$.
Recall that this $\gamma$ holds in $K$ directly, i.e., we can assume all edges of $K$ are negative,
while $\gamma$ assigns colours from $M_k'$.
Since $k \geq 4$, note that $\gamma$ must induce both an n-edge $u_1v_1$ of type $(1,1)$ and an n-edge $u_2v_2$ of type $(2,2)$.
Because all edges of $K$ are negative, note that either $\gamma(u_1)=\gamma(v_1)=1^-$ or $\gamma(u_1)=\gamma(v_1)=1^+$,
and, similarly, either $\gamma(u_2)=\gamma(v_2)=2^-$ or $\gamma(u_2)=\gamma(v_2)=2^+$.
By Lemma~\ref{lemma:inferred-force-colours},
we may assume that $\gamma(u_1)=\gamma(v_1)=1^+$ and $\gamma(u_2)=\gamma(v_2)=2^+$.
Because all edges of $K$ joining a vertex in $\{u_1,v_1\}$ and a vertex in $\{u_2,v_2\}$ are negative,
note that all these edges are n-edges of type $(1,2)$.
Since $\gamma$ is complete, there must also be p-edges of type $(1,2)$.
However, note that any other vertex $w$ with $\gamma(w) \in \{1^-,1^+,2^-,2^+\}$ must actually be assigned a colour in $\{1^+,2^+\}$,
as otherwise, because all edges of $K$ are negative, we would deduce that there is a p-edge $wu_i$ or $wv_i$ of type $(1,1)$ or $(2,2)$.
So we are in a situation where all edges of the signed graph are negative, and colours $1^-$ and $2^-$ cannot be assigned;
note that $\gamma$ cannot have a p-edge of type $(1,2)$ in those conditions, a contradiction.
Furthermore, note that these arguments hold for $K'$ as well, as, because we removed the edges of a matching, 
vertex $w$ above remains adjacent to at least one of $u_1$ and $v_1$, and to at least one of $u_2$ and $v_2$.
Thus, also $K'$ cannot %admit complete $k$-colourings with $k \geq 4$.
have its vertices switched to reach an equivalent signed graph admitting complete $k$-colourings with $k \geq 4$.

We now determine the exact values of $\psi(K)$ and $\psi(K')$.

\begin{itemize}
	\item It can be noted first that $\psi(K) < 3$, essentially through the same arguments as earlier.
	Towards a contradiction, assume indeed that $K$ has an equivalent signed graph that
	 admits a complete $3$-colouring $\phi$.
	We can again consider the complete $3$-colouring $\gamma$ of $K$ inferred from $\phi$.
	Thus, no vertex of $K$ was switched.
	As earlier, we may assume there is an n-edge $u_1v_1$ of type $(1,1)$, where $u_1$ and $v_1$ are assigned colour $1^+$ by $\gamma$,
	and, from this, we deduce that all other vertices assigned a colour in $\{1^-,1^+\}$ are assigned colour~$1^+$.
	Because there should be neither p-edges nor n-edges of type $(0,0)$, note that a colour in $\{0^-,0^+\}$ can be assigned to at most one vertex of $K$.
	Since all edges of $K$ are negative, we then deduce that there are p-edges of type $(0,1)$ but no n-edges of type $(0,1)$, or \textit{vice versa}.
	Then $\gamma$ and $\phi$ cannot be complete, a contradiction.
	
	Recall that $K_2^*$ is the signed graph with only one vertex, $1$, incident to the only edge, being a negative loop.
	It is now easy to see that by assigning colour~$+1$ to all vertices of $K$, we get a complete $2$-colouring.
	Thus, $\psi(K)=2$.
	
	\item We claim that $\psi(K')=3$.
	To see this is true, consider the following arguments.
	Let $u$ and $v$ be two non-adjacent vertices of $K'$ (they exist since $M$ is non-empty).
	Because $n \geq 5$ and $M$ is a matching, note that $u$ and $v$ share a common neighbour $w$.
	Let us now denote by $K''$ the signed graph equivalent to $K'$ obtained by switching $v$,
	and consider the $3$-colouring $\phi$ of $K''$ obtained by assigning colour~$\pm 0$ to $u$ and  $v$,
	and colour~$+1$ to all other vertices.
	Since colour $\pm 0$ is assigned to $u$ and $v$ only, note that this yields neither p-edges nor n-edges of type $(0,0)$.
	Furthermore, the edges $uw$ and $vw$ are an n-edge of type $(0,1)$ and a p-edge of type $(0,1)$, respectively.
	Finally, every other edge, i.e., not incident to $u$ and $v$ (there is at least one such, since $n \geq 5$ and $M$ is a matching), is negative and has both ends assigned colour $+1$ by $\phi$, 
	and is thus an n-edge of type $(1,1)$.
	Thus, $\phi$ is complete.\qedhere
\end{itemize}
\end{proof}

\subsection{Achromatic number of signed paths and cycles}\label{subsection:achromatic-paths-cycles}

As reported in~\cite{HM97}, the achromatic number of paths and cycles was investigated by a few groups of authors,
resulting in the following two results (in which $P_n$ and $C_n$ refer to the path and cycle of order $n$, respectively):

\begin{theorem}[Hell, Miller~\cite{HM76}]
For every $n \geq 1$, we have $$\psi(P_n)=\max\left\{k : \left(\left\lfloor \frac{k}{2} \right\rfloor +1\right)(k-2)+1 \leq n-1\right\}.$$
\end{theorem}

\begin{theorem}[Hell, Miller~\cite{HM76}]
For every $n \geq 3$, we have $$\psi(C_n)=\max\left\{k : k \left\lfloor \frac{k}{2} \right\rfloor \leq n\right\} - s(n),$$
where $s(n) \in \{0,1\}$ is the number of positive integer solutions of the equation $n=2x^2+x+1$.
\end{theorem}

These results are actually reminiscent of the corresponding ones for signed paths and cycles,
as proved in what follows below.

\begin{theorem}\label{theorem:pn}
Let $n \geq 1$, and $(P_n,\sigma)$ be a signed graph defined over $P_n$. We have: 
$$ \psi(P_n, \sigma) = k = \max \left( \left\{ 2x : x^2 \le n-1\right\} \cup \left\{ 2x+1 : \left(x+1\right)^2-1 \le n-1\right\}  \right).$$
\end{theorem}

\begin{proof}
By Proposition~\ref{proposition:upper-bound-size}, we have $\psi(P_n,\sigma) \leq k$.
We now prove that some vertices of $(P_n,\sigma)$ can be switched so that the resulting signed graph admits a complete $k$-colouring $\phi$.
We construct $\phi$ (and switch vertices of $(P_n,\sigma)$) through the following procedure below.
Because $K_k^*$ is connected and has all of its vertices being of even degree,
by Euler's Theorem there is a Eulerian trail $W=(x_1, x_2, \dots, x_1)$ in $K_k^*$.
By definition, every edge of $K_k^*$ appears exactly once in $W$.
The procedure now considers the vertices $v_1,\dots,v_n$ of $(P_n,\sigma)$ one after the other, from one end-vertex $v_1$ to the other one $v_n$.
For the first vertex $v_1$ of $(P_n,\sigma)$, we set $\phi(v_1)=+x_1$.
Considering now the $i$th vertex $v_i$ of $(P_n,\sigma)$ such that $v_{i-1}$ has been treated, we set $\phi(v_i)=+x_i$,
and, if necessary, we switch $v_i$ so that, with respect to $x_{i-1}x_i$, the sign of $v_{i-1}v_i$ matches that of $x_{i-1}x_i$ 
(i.e., $v_{i-1}v_i$ is positive if $x_{i-1}x_i$ is positive in $K_k^*$, or $v_{i-1}v_i$ is negative otherwise).
Once the procedure ends, $\phi$ is clearly a complete $k$-colouring,
as, for every edge $x_ix_j$ of $K_k^*$ with sign $s$, there is, in the resulting equivalent signed graph defined over $P_n$,
an edge with sign $s$ having its ends assigned colours $+x_i$ and $+x_j$.
\end{proof}

In the next result, the \textit{balance} of a signed graph, being either \textit{even} or \textit{odd}, 
refers to the parity of its number of negative edges. That is, a signed graph is of \textit{even balance} if its number of negative edges is even, while it is of \textit{odd balance} otherwise. 
For simplicity, in the next result, for a $k \geq 1$ we denote by $m_k$ the number of edges of $K_k^*$.

\begin{theorem}\label{theorem:cn}
Let $n \geq 3$, and $(C_n,\sigma)$ be a signed graph defined over $C_n$. Let also $k_0$ denote the largest integer $k$ such that $n \ge m_k$. Then:
\begin{itemize}
    \item if either 
    \begin{itemize}
        \item $n \ge m_{k_0} + 2$,
        \item $n = m_{k_0} + 1$ and $K_k^*$ and $(C_n,\si)$ have different balance, or
        \item $n = m_{k_0}$ and $K_k^*$ and $(C_n,\si)$ have the same balance,
    \end{itemize}
    then $\psi(C_n, \si) = k_0$;
    
    \item if either 
    \begin{itemize}
        \item $n = m_{k_0}+1$ and $K_k^*$ and $(C_n,\si)$ have the same balance, or
        \item $n = m_{k_0}$, $k_0$ is odd, and $K_k^*$ and $(C_n, \si)$ have different balance,
    \end{itemize}
    then $\psi(C_n,\si) = k_0 - 1$;
    
    \item otherwise, $\psi(C_n, \si) = k_0 -2$.
\end{itemize}
\end{theorem}

\begin{proof}
Since the claim can easily be verified for the cases where $n \leq 5$, we focus on cases where $n>5$. Since $C_n$ has size $n$, note that, by definition of $m_{k_0}$, we have $\psi(C_n,\sigma) \leq k_0$. We split the proof into several cases.

\begin{itemize}
    \item \textit{$n \ge m_{k_0} +2$.}
    
    Then, note that $C_n$ contains $P_{n-1}$ as an induced subgraph. In that case, Theorem~\ref{theorem:pn}, applied to that subgraph (see later Corollary~\ref{corollary:achromatic-signed-monotonous-subgraphs}), implies that $\psi(C_n, \si) = k_0$.
    
    \item \textit{$n = m_{k_0} +1$, and $K_{k_0}^*$ and $(C_n,\sigma)$ have different balance. }
    
    Here, consider $P_n^*$, the signed graph obtained from $(C_n,\sigma)$ by removing a negative edge (which can be assumed to exist, up to switching a vertex -- for convenience, we still denote by $(C_n,\sigma)$ the eventual signed graph). By Theorem~\ref{theorem:pn}, we have $\psi(P_n^*)=k_0$, and, since $P_n^*$ has exactly $m_{k_0}$ edges and all vertices of $K_{k_0}^*$ have even degree, as described in the proof of Theorem~\ref{theorem:pn} a complete $k_0$-colouring $\phi$ of $P_n^*$ can be derived to a Eulerian trail in $K_{k_0}^*$. Furthermore, the two end-vertices of $P_n^*$ must be assigned the same colour by $\phi$, as they are the only two vertices with odd degree of $P_n^*$, while all vertices of $K_{k_0}^*$ have even degree and thus must appear an even number of times along the trail. So, because, in $(C_n,\sigma)$, these two vertices are joined by a negative edge, we deduce that $\phi$ also forms a complete $k_0$-colouring of $(C_n,\sigma)$. Thus, $\psi(C_n, \si) = k_0$.
    
    \item \textit{$n = m_{k_0} +1$, and $K_{k_0}^*$ and $(C_n,\sigma)$ have the same balance.}
    
    Note first that $C_n$ contains $P_{n-1}$ as an induced subgraph. Thus, from that subgraph, as earlier we deduce that $\psi(C_n, \si) \ge k_0 -1$. 
    Towards a contradiction to the claim, assume $(C_n, \si)$ (or, rather, a signed graph obtained from $(C_n, \si)$ by switching vertices -- in what follows and later on, for simplicity we still refer to this signed graph as $(C_n,\sigma)$) admits a complete $k_0$-colouring $\phi$.
    Since $C_n$ has $n = m_{k_0} +1$ edges and $K_{k_0}^*$ has $n-1$ edges, we deduce that two edges, $e$ and $e'$, of $(C_n,\sigma)$ correspond to the same edge of $K_{k_0}^*$ through $\phi$. Now consider $P_n^* = (C_n,\sigma) - e$.
    Clearly, $\phi$ is also complete in $P_n^*$, meaning that $P_n^*$ and $K_{k_0}^*$ must have the same balance. Since $(C_n,\sigma)$ and $K_{k_0}^*$ also have the same balance, we deduce that $e$ must be a positive edge. But, now, since, again, $P_n^*$ has exactly $m_{k_0}$ edges, by $\phi$ necessarily its two end-vertices must be assigned the same colour. This contradicts the fact that $\phi$ is complete in $(C_n,\sigma)$, since, in that signed graph, these two vertices are joined by a positive edge. Thus, $\psi(C_n, \si) = k_0 -1$.
    
    \item \textit{$n = m_{k_0}$, and $K_{k_0}^*$ and $(C_n,\sigma)$ have the same balance.}
    
    Consider a Eulerian tour $C$ in $K_{k_0}^*$. Because $C$ and $(C_n,\sigma)$ are two signed cycles (defined over $C_n$) with the same balance, they are equivalent by Lemma~\ref{lemma:switching-balance}.
    It is thus possible, through switching vertices of $(C_n,\sigma)$, to reach another signed graph $(C_n,\sigma')$ which is precisely $C$. By then assigning colours to the vertices of $(C_n,\sigma')$ as the vertices of $K_{k_0}^*$ are traversed by $C$, we obtain a complete $k_0$-colouring of $(C_n,\sigma')$. Thus, $\psi(C_n, \si) = k_0$. 
    
    \item \textit{$n = m_{k_0}$, and $K_{k_0}^*$ and $(C_n,\sigma)$ have different balance.}
    
    Note that $\psi(C_n, \si) < k_0$. Indeed, towards a contradiction suppose $(C_n, \si)$ admits a complete $k_0$-colouring $\phi$. Since $C_n$ has size precisely $m_{k_0}$, the size of $K_{k_0}^*$, to $\phi$ corresponds a Eulerian tour in $K_{k_0}^*$, which is not possible since $K_{k_0}^*$ and $(C_n,\sigma)$ have different balance. Thus, $\phi$ cannot be complete and $\psi(C_n, \si) < k_0$.
    
    We now consider two last cases.
    
    \begin{itemize}
        \item If $k_0$ is odd, then $K_{k_0-1}^*$ has less than $m_{k_0} - 2$ edges. Since $C_n$ contains $P_{n-1}$, which has size $m_{k_0} -2$, as an induced subgraph, we directly deduce that $\psi(C_n, \si) = k_0 -1$. 
        
        \item If $k_0$ is even, then $K_{k_0-2}^*$ has less than $m_{k_0} - 2$ edges. Since $C_n$ contains $P_{n-1}$, which has size $m_{k_0} -2$, as an induced subgraph, we deduce that $\psi(C_n, \si) \ge k_0 - 2$. Towards a contradiction to the statement, assume $(C_n,\sigma)$ admits a complete $(k_0-1)$-colouring $\phi$. Said contradiction can then be obtained similarly as in a previous case (that where $n = m_{k_0} +1$, and $K_{k_0}^*$ and $(C_n,\sigma)$ have the same balance), by removing a positive edge from $(C_n,\sigma)$, looking at the correspondence between the remaining signed path (with respect to $\phi$) and a Eulerian tour in $K_{k_0-1}^*$, and deducing the existence of two adjacent vertices of $(C_n,\sigma)$ being assigned the same colour by $\phi$ while being joined through a positive edge. Thus, $\psi(C_n, \si)=k_0 - 2$. \qedhere
    \end{itemize}
\end{itemize}
\end{proof}

%%%%%%%%%%%%%%%%%%%%%%%%%%%%%%%%%%%%%%%%%%%%%%%%%%%%%%
%%%%%%%%%%%%%%%%%%%%%%%%%%%%%%%%%%%%%%%%%%%%%%%%%%%%%%
%%%%%%%%%%%%%%%%%%%%%%%%%%%%%%%%%%%%%%%%%%%%%%%%%%%%%%
%%%%%%%%%%%%%%%%%%%%%%%%%%%%%%%%%%%%%%%%%%%%%%%%%%%%%%
%%%%%%%%%%%%%%%%%%%%%%%%%%%%%%%%%%%%%%%%%%%%%%%%%%%%%%
%%%%%%%%%%%%%%%%%%%%%%%%%%%%%%%%%%%%%%%%%%%%%%%%%%%%%%
%%%%%%%%%%%%%%%%%%%%%%%%%%%%%%%%%%%%%%%%%%%%%%%%%%%%%%
%%%%%%%%%%%%%%%%%%%%%%%%%%%%%%%%%%%%%%%%%%%%%%%%%%%%%%
%%%%%%%%%%%%%%%%%%%%%%%%%%%%%%%%%%%%%%%%%%%%%%%%%%%%%%
%%%%%%%%%%%%%%%%%%%%%%%%%%%%%%%%%%%%%%%%%%%%%%%%%%%%%%
%%%%%%%%%%%%%%%%%%%%%%%%%%%%%%%%%%%%%%%%%%%%%%%%%%%%%%
%%%%%%%%%%%%%%%%%%%%%%%%%%%%%%%%%%%%%%%%%%%%%%%%%%%%%%
%%%%%%%%%%%%%%%%%%%%%%%%%%%%%%%%%%%%%%%%%%%%%%%%%%%%%%
%%%%%%%%%%%%%%%%%%%%%%%%%%%%%%%%%%%%%%%%%%%%%%%%%%%%%%

\section{Operations on signed graphs} \label{section:operations}

The results of this section are split into two distinct parts.
We start, in Subsection~\ref{subsection:elementary-operations}, 
by studying how modifying a signed graph (by removing vertices or edges, or by changing its signature) can alter its achromatic number.
Then, in Subsection~\ref{subsection:homomorphisms-signed-graphs}, we explore the connection, in terms of achromatic number, 
between a signed graph and its homomorphic images.
In each case, every bound we establish is proved to be tight, in a strong way.

\subsection{Elementary operations} \label{subsection:elementary-operations}

We start by investigating the possible consequences, on the achromatic number,
of removing vertices in signed graphs.
In the unsigned context, this very question was considered by Geller and Kronk, who proved the following:

\begin{theorem}[Geller, Kronk \cite{GK74}] \label{theorem:removing-vertex-unsigned}
Let $G$ be a graph. For every $v \in V(G)$, we have
$$ \psi(G) -1 \le \psi(G-v) \le \psi(G).$$
Furthermore, there are contexts where both bounds can be attained.
\end{theorem}

The remarkable point in Theorem~\ref{theorem:removing-vertex-unsigned},
and also in close results we will consider later on,
is that both bounds can be attained for some graphs $G$ and vertices $v$.
This is also an aspect that we will strive to obtain in our results throughout this section.
Particularly, in the signed context, note that the parity of the number $k$ of colours in a complete $k$-colouring is a crucial parameter,
as the difference between the size of any $K_n^*$ and the size of the next $K_{n+1}^*$ depends mostly on the parity on $n$.
Mainly for this reason, we think an interesting aspect behind our bounds in this section
is to investigate whether they are tight regardless of the parity of the number of assigned colours.
As will be apparent later on, this requires to design several different types of configurations and examples.

More precisely, our analogue of Theorem~\ref{theorem:removing-vertex-unsigned} reads as follows:

\begin{theorem}\label{theorem:removing-vertex-signed}
Let $(G, \si)$ be a signed graph.
For every $v \in V(G)$, we have
$$ \psi(G,\si) -2 \le \psi(G-v,\si) \le \psi(G,\si) .$$
Furthermore, there are contexts where both bounds can be attained.
Particularly, for every $\psi(G,\sigma) \geq 3$, the lower bound can be attained.
\end{theorem}

\begin{proof}
We prove first that $\psi(G-v,\sigma) \le \psi(G,\sigma)$.
Let $\phi$ be a complete $k$-colouring of $(G-v,\sigma)$.
Our goal is to extend $\phi$ to a complete $k'$-colouring $\phi'$ of $(G,\sigma)$, where $k' \geq k$, by essentially finding a valid colour for $v$.
To that aim, let us start by setting $\phi'(u)=\phi(u)$ for every $u \in V(G) \setminus \{v\}$.
Note that if there is an $i$ such that neither $-i$ nor $+i$ is assigned by $\phi$ to any of the neighbours of $v$ in $(G,\sigma)$,
then, upon setting $\phi'(v)=+i$, we get that $\phi'$ is a complete $k$-colouring of $(G,\sigma)$.
A similar conclusion holds if, for any $i \in V(K_k^*)$,
vertex $v$ has neither a neighbour assigned colour $+i$ connected via a positive edge nor a neighbour assigned colour $-i$ connected via a negative edge.

Thus, assume now all colours $\pm i$ of $M_k$ are assigned by $\phi$ to some neighbours of $v$ in $(G,\sigma)$,
and, more precisely, that the neighbours of $v$ verify the colour conditions above.
Now, if $\pm 0 \not \in M_k$ (which occurs if $k=2n$), then note that assigning colour~$\pm 0$ by $\phi'$ to $v$ results in a complete $(k+1)$-colouring of $(G,\sigma)$.
Particularly, note that $(G,v)$ has both a p-edge and an n-edge of type $(0,i)$ for every $i \in \{1,\dots,k\}$.
So, assume last that $\pm 0 \in M_k$ (and thus $k=2n+1)$.
By the arguments above, because we are not done when just setting $\phi'(v)=\pm 0$, it means that $v$ has neighbours assigned label~$\pm 0$ by $\phi$.
We start by switching some of these neighbours, if necessary, so that $v$ is connected to all its neighbours assigned colour~$\pm 0$ via negative edges.
We next modify $\phi'$ by setting $\phi'(u)=+(k+1)$ for every vertex $u \in V(G) \setminus \{v\}$ verifying $\phi(u)=\pm 0$,
before setting $\phi'(v)=+(k+1)$.
As a result, $\phi'$ is a complete $(k+1)$-colouring of $(G,\sigma)$.
In particular, there is an n-edge of type $(k+1,k+1)$ due to $v$ having all of its neighbours of colour~$+(k+1)$ (there is at least one such) being joined to $v$ via negative edges.
Also, for every $\pm i \in M_{k+1} \setminus \{\pm(k+1)\}$, there is a p-edge and an n-edge of type $(i,k+1)$: the p-edge because, for every non-negative colour $+i$,
vertex $v$ is connected to a neighbour with colour~$+i$ via a positive edge or to a neighbour with colour $-i$ via a negative edge, and the n-edge because, for every non-negative colour $+i$, vertex $v$ is connected to a neighbour with colour~$-i$ via a positive edge or to a neighbour with colour $+i$ via a negative edge.
Thus, we have $\psi(G-v,\sigma) \le \psi(G,\sigma)$.

We now prove that $\psi(G,\sigma) -2 \le \psi(G-v,\sigma)$,
which follows essentially by the same arguments as earlier.
Indeed, assume that $\psi(G-v,\sigma)=k$, and let $\phi$ be a complete $k$-colouring of $(G,\sigma)$.
Suppose $\phi(v)=+i$.
Let $G'$ be the graph obtained from $G$ by removing all vertices, including $v$, being assigned colour~$\pm i$ by $\phi$.
Let us denote by $u_1,\dots,u_d$ the vertices different from $v$ that were removed.
By Observation~\ref{observation:remove-class}, the restriction of $\phi$ to $(G',\sigma)$ forms a complete $(k-2)$-colouring $\phi'$ (or a complete $(k-1)$-colouring if $\phi(v) = \pm 0$).
Using the same arguments as in the first case of the current proof, $\phi'$ can first be extended to a complete colouring of $(G'+ u_1, \sigma)$ using the same number of colours ($k-2$ or $k-1$). After that, the resulting colouring can then be extended to a complete colouring of $(G'+ \{u_1,u_2\}, \sigma)$ using the same number of colours. Repeating those arguments by adding back all $u_i$'s one by one, we end up, from $\phi'$, with a complete $k'$-colouring of $(G-v,\sigma)$ with $k' \ge k-2$. Thus, $\psi(G,\sigma) -2 \le \psi(G-v,\sigma)$.

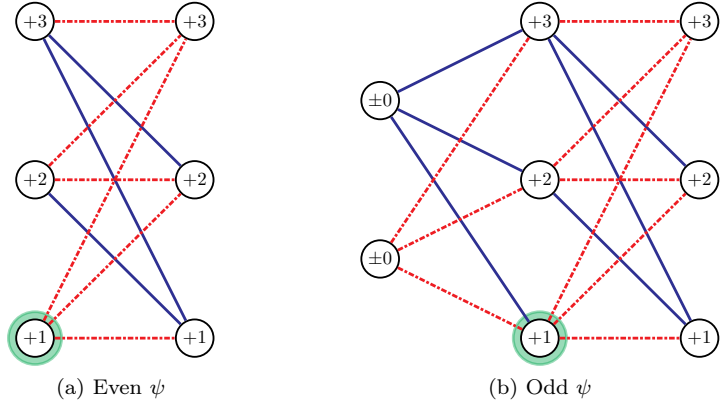
\begin{figure}[!t]
 	\centering
 	\subfloat[Even $\psi$]{
    \scalebox{0.7}{
	\begin{tikzpicture}[inner sep=0.7mm]	
	
    \draw[Green, line width=1.5pt,opacity=0.4, fill] (0,-6) circle (0.5cm);
	
	\node[draw, circle, black, line width=1pt, fill=white](u1) at (0,0)[]{$+3$};
	\node[draw, circle, black, line width=1pt,fill=white](u2) at (0,-3)[]{$+2$};
	\node[draw, circle, black, line width=1pt,fill=white](u3) at (0,-6)[]{$+1$};
	\node[draw, circle, black, line width=1pt,fill=white](v1) at (3,0)[]{$+3$};
	\node[draw, circle, black, line width=1pt,fill=white](v2) at (3,-3)[]{$+2$};
	\node[draw, circle, black, line width=1pt,fill=white](v3) at (3,-6)[]{$+1$};
	
    \draw [-, line width=1.5pt, Red, densely dashdotted] (u1) -- (v1);
    \draw [-, line width=1.5pt, Blue] (u1) -- (v2);
    \draw [-, line width=1.5pt, Blue] (u1) -- (v3);
	
    \draw [-, line width=1.5pt, Red, densely dashdotted] (u2) -- (v1);
    \draw [-, line width=1.5pt, Red, densely dashdotted] (u2) -- (v2);
    \draw [-, line width=1.5pt, Blue] (u2) -- (v3);
	
    \draw [-, line width=1.5pt, Red, densely dashdotted] (u3) -- (v1);
    \draw [-, line width=1.5pt, Red, densely dashdotted] (u3) -- (v2);
    \draw [-, line width=1.5pt, Red, densely dashdotted] (u3) -- (v3);
	\end{tikzpicture}
    }
    }
    \hspace{40pt}
 	\subfloat[Odd $\psi$]{
    \scalebox{0.7}{
	\begin{tikzpicture}[inner sep=0.7mm]	
	
    \draw[Green, line width=1.5pt,opacity=0.4, fill] (u3) circle (0.5cm);
	
	\node[draw, circle, black, line width=1pt,fill=white](w1) at (-3,-1.5)[]{$\pm 0$};
	\node[draw, circle, black, line width=1pt,fill=white](w2) at (-3,-4.5)[]{$\pm 0$};
	\node[draw, circle, black, line width=1pt,fill=white](u1) at (0,0)[]{$+3$};
	\node[draw, circle, black, line width=1pt,fill=white](u2) at (0,-3)[]{$+2$};
	\node[draw, circle, black, line width=1pt,fill=white](u3) at (0,-6)[]{$+1$};
	\node[draw, circle, black, line width=1pt,fill=white](v1) at (3,0)[]{$+3$};
	\node[draw, circle, black, line width=1pt,fill=white](v2) at (3,-3)[]{$+2$};
	\node[draw, circle, black, line width=1pt,fill=white](v3) at (3,-6)[]{$+1$};
	
    \draw [-, line width=1.5pt, Red, densely dashdotted] (u1) -- (v1);
    \draw [-, line width=1.5pt, Blue] (u1) -- (v2);
    \draw [-, line width=1.5pt, Blue] (u1) -- (v3);
	
    \draw [-, line width=1.5pt, Red, densely dashdotted] (u2) -- (v1);
    \draw [-, line width=1.5pt, Red, densely dashdotted] (u2) -- (v2);
    \draw [-, line width=1.5pt, Blue] (u2) -- (v3);
	
    \draw [-, line width=1.5pt, Red, densely dashdotted] (u3) -- (v1);
    \draw [-, line width=1.5pt, Red, densely dashdotted] (u3) -- (v2);
    \draw [-, line width=1.5pt, Red, densely dashdotted] (u3) -- (v3);
	
    \draw [-, line width=1.5pt, Blue] (w1) -- (u1);
    \draw [-, line width=1.5pt, Blue] (w1) -- (u2);
    \draw [-, line width=1.5pt, Blue] (w1) -- (u3);
	
    \draw [-, line width=1.5pt, Red, densely dashdotted] (w2) -- (u1);
    \draw [-, line width=1.5pt, Red, densely dashdotted] (w2) -- (u2);
    \draw [-, line width=1.5pt, Red, densely dashdotted] (w2) -- (u3);
	\end{tikzpicture}
    }
    }

\caption{Signed graphs attesting of the tightness of the lower bound in Theorem~\ref{theorem:removing-vertex-signed}.
Vertices highlighted in green are vertices which, when removed, make the achromatic number decrease by~$2$.
In (a) is also depicted a complete $6$-colouring, 
and in (b) a complete $7$-colouring.
%and in (d) a complete $5$-colouring.
Dashed red edges are negative edges, while solid blue edges are positive edges.  
\label{figure:remove-vertex}}
\end{figure}

\medskip

Let us now focus on the tightness of the bounds.
Regarding the upper bound, we simply remark that adding an isolated vertex $v$ to a signed graph $(G,\sigma)$ does not modify its achromatic number.
Thus, finding signed graphs $(G, \si)$ having a vertex $v \in V(G)$ such that $\psi(G,\sigma)=\psi(G-v,\sigma)$ is possible for any value of $\psi(G,\si)$.
We now provide examples showing that removing a vertex from a signed graph can decrease its achromatic number by exactly~$2$.
Precisely, for every $k \geq 3$, we prove that there is a signed graph $(G,\sigma)$ such that $\psi(G,\sigma)=k$,
and having a vertex $v$ such that $\psi(G-v,\sigma)=k-2$.
The constructions we provide are illustrated in Figure~\ref{figure:remove-vertex}.

\begin{itemize}
	\item We start by providing a construction for even values of $k$.
	Set $p=k/2$, where $k \geq 4$. We consider the following signed graph $(G,\sigma)$.
	The graph $G$ is the complete bipartite graph $K_{p,p}$, in which the vertices from the two partite sets are denoted $u_1,\dots,u_p$ and $v_1,\dots,v_p$, respectively.
	The signature $\sigma$ of $G$ we consider, is obtained by considering every edge $u_iv_j$,
	and setting $\sigma(u_iv_j)=-$ if $i \leq j$, and $\sigma(u_iv_j)=+$ otherwise.
	
	We claim that $\psi(G,\sigma)=k$, while $\psi(G-u_1,\sigma)=k-2$.
	First off, recall that $\psi(G,\sigma) \leq k$ by Theorem~\ref{proposition:upper-bound-n}, since $k=|V(G)|$.
	Furthermore, $\psi(G,\sigma)$ admits complete $k$-colourings, as, due to the signature $\sigma$,
	it can be checked that assigning colour~$+i$ to both $u_i$ and $v_i$ for every $i \in \{1,\dots,p\}$ results in a said colouring.
	Precisely, for every two $x,y \in \{1,\dots,p\}$,
	there is an edge $u_iv_j$ of any sign in $\{-,+\}$ such that $u_i$ and $u_j$ are assigned colour~$+x$ and $+y$, respectively,
	with the exception of a positive edge when $x=y$.
	Thus, $\psi(G,\sigma)=k$.
	To see now that $\psi(G-u_1,\sigma)=k-2$ holds, note that this signed graph $(G-u_1,\sigma)$ has precisely $p(p-1)$ edges.
	By Proposition~\ref{proposition:upper-bound-size}, we thus get that $\psi(G-u_1,\sigma)$ is at most $2p-2=k-2$,
	since $p(p-1)<p^2-1$.
	Since we have proved earlier that removing a vertex from a signed graph can decrease its achromatic number by at most~$2$,
	we thus get that $\psi(G-u_1,\sigma)=k-2$.
	
	\item We now provide a similar result for odd values of $k \geq 3$.
	Set $p=\lfloor k/2 \rfloor$, and consider the signed graph $(G,\sigma)$ obtained in the following way.
	We start from the same signature of $K_{p,p}$ as in the previous case, and add two new vertices,
	$w_p$ and $w_n$, where $w_p$ is joined to all $u_i$'s via positive edges,
	while $w_n$ is joined to all $u_i$'s via negative edges.
	
	We claim that, again, $\psi(G,\sigma)=k$, while $\psi(G-u_1,\sigma)=k-2$.
	Let us first consider $\psi(G,\sigma)$.
	By Proposition~\ref{proposition:upper-bound-size}, we deduce that $\psi(G,\sigma) \leq k$.
	We note that, to obtain a complete $k$-colouring of $\psi(G,\sigma)$,
	we can just assign colour~$\pm 0$ to both $w_p$ and $w_n$, and colour $+i$ to both $u_i$ and $v_i$, for every $i \in \{1,\dots,p\}$.
	Particularly, $(G,\sigma)$ has, for every $i \in \{1,\dots,p\}$, both a p-edge of type $(0,i)$ (due to the edge $w_pu_i$) and an n-edge of type $(0,i)$ (due to the edge $w_nu_i)$,
	while all other required types of edges are realised due to the edges joining the $u_i$'s and the $v_i$'s (recall the arguments used to deal with the previous case). Also, there are neither n-edges nor p-edges of type $(0,0)$ since $w_p$ and $w_n$ are not adjacent.
	Thus, $\psi(G,\sigma)=k$.
	We now prove that $\psi(G-u_1,\sigma)=k-2$.
	Note that the maximum matching of $G-u_1$, due to its bipartiteness, has size $p-1$.
	By Observation~\ref{observation:large-matching}, we deduce that $\psi(G-u_1,\sigma) \leq k-2$.
	Now, since we have proved earlier that removing a vertex from a signed graph can decrease its achromatic number by at most~$2$,
	we get that $\psi(G-u_1,\sigma) = k-2$. \qedhere
\end{itemize}
\end{proof}

An interesting consequence of Theorem~\ref{theorem:removing-vertex-signed}, is the following:

\begin{corollary} \label{corollary:achromatic-signed-monotonous-subgraphs}
If $(H,\pi)$ is an induced signed subgraph of a signed graph $(G,\sigma)$, then
$$ \psi(H, \pi) \leq \psi(G, \sigma).$$
\end{corollary}

Particularly, Corollary~\ref{corollary:achromatic-signed-monotonous-subgraphs} can be combined with results on the achromatic number of particular signed graphs.
For instance, having a long induced path in a signed graph, guarantees large achromatic number, according to Theorem~\ref{theorem:pn}.
Similarly, having a large positive clique (up to switching vertices) guarantees the same, by Theorem~\ref{theorem:kn-positive}.

Still about Corollary~\ref{corollary:achromatic-signed-monotonous-subgraphs}, a legitimate question to wonder is whether the ``induced'' requirement can be dropped from the statement.
One way to investigate this, is through studying the consequences, on the achromatic number, of removing an edge from a signed graph.
A similar question was actually studied in the unsigned context, with Geller and Kronk proving the following:

\begin{theorem}[Geller, Kronk \cite{GK74}] \label{theorem:removing-edge-unsigned}
Let $G$ be a graph. For every $e \in E(G)$, we have
$$ \psi(G) - 1 \leq \psi(G-e) \leq \psi(G) +1.$$
Furthermore, there are contexts where both bounds can be attained.
\end{theorem} 

Thus, in the unsigned context, removing a vertex from a graph cannot increase the achromatic number,
while removing just an edge can increase it.
Put differently, the achromatic number of unsigned graphs is monotonic with respect to the induced subgraphs,
but this is not true with respect to subgraphs in general.

In the case of signed graphs, any edge can be modified in two relevant ways,
being that it can be deleted or have its sign changed.
Note that, in general, changing the sign of a single edge cannot be obtained through switching vertices. 
We prove that these two types of edge modifications affect the achromatic number of signed graphs quite similarly as in the unsigned context.
We start off by considering sign modifications.

\begin{theorem} \label{theorem:change-edge-sign}
Let $(G,\sigma)$ be a signed graph, and let $(G,\sigma')$ be a signed graph obtained from $(G,\sigma)$ by changing the sign of a single edge.
Then, we have $$\psi(G,\sigma) - 2 \le \psi(G,\sigma') \leq \psi(G,\sigma)+2.$$
Furthermore, there are contexts where both bounds can be attained.
Particularly, for every $\psi(G,\sigma) \geq 4$, both the lower bound and the upper bound can be attained.
\end{theorem}

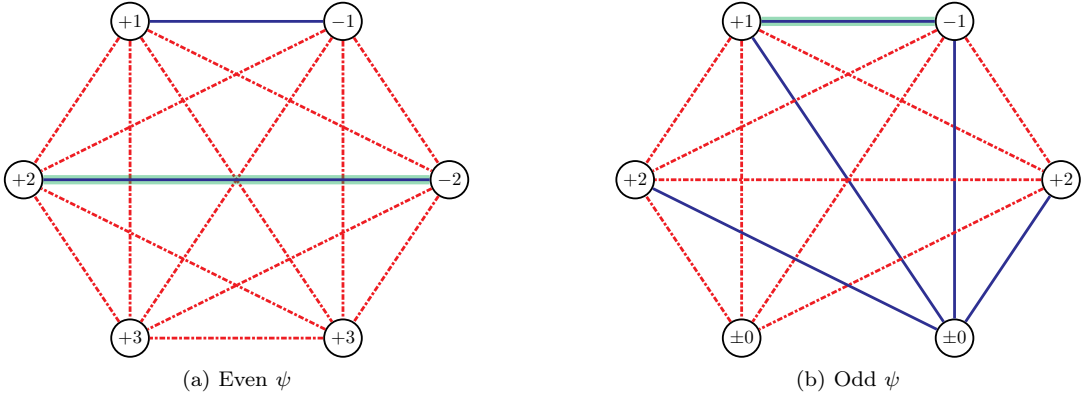
\begin{figure}[!t]
 	\centering
 	\subfloat[Even $\psi$]{
    \scalebox{0.7}{
	\begin{tikzpicture}[inner sep=0.7mm]	
	
	\node[draw, circle, black, line width=1pt](u1) at (0,0)[]{$+1$};
	\node[draw, circle, black, line width=1pt](u2) at (4,0)[]{$-1$};
	\node[draw, circle, black, line width=1pt](u3) at (6,-3)[]{$-2$};
	\node[draw, circle, black, line width=1pt](u4) at (4,-6)[]{$+3$};
	\node[draw, circle, black, line width=1pt](u5) at (0,-6)[]{$+3$};
	\node[draw, circle, black, line width=1pt](u6) at (-2,-3)[]{$+2$};
	
    \draw [-, line width=1.5pt, Red, densely dashdotted] (u1) -- (u3);
    \draw [-, line width=1.5pt, Red, densely dashdotted] (u1) -- (u4);
    \draw [-, line width=1.5pt, Red, densely dashdotted] (u1) -- (u5);
    \draw [-, line width=1.5pt, Red, densely dashdotted] (u1) -- (u6);
    \draw [-, line width=1.5pt, Red, densely dashdotted] (u2) -- (u3);
    \draw [-, line width=1.5pt, Red, densely dashdotted] (u2) -- (u4);
    \draw [-, line width=1.5pt, Red, densely dashdotted] (u2) -- (u5);
    \draw [-, line width=1.5pt, Red, densely dashdotted] (u2) -- (u6);
    \draw [-, line width=1.5pt, Red, densely dashdotted] (u3) -- (u4);
    \draw [-, line width=1.5pt, Red, densely dashdotted] (u3) -- (u5);
    \draw [-, line width=1.5pt, Red, densely dashdotted] (u4) -- (u5);
    \draw [-, line width=1.5pt, Red, densely dashdotted] (u4) -- (u6);
    \draw [-, line width=1.5pt, Red, densely dashdotted] (u5) -- (u6);
    
    \draw [-, line width=1.5pt, Blue] (u1) -- (u2);
    \draw [-, line width=5pt, Green, opacity=0.4] (u3) -- (u6);
    \draw [-, line width=1.5pt, Blue] (u3) -- (u6);
	\end{tikzpicture}
    }
    }
    \hspace{40pt}
 	\subfloat[Odd $\psi$]{
    \scalebox{0.7}{
	\begin{tikzpicture}[inner sep=0.7mm]	
	
	\node[draw, circle, black, line width=1pt](u1) at (0,0)[]{$+1$};
	\node[draw, circle, black, line width=1pt](u2) at (4,0)[]{$-1$};
	\node[draw, circle, black, line width=1pt](u3) at (6,-3)[]{$+2$};
	\node[draw, circle, black, line width=1pt](u4) at (4,-6)[]{$\pm 0$};
	\node[draw, circle, black, line width=1pt](u5) at (0,-6)[]{$\pm 0$};
	\node[draw, circle, black, line width=1pt](u6) at (-2,-3)[]{$+2$};
	
    \draw [-, line width=1.5pt, Red, densely dashdotted] (u1) -- (u3);
    \draw [-, line width=1.5pt, Blue] (u1) -- (u4);
    \draw [-, line width=1.5pt, Red, densely dashdotted] (u1) -- (u5);
    \draw [-, line width=1.5pt, Red, densely dashdotted] (u1) -- (u6);
    \draw [-, line width=1.5pt, Red, densely dashdotted] (u2) -- (u3);
    \draw [-, line width=1.5pt, Blue] (u2) -- (u4);
    \draw [-, line width=1.5pt, Red, densely dashdotted] (u2) -- (u5);
    \draw [-, line width=1.5pt, Red, densely dashdotted] (u2) -- (u6);
    \draw [-, line width=1.5pt, Blue] (u3) -- (u4);
    \draw [-, line width=1.5pt, Red, densely dashdotted] (u3) -- (u5);
    \draw [-, line width=1.5pt, Red, densely dashdotted] (u3) -- (u6);
    \draw [-, line width=1.5pt, Blue] (u4) -- (u6);
    \draw [-, line width=1.5pt, Red, densely dashdotted] (u5) -- (u6);
    
    \draw [-, line width=5pt, Green, opacity=0.4] (u1) -- (u2);
    \draw [-, line width=1.5pt, Blue] (u1) -- (u2);
	\end{tikzpicture}
    }
    }

\caption{Signed graphs attesting of the tightness of the lower bound in Theorem~\ref{theorem:change-edge-sign}.
Edges highlighted in green are edges which, when having their sign changed, make the achromatic number decrease by~$2$.
In (a) is also depicted a complete $6$-colouring, and in (b) a complete $5$-colouring.
Dashed red edges are negative edges, while solid blue edges are positive edges.  
\label{figure:resign-edge}}
\end{figure}

\begin{proof}
The claimed bounds can be established through applying Theorem~\ref{theorem:removing-vertex-signed} twice.
Indeed, denote by $v$ any one end of $e$.
By Theorem~\ref{theorem:removing-vertex-signed}, we have $\psi(G,\sigma')-2 \leq \psi(G-v,\sigma')$,
while, by Corollary~\ref{corollary:achromatic-signed-monotonous-subgraphs}, we have $\psi(G-v,\sigma')=\psi(G-v,\sigma) \leq \psi(G,\sigma)$, 
from which we deduce the claimed upper bound.
The lower bound follows from the same operations by inverting $\si$ and $\si'$.

We now prove the tightness of the lower bound (the upper bound following similarly).
That is, for every $k \geq 4$, we give a signed graph $(G,\sigma)$ such that $\psi(G,\sigma)=k$,
and $(G,\sigma)$ has an edge which, when changing its sign, results in another signed graph $(G,\sigma')$ with $\psi(G,\sigma')=k-2$.
We give two possible constructions, depicted in Figure~\ref{figure:resign-edge}, depending on the parity of $k$.

\begin{itemize}
    \item We first provide a construction for even values of $k \ge 4$. Let $p = k/2$. Consider the signed graph $(G, \sigma)$ constructed as follows. We start from a signed complete graph $(K,-)$ on $k$ vertices in which all edges are negative. Consider now a matching $M = \{u_1v_1, \dots, u_{p-1}v_{p-1}\}$ of size $p-1$ in $(K,-)$. To obtain $(G,\si)$, we just turn to positive the $p-1$ edges of $M$, and let the other edges of $(K,-)$ negative.
    
    We first prove  that $\psi(G,\sigma)=k$.
    Let us denote by $u_p$ and $v_p$ the two vertices of $G$ that are not covered by $M$. 
    Since $(G, \si)$ can be obtained from a negative complete graph, which has achromatic number $2$ by Theorem~\ref{theorem:kn-negative}, by changing the sign of $p-1$ of its edges, 
    due to the upper bound of the current theorem we have proved earlier we deduce that $\psi(G, \si) \le 2p=k$. 
    Now, consider the colouring $\phi$ of $(G,\sigma)$ assigning colour $+i$ to every vertex $u_i$ and  colour $-i$ to every vertex $v_i$ for $1 \le i \le p-1$, 
    and colour $+p$ to both $u_p$ and $v_p$. This colouring is a complete $k$-colouring of $(G, \si)$, as, for every $1 \le i < j \le p$, the edge $v_iu_j$ is a p-edge of type $(i,j)$ and the edge $u_iu_j$ is an n-edge of type $(i,j)$, and, for every $1 \le i \le p$, the edge $u_iv_i$ is an n-edge of type $(i,i)$ (in particular, no edge is a p-edge of type $(i,i)$).
	Thus, $\psi(G,\sigma)=k$ due to the upper bound proved earlier.
	
	Finally, we note that for $(G,\sigma')$, the signed graph obtained from $(G,\sigma)$ by changing the sign of $u_1v_1$,
	the upper bound we have proved earlier gives that $\psi(G, \si') \le k-2$,
	since $(G,\sigma')$ can be obtained from a negative complete graph, which has achromatic number $2$ by Theorem~\ref{theorem:kn-negative}, 
	by changing the sign of $p-2$ of its edges.
	Thus, we obtain our desired conclusion.

    \item We now consider odd values of $k \geq 5$, the construction we provide being almost the same as in the even case.
    Let $p = \lfloor k/2 \rfloor$.
    We obtain $(G,\sigma)$ as follows.
    We start from $(K,-)$, the negative complete graph on $k+1$ vertices, in which we consider a matching $M = \{u_0v_0, \dots, u_{p-1}v_{p-1}\}$ of size $p$.
    We denote by $u_p$ and $v_p$ the two vertices of $(K,-)$ that are not in $M$. 
    To obtain $(G, \si)$, we remove $u_0v_0$ from $(K,-)$, and turn to positive all other edges of $M$ (and keeping all other edges negative). 
    
    We claim that $\psi(G,\sigma)=k$.
    To see this is true, observe first that $(G,\sigma)$ can be obtained from a negative complete graph with an edge removed, 
    which has achromatic number $3$ by Theorem~\ref{theorem:kn-negative}, 
    by changing the sign of $p-1$ edges. 
    By the upper bound we have proved earlier, we get that $\psi(G, \si) \le 3 + 2(p-1) = 2p+1 = k$. 
    Now, to obtain a complete $k$-colouring of $(G,\sigma)$, or, rather, of an equivalent signed graph $(G,\sigma')$, one can first switch $u_0$, and then consider the colouring of $(G,\sigma')$ assigning colour $\pm 0$ to both $u_0$ and $v_0$, colour $+i$ to $u_i$ and colour $-i$ to $v_i$ for every $1 \le i \le p-1$, 
    and colour $+p$ to both $u_p$ and $v_p$.
    This is indeed a complete $k$-colouring for the same reasons as in the previous case for even values of $k$, and, also, because of all the edges incident to $u_p$ and $v_p$ (which provide all required n-edges and p-edges of type $(0,i)$ for $i \in \{1,\dots,p-1\}$), and the edge $u_0v_0$ is not part of $G$ here (thus there are neither n-edges nor p-edges of type $(0,0)$).
    
    Finally, note that $(G,\sigma')$, the signed graph obtained from $(G,\sigma)$ by changing the sign of $u_1v_1$,
    is obtained from a negative complete graph with an edge removed by changing the sign of exactly $p-2$ edges.
    By arguments used earlier, and in particular by the upper bound proved earlier, we deduce that $\psi(G, \si') \le k-2$.
    This concludes the whole proof.\qedhere
\end{itemize}
\end{proof}

Let us mention that our proof of Theorem~\ref{theorem:change-edge-sign} can be used to prove a similar result for the second type of edge modifications mentioned earlier.
Precisely, it can be used to prove that removing an edge in a signed graph can, as a consequence, alter the achromatic number by $\pm 2$ at most.
The fact that this can be proved in this way is because, in a signed graph,
removing or changing the sign of a given edge is, in essence, similar to removing a vertex and putting it back with different adjacencies and incidences.

\begin{theorem}\label{theorem:removing-edge-signed}
Let $(G, \sigma)$ be a signed graph.
For every $e \in E(G)$, we have
$$ \psi(G,\sigma) -2 \leq \psi(G-e,\sigma) \leq \psi(G,\sigma) +2.$$
Furthermore, there are contexts where both bounds can be attained.
Particularly, the lower bound can be attained for every $\psi(G,\sigma) \geq 5$, and the upper bound can be attained for every $\psi(G,\sigma) \geq 6$.
\end{theorem}

\begin{figure}[!t]
 	\centering
 	\subfloat[Even $\psi$]{
    \scalebox{0.7}{
	\begin{tikzpicture}[inner sep=0.7mm]	
	
	\node[draw, circle, black, line width=1pt](u1) at (0,0)[]{$+1$};
	\node[draw, circle, black, line width=1pt](u2) at (0,-3)[]{$+2$};
	\node[draw, circle, black, line width=1pt](u3) at (0,-6)[]{$+3$};
	\node[draw, circle, black, line width=1pt](v1) at (3,0)[]{$+1$};
	\node[draw, circle, black, line width=1pt](v2) at (3,-3)[]{$+2$};
	\node[draw, circle, black, line width=1pt](v3) at (3,-6)[]{$+3$};
	
    \draw [-, line width=1.5pt, Red, densely dashdotted] (u1) -- (v1);
    \draw [-, line width=1.5pt, Red, densely dashdotted] (u2) -- (v2);
    \draw [-, line width=1.5pt, Red, densely dashdotted] (u3) -- (v3);
    
    \draw [-, line width=1.5pt, Red, densely dashdotted] (u1) -- (u2);
    \draw [-, line width=1.5pt, Red, densely dashdotted] (u2) -- (u3);
    
    \draw [-, line width=1.5pt, Blue] (v1) -- (v2);
    \draw [-, line width=1.5pt, Blue] (v2) -- (v3);
    
    \draw [-, line width=1.5pt, Red, densely dashdotted] (u1) to[out=180,in=180,bend right=60] (u3);
    \draw [-, line width=5pt, Green, opacity=0.4] (v1) to[out=0,in=0,bend left=60] (v3);
    \draw [-, line width=1.5pt, Blue] (v1) to[out=0,in=0,bend left=60] (v3);
	\end{tikzpicture}
    }
    }
    \hspace{40pt}
 	\subfloat[Odd $\psi$]{
    \scalebox{0.7}{
	\begin{tikzpicture}[inner sep=0.7mm]	
	
	\node[draw, circle, black, line width=1pt](u1) at (0,0)[]{$+1$};
	\node[draw, circle, black, line width=1pt](u2) at (4,0)[]{$+2$};
	\node[draw, circle, black, line width=1pt](u3) at (6,-3)[]{$\pm 0$};
	\node[draw, circle, black, line width=1pt](u4) at (4,-6)[]{$\pm 0$};
	\node[draw, circle, black, line width=1pt](u5) at (0,-6)[]{$-2$};
	\node[draw, circle, black, line width=1pt](u6) at (-2,-3)[]{$+1$};
	
    \draw [-, line width=1.5pt, Red, densely dashdotted] (u1) -- (u2);
    \draw [-, line width=1.5pt, Red, densely dashdotted] (u1) -- (u3);
    \draw [-, line width=1.5pt, Red, densely dashdotted] (u1) -- (u5);
    \draw [-, line width=1.5pt, Red, densely dashdotted] (u1) -- (u6);
    \draw [-, line width=1.5pt, Red, densely dashdotted] (u2) -- (u3);
    \draw [-, line width=1.5pt, Red, densely dashdotted] (u2) -- (u6);
    \draw [-, line width=1.5pt, Red, densely dashdotted] (u3) -- (u5);
    \draw [-, line width=1.5pt, Red, densely dashdotted] (u3) -- (u6);
    \draw [-, line width=1.5pt, Red, densely dashdotted] (u5) -- (u6);
    
    \draw [-, line width=1.5pt, Blue] (u1) -- (u4);
    \draw [-, line width=1.5pt, Blue] (u2) -- (u4);
    \draw [-, line width=1.5pt, Blue] (u5) -- (u4);
    \draw [-, line width=1.5pt, Blue] (u6) -- (u4);
    
    \draw [-, line width=5pt, Green, opacity=0.4] (u2) -- (u5);
    \draw [-, line width=1.5pt, Blue] (u2) -- (u5);
	\end{tikzpicture}
    }
    }

\caption{Signed graphs attesting of the tightness of the lower bound in Theorem~\ref{theorem:removing-edge-signed}.
Edges highlighted in green are edges which, when removed, make the achromatic number decrease by~$2$.
In (a) is also depicted a complete $6$-colouring, and in (b) a complete $5$-colouring.
Dashed red edges are negative edges, while solid blue edges are positive edges.  
\label{figure:remove-edge-lower-bound}}
\end{figure}
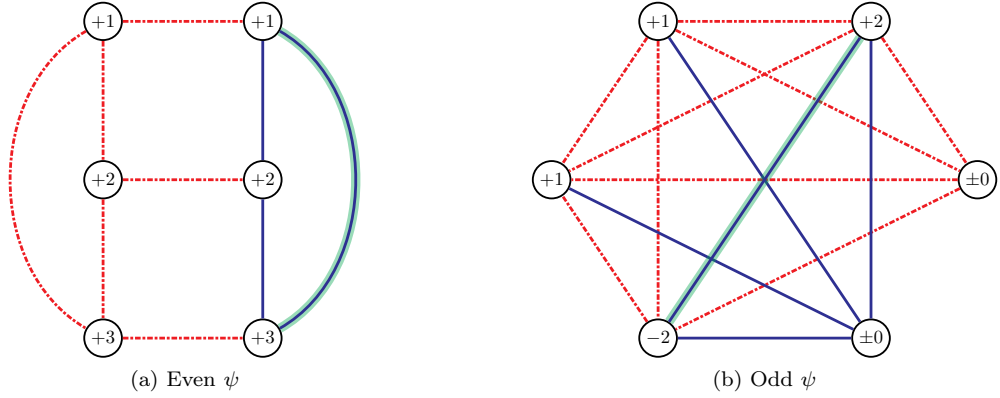

\begin{proof}
The bounds can be proved similarly as in Theorem~\ref{theorem:change-edge-sign}. 
Indeed, denote by $v$ any one end of $e$.
By Theorem~\ref{theorem:removing-vertex-signed}, we have $\psi(G-e,\sigma)-2 \leq \psi(G-e-v,\sigma)$,
while, by Corollary~\ref{corollary:achromatic-signed-monotonous-subgraphs}, we have $\psi(G-e-v,\sigma)=\psi(G-v,\sigma) \leq \psi(G,\sigma)$, 
from which we deduce the claimed upper bound.
The lower bound follows from the fact, still by Theorem~\ref{theorem:removing-vertex-signed}, 
that $\psi(G,\sigma)-2 \leq \psi(G-v,\sigma)=\psi(G-e-v,\sigma) \leq \psi(G-e,\sigma)$.

We now focus on proving the tightness of the bounds.
We start off by considering the lower bound.
For every $k \geq 5$, we provide a signed graph $(G,\sigma)$ such that $\psi(G,\sigma)=k$, and which has an edge that can be removed to decrease the achromatic number by exactly~$2$.
We give two possible constructions, depicted in Figure~\ref{figure:remove-edge-lower-bound}, applying depending on the parity of $k$.

\begin{itemize}
	\item We first provide a construction for even values of $k \geq 6$.
	Set $p=k/2$.
	Note that $p \geq 3$.
	Consider the signed graph $(G,\sigma)$ constructed as follows.
	We start from the disjoint union of the negative complete graph $K^-=(K_p,-)$ on $p$ vertices,
	and of the positive complete graph $K^+=(K_p,+)$ on $p$ vertices.
	Set $V(K^-)=\{u_1,\dots,u_p\}$ and $V(K^+)=\{v_1,\dots,v_p\}$.
	We then add the negative edge $u_iv_i$ for every $i \in \{1,\dots,p\}$,
	to achieve the construction of $(G,\sigma)$.
	
	We claim that $\psi(G,\sigma)=k$, while $\psi(G-v_1v_p,\sigma)=k-2$.
	To see that $\psi(G,\sigma)=k$, note that assigning colour~$+i$ to vertices $u_i$ and $v_i$ for every $i \in \{1,\dots,p\}$
	results in a complete $k$-colouring (the $u_iv_i$'s yielding all n-edges of type $(i,i)$, the $u_iu_j$'s yielding all n-edges of type $(i,j)$ for $i \neq j$, and the $v_iv_j$'s providing all p-edges of type $(i,j)$ for $i \neq j$, while there are no p-edges of type $(i,i)$).
	This is indeed the best we can hope for, as attested by Proposition~\ref{proposition:upper-bound-n}.
	So let us now focus on $\psi(G',\sigma)$, where $G'=G-v_1v_p$.
	By Proposition~\ref{proposition:upper-bound-size}, note that $\psi(G',\sigma) \leq k-1$.
	Due to the lower bound we have proved earlier, it suffices to show that $(G',\sigma)$ admits no complete $(k-1)$-colouring.
	Towards a contradiction, assume %$(G',\sigma)$
	vertices of $(G',\sigma)$ can be switched, so that the resulting signed graph,
	which we still call $(G',\sigma)$ below, does 
	admit such a colouring $\phi$.
	
	Since the maximum degree of $G'$ is precisely~$p$, note that, 
	in order to have a p-edge of type $(0,i)$ and an n-edge of type $(0,i)$ for every $i \in \{1,\dots,p-1\}$,
	colour~$\pm 0$ must be assigned to at least two vertices of $(G',\sigma)$ by $\phi$.
	Actually, due to Observation~\ref{observation:large-matching}, note that, for every $i \in V(K_{k-1}^*)$, there must be exactly two vertices of $(G',\sigma)$ with colour $\pm i$, and these two vertices must be adjacent if $i \neq 0$.
	Because there should not be a p-edge or an n-edge of type $(0,0)$,
	at least one of the two vertices with colour $\pm 0$ must belong to $K^+$.
	If these two vertices assigned colour~$\pm 0$ belong to $K^+$, then they must not be adjacent, and they are thus $v_1$ and $v_p$.
	We note, however, that the total number of neighbours of $v_1$ and $v_p$ in $G'$ is at most $p$, while, in $K^*_{k-1}=K^*_{2p-1}$, vertex $0$ has degree $2p-2$ (since $p \geq 3$). Thus, $\phi$ cannot realise all desired n-edges and p-edges of type $(0,i)$ for $i \neq 0$, a contradiction.
	%Then, because $K^+$ has $p$ vertices, there must be an $i \in \{1,\dots,p-1\}$ such that no vertex of $V(K^+)$ is assigned colour $\pm i$.
	%W.l.o.g., assume $i=1$.	
	%So that $(G',\sigma)$ has a p-edge of type $(0,1)$ and an n-edge of type $(0,1)$,
	%the vertices $u_1$ and $u_p$ must be assigned colour~$\pm 1$ by $\phi$. 
	%Furthermore, due to the edge $u_1u_p$ in $(G',\sigma)$ being the only n-edge of type %$(1,1)$, either $\phi(u_1)=\phi(u_p)$ and $u_1u_p$ is negative,
	%or $\phi(u_1)=-\phi(u_p)$ and $u_1u_p$ is positive.	
	%Now, because no vertex of $V(K^+)$ is assigned colour~$\pm 1$,
	%so that there is a p-edge and an n-edge of type $(1,i)$ for all $i \in \{2,\dots,p-1\}$,
	%all vertices in $V(K^-) \setminus \{u_1,u_p\}$ must be assigned colours with distinct %absolute values in $\{2,\dots, p-1\}$.
	%Assuming, now, for instance, that $\phi(u_2)=\pm 2$, so that $(G',\sigma)$ has a p-edge and an n-edge of type $(1,2)$,
	%either $u_1u_p$ is negative and $u_1u_2$ and $u_pu_2$ have different signs (case where %$\phi(u_1)=\phi(u_p)$),
	%or $u_1u_p$ is positive and $u_1u_2$ and $u_pu_2$ have the same sign (case where %$\phi(u_1)=-\phi(u_p)$).
	%In all cases, we deduce that the cycle $(u_1,u_2,u_p,u_1)$ has even balance, while, before switching vertices in $(G',\sigma)$, it had odd balance;
	%this is impossible by Lemma~\ref{lemma:switching-balance}.
	
	So, assume now that only one vertex of $K^+$ is assigned colour~$\pm 0$ by $\phi$,
	while a non-adjacent vertex $u_{i_0}$ in $K^-$ is the second vertex assigned colour~$\pm 0$.
	If there is an $i \in \{1,\dots,p-1\}$ such that the two vertices $u_{i_1}$ and $u_{i_2}$ with colour $\pm i$ belong to $K^-$, then, so that we do not run into the same previous contradiction (i.e., deduce that $(u_{i_0},u_{i_1},u_{i_2},u_{i_0})$ is a triangle with even balance), it must be that $v_{i_1}$ or $v_{i_2}$ is the vertex with colour $\pm 0$ in $K^+$. Now, because $u_{i_0}$ is assigned colour~$\pm 0$ and $u_{i_1}$ and $u_{i_2}$ are assigned colour $\pm i$, there must be a $j \in \{1,\dots,p-1\} \setminus \{i\}$ such that no vertex of $K^-$ is assigned colour $\pm j$. This means that the two vertices with colour $\pm j$ belong to $K^+$. In that case, since no vertex of $K^+$ can be assigned colour $\pm i$, we deduce that the p-edge and n-edge of type $(i,j)$ must be $u_{i_1}v_{i_1}$ and $u_{i_2}v_{i_2}$, which is not possible since one of $v_{i_1}$ and $v_{i_2}$ is assigned colour $\pm 0$.
	
	From these arguments, we deduce that, for every $i \in \{1,\dots,p-1\}$, there must be an edge $u_jv_j$ such that $u_j$ and $v_j$ are the only two vertices assigned colour $\pm i$. From this, we deduce further that there must be a $j$ such that both ends of $u_jv_j$ are assigned colour $\pm 0$, which is a contradiction to the properness of $\phi$.
	
	\item We now prove a similar result for odd values of $k \geq 5$.
	Set $p=\lfloor k/2 \rfloor$, and consider the following construction.
	We start from a complete graph $K_{2p+2}$ on $2p+2$ vertices, in which all edges are negative.
	In this signed graph, we then consider a matching $M=\{e_1,\dots,e_{p+1}\}$ of size $p+1$,
	remove $e_{p+1}=uv$ from the graph, switch, say, $u$, and finally turn to positive all other edges of $M$ but $e_1$.
	We denote by $(G,\sigma)$ the resulting signed graph.
	
	We claim that $\psi(G,\sigma)=k$ and $\psi(G-e_p,\sigma)=k-2$.
	Note that $(G,\sigma)$ is actually similar (up to switching vertices) to a signed graph considered in the proof of Theorem~\ref{theorem:change-edge-sign},
	in which we proved that indeed $\psi(G,\sigma)=k$.
	So let us now focus on $(G',\sigma)=(G-e_p,\sigma)$.
	To see that $\psi(G',\sigma)=k-2$, it suffices to note that $(G',\sigma)$ 
	is essentially obtained from a negative complete graph $K^-=(K_{2p+2},-)$ with the edges of a matching $\{e_{p+1}, e_p\}$ removed, by switching $u$ and turning  $p-2$ other edges of $M$ into positives ones (note that only $e_1$ is neither removed nor has its sign changed).
	Due to Theorem~\ref{theorem:kn-negative}, recall that $\psi(K^-)=3$.
	Now, by applying Theorem~\ref{theorem:change-edge-sign} $p-2$ times, we deduce that $\psi(G',\sigma) \leq 3 + 2(p-2)=2p-1=k-2$.
	Since $\psi(G,\sigma)=k$, by the lower bound we have proved earlier the equality actually holds.
\end{itemize}

\begin{figure}[!t]
 	\centering
 	\subfloat[Even $\psi$]{
    \scalebox{0.7}{
	\begin{tikzpicture}[inner sep=0.7mm]	
	
	\node[draw, circle, black, line width=1pt](u1) at (0,0)[]{$+1$};
	\node[draw, circle, black, line width=1pt](u2) at (0,-3)[]{$+1$};
	
	\node[draw, circle, black, line width=1pt](v1) at (3,3)[]{$+2$};
	\node[draw, circle, black, line width=1pt](v2) at (3,0)[]{$+2$};
	\node[draw, circle, black, line width=1pt](v3) at (3,-3)[]{$+2$};
	\node[draw, circle, black, line width=1pt](v4) at (3,-6)[]{$+2$};
	
	\node[draw, circle, black, line width=1pt](w1) at (-3,-1.5)[]{$+2$};
	\node[draw, circle, black, line width=1pt](w2) at (-6,-1.5)[]{$+2$};

    \draw [-, line width=1.5pt, Blue] (u1) -- (v2);
    \draw [-, line width=1.5pt, Blue] (u1) -- (v3);
    \draw [-, line width=1.5pt, Blue] (u2) -- (v2);
    \draw [-, line width=1.5pt, Blue] (u2) -- (v3);
    \draw [-, line width=1.5pt, Red, densely dashdotted] (u1) -- (u2);
    \draw [-, line width=1.5pt, Red, densely dashdotted] (u1) -- (w1);
    \draw [-, line width=1.5pt, Red, densely dashdotted] (u2) -- (w1);
    \draw [-, line width=1.5pt, Red, densely dashdotted] (u1) to[bend right=20] (w2);
    \draw [-, line width=1.5pt, Red, densely dashdotted] (u2) to[bend left=20] (w2);
    \draw [-, line width=1.5pt, Red, densely dashdotted] (w1) -- (w2);
    
    \draw  [-, line width=1.5pt, Blue] (u1) -- (v1);
    \draw  [-, line width=1.5pt, Blue] (u1) -- (v4);
    \draw  [-, line width=1.5pt, Blue] (u2) -- (v1);
    \draw [-, line width=5pt, Green, opacity=0.4] (u2) -- (v4);
    \draw  [-, line width=1.5pt, Blue] (u2) -- (v4);
	\end{tikzpicture}
    }
    }
    \hspace{40pt}
 	\subfloat[Odd $\psi$]{
    \scalebox{0.9}{
	\begin{tikzpicture}[inner sep=0.7mm]	
	
	\node[draw, circle, black, line width=1pt](v1) at (0,0)[]{$\pm 0$};
	\node[draw, circle, black, line width=1pt](v2) at (0,-1)[]{$\pm 0$};
	\node[draw, circle, black, line width=1pt](v3) at (0,-2)[]{$+1$};
	\node[draw, circle, black, line width=1pt](v4) at (0,-3)[]{$+1$};
	\node[draw, circle, black, line width=1pt](v5) at (0,-4)[]{$+2$};
	\node[draw, circle, black, line width=1pt](v6) at (0,-5)[]{$+2$};
	\node[draw, circle, black, line width=1pt](v7) at (0,-6)[]{$+2$};
	
	\node[draw, circle, black, line width=1pt](u1) at (-5,-1)[]{$+1$};
	\node[draw, circle, black, line width=1pt](u2) at (-5,-3)[]{$+1$};
	\node[draw, circle, black, line width=1pt](u3) at (-5,-5)[]{$+2$};
	
    \draw [-, line width=1.5pt, Red, densely dashdotted] (u1) -- (v2);
    \draw [-, line width=1.5pt, Red, densely dashdotted] (u1) -- (v3);
    \draw [-, line width=1.5pt, Red, densely dashdotted] (u1) -- (v4);
    \draw [-, line width=1.5pt, Red, densely dashdotted] (u1) -- (v5);
    \draw [-, line width=1.5pt, Red, densely dashdotted] (u1) -- (v6);
    \draw [-, line width=5pt, Green, opacity=0.4] (u1) -- (v6);    %ERIC
    \draw [-, line width=1.5pt, Red, densely dashdotted] (u1) -- (v7);
	
    \draw [-, line width=1.5pt, Red, densely dashdotted] (u2) -- (v2);
    \draw [-, line width=1.5pt, Red, densely dashdotted] (u2) -- (v3);
    \draw [-, line width=1.5pt, Red, densely dashdotted] (u2) -- (v4);
    \draw [-, line width=1.5pt, Red, densely dashdotted] (u2) -- (v5);
    \draw [-, line width=1.5pt, Red, densely dashdotted] (u2) -- (v6);
    \draw [-, line width=1.5pt, Red, densely dashdotted] (u2) -- (v7);
	
    \draw [-, line width=1.5pt, Red, densely dashdotted] (u3) -- (v2);
    \draw [-, line width=1.5pt, Red, densely dashdotted] (u3) -- (v3);
%    \draw [-, line width=5pt, Green, opacity=0.4] (u3) -- (v3);
    \draw [-, line width=1.5pt, Red, densely dashdotted] (u3) -- (v3);
    \draw [-, line width=1.5pt, Red, densely dashdotted] (u3) -- (v4);
    \draw [-, line width=1.5pt, Red, densely dashdotted] (u3) -- (v5);
    \draw [-, line width=1.5pt, Red, densely dashdotted] (u3) -- (v6);

    \draw [-, line width=1.5pt, Blue] (u1) -- (v1);
    \draw [-, line width=1.5pt, Blue] (u2) -- (v1);
    \draw [-, line width=1.5pt, Blue] (u3) -- (v1);
	\end{tikzpicture}
    }
    }

\caption{Signed graphs attesting of the tightness of the upper bound in Theorem~\ref{theorem:removing-edge-signed}.
Edges highlighted in green are edges which, when removed, make the achromatic number increase by~$2$.
In (a) is also depicted a complete $4$-colouring, and in (b) a complete $5$-colouring.
Dashed red edges are negative edges, while solid blue edges are positive edges.  
\label{figure:remove-edge-upper-bound}}
\end{figure}

We now prove the tightness of the upper bound.
Again, we provide two constructions, illustrated in Figure~\ref{figure:remove-edge-upper-bound}, to cover the possible parities for the number of involved colours.
Namely, for every $k \geq 6$, we prove that there exist signed graphs $(G,\sigma)$ with $\psi(G,\sigma)=k$,
in which an edge can be removed to increase the achromatic number by exactly~$2$.

\begin{itemize}
	\item Let first consider even values of $k \geq 6$.
	We set $p=k/2$, and consider the following recursive construction, which provides signed graphs $G_3, G_4, \dots$.
	The first signed graph, $G_3$, is the one depicted in Figure~\ref{figure:remove-edge-upper-bound}(a) (in which, compared to what is described in the following explanations,
	some vertices have been switched).
	It is obtained starting from a negative complete graph of order~$3$ with vertices $k_1,k_2,k_3$,
	then joining $k_1$ and $k_2$ to four new vertices $v_1,\dots,v_4$ through negative edges,
	and eventually joining a new vertex $w_3$ to all of $k_1,k_2,k_3$ through negative edges.
	The general construction is then as follows: assuming $G_p$ was built,
	$G_{p+1}$ is obtained by starting from $G_p$,
	adding three new vertices $w_p'$, $k_{p+1}$ and $w_{p+1}$ connected to every vertex $k_i$ with $i \in \{1,\dots,p\}$ through a negative edge,
	and eventually joining $w_{p+1}$ and $k_{p+1}$ through a negative edge. 
	We also denote, for every $i$, by $G'_i$ the signed graph obtained by removing the edge $k_1v_1$ from $G_i$.
	
	We claim that $\psi(G'_p) \geq 2p$ while $\psi(G_p) \leq 2p-2$.
	Note that, by the upper bound we have proved earlier, this implies equality for both parameters.
	Let us focus on $G'_p$ first.
	A complete $2p$-colouring of $G'_p$ can be obtained from an initial complete $6$-colouring of $G_3'$,
	by first extending it to a complete $8$-colouring of $G_4'$, then to a complete $10$-colouring of $G_5'$, and so on.	
	The initial complete $6$-colouring of $G_3'$ is
    obtained by first switching $v_1$ and $v_4$, and then assigning colour $+1$ to $v_3$, $v_4$ and $k_2$, colour $+2$ to $v_2$ and $k_1$, and colour $+3$ to $v_1$, $k_3$ and $w_3$. This colouring of $G_3'$ can indeed be checked to be complete.
	Now, assuming a complete $2p$-colouring of $G_p'$ was previously constructed,
	a complete $(2p+2)$-colouring of $G_{p+1}'$ is obtained from it by first switching $w_p'$,
	and then assigning colour $+(p+1)$ to the three new vertices $w_p'$, $k_{p+1}$ and $w_{p+1}$.
	Particularly, it can be checked that $G'_{p+1}$ has all types of edges required in a complete colouring,
	due to all the edges incident to these three vertices (which yield the desired p-edge and n-edge of type $(i,p+1)$ for every $i \in \{1,\dots,p\}$, and n-edge of type $(p+1,p+1)$, without yielding undesired types of edges).

	We now prove that $\psi(G_p) \leq 2p-2$.
	Assume this is wrong, and, towards a contradiction, assume $G_p$
	can have some of its vertices switched, so that 
	the resulting signed graph	 admits a complete $\ell$-colouring $\phi$, for some $\ell\ge 2p-1$ (which is at most $2p+2$, by the upper bound we have proved earlier).
	Free to consider signed colours instead, 
	we can assume that we have a complete $\ell$-colouring $\gamma$ of $G_p$ (in which no vertex was switched) inferred from $\phi$.
	Let us focus on the two vertices $k_1$ and $k_2$.
	By construction, they are universal, and, as a result, they are adjacent.
	We claim that we must have $\gamma(k_1)=\gamma(k_2)$.
	Indeed, suppose the contrary.
	First off, we note that, because $k_1k_2$ is negative, we cannot have $\{\gamma(k_1),\gamma(k_2)\}=\{i^-,i^+\}$ for some $i$,
	as otherwise $k_1k_2$ would be a p-edge of type $(i,i)$ by $\phi$.
	Assume thus that $\gamma(k_1) \in \{i^-,i^+\}$ and $\gamma(k_2) \in \{j^-,j^+\}$ for some $i \neq j$.
	We may also assume that $0 \not \in \{i,j\}$, as, if, say, $i=0$ and $j=1$, then, because $k_1$ and $k_2$ are universal, any other vertex $x$ of $G_p$ assigned colour $\pm \phi(k_2)$ would have to have $\gamma(x)=\gamma(k_2)$ (so that we do not have a p-edge of type $(1,1)$), and it would thus be impossible to have a p-edge of type $(0,1)$ if $k_1k_2$ is an n-edge of type $(0,1)$, and \textit{vice versa}.
	Thus, $0 \not \in \{i,j\}$.
	By Lemma~\ref{lemma:inferred-force-colours},
	we can assume that $\gamma(k_1)=i^+$ and $\gamma(k_2)=j^+$.
	Thus, $k_1k_2$ is an n-edge of type $(i,j)$.
	So that $\gamma$ is complete, there must be also a p-edge of type $(i,j)$.
	Since $\gamma(k_1)=i^+$ and $k_1$ is connected to every vertex of $G_p$ through a negative edge,
	note that every other vertex with a colour in $\{i^-,i^+\}$ by $\gamma$ must be assigned colour~$i^+$.
	Similarly, due to $k_2$, every other vertex with a colour in $\{j^-,j^+\}$ must be assigned colour~$j^+$.
	Because all edges of $G_p$ are negative, so that a p-edge of type $(i,j)$ exists,
	some vertex must be assigned a colour in $\{i^-,j^-\}$, which is thus not possible, and the completeness of $\gamma$ is contradicted.
	
	Thus, $\gamma(k_1)=\gamma(k_2)$. Since $k_1k_2$ is an edge, we cannot have $\gamma(k_1),\gamma(k_2) \in \{0^-,0^+\}$.
	By Lemma~\ref{lemma:inferred-force-colours}, we may assume, w.l.o.g., that $\gamma(k_1)=\gamma(k_2)=1^+$.
	By Observation~\ref{observation:large-matching}, for every $i \in \{2,\dots,\lfloor\frac{\ell}{2}\rfloor\}$ there must be an edge with ends having colours $\pm i$ by $\phi$.
	Note that the signed graph $G_p-k_1-k_2$ forms isolated vertices (the $v_i$'s) and a single connected component with a dominating clique on $p-2$ vertices (the other vertices).
	Thus, in particular, for each $i \in \{2,\dots,p-1\}$ (recall that $p-1\le\lfloor\frac{\ell}{2}\rfloor$), a vertex of that clique must be assigned colour $\pm i$ by $\phi$.
	Particularly, focus on $k_p$; assume its colour is $\pm i$ by $\phi$.
	Since the only neighbour of $k_p$ not in the clique is $w_p$, we have $\phi(w_p)=\pm i$.
	But, now, we have that the only neighbours of $k_p$ and $w_p$, in $G_p$, are the vertices of the clique,
	in which no vertex is assigned colour~$\pm 0$ (if $\ell$ is odd) or $\pm p$ (if $\ell$ is even).
	Precisely, we note that only the $v_i$'s, which form an independent set, can be assigned colour $\pm 0$ or $\pm p$.
	Then we deduce that there cannot be a p-edge nor an n-edge of type $(0,i)$ or $(p,i)$, and $\phi$ cannot be complete, a contradiction.

	\item We finally prove a similar result for odd $k \geq 7$.
	Set $p=\lfloor k/2 \rfloor$.
	The signed graph $(G,\sigma)$ we consider here,
	is obtained simply from a complete bipartite graph $K_{p,2p+1}$ with all edges negative,
	by removing a matching of size $p-2$.
	We denote by $u_1,\dots,u_p$ and $v_1,\dots,v_{2p+1}$ the vertices from the two partite sets of $G$,
	and let $\{u_2v_2,\dots,u_{p-1}v_{p-1}\}$ be the removed matching.
	We also set $G'=G-u_pv_p$.
	
	We claim that $\psi(G,\sigma) \leq 2p-1$ while $\psi(G',\sigma) \geq 2p+1$, 
	from which we get the desired result since equality holds in both cases due to the upper bound we have proved earlier.
	Consider $(G,\sigma)$ first.
	Note that the maximum size of a matching in $G$ is $p$, and so, by Observation~\ref{observation:large-matching}, we have $\psi(G,\sigma) \leq 2p+1$.
	Towards a contradiction to the inequality $\psi(G,\sigma) \leq 2p-1$, suppose $(G,\sigma)$ can have some of its vertices switched, so that the resulting signed graph, that we still denote by $(G,\sigma)$ for convenience,
admits a complete $\ell$-colouring $\phi$ with $\ell\in\{2p,2p+1\}$.
    Free to consider signed colours, we consider $\gamma$, the $\ell$-colouring inferred from $\phi$, directly in $(G,\sigma)$ (i.e., with all its edges negative).
    By Observation~\ref{observation:large-matching}, for every $i \in \{1,\dots,p\}$,
	there must be an n-edge of type $(i,i)$, and those $p$ edges must form a matching.
	Because the partite set of $G$ containing the $u_i$'s has cardinality~$p$,
	the $u_i$'s must thus be assigned colours with distinct absolute values.
	By Lemma~\ref{lemma:inferred-force-colours}, we can assume, w.l.o.g., that $\gamma(u_1)=\gamma(v_1)=1^+$ and that $\gamma(u_p)=\gamma(v_p)=p^+$. Note that this makes both edges $u_1v_p$ and $u_pv_1$ be n-edges of type $(1,p)$. To also have the desired p-edge of type $(1,p)$, there must thus be other $v_i$'s being assigned colour $\pm 1$ (or, similarly, $\pm p$) by $\phi$. However, so that $u_1$ together with such $v_i$'s do not yield a p-edge of type $(1,1)$ (or $(p,p)$), every such $v_i$ must be assigned colour $1^+$ (or $p^+$) by $\gamma$. From this, we deduce that there cannot be any p-edge of type $(1,p)$ by $\phi$, a contradiction.
	
	Now consider $(G',\sigma)$.
	To see that $(G',\sigma)$ has achromatic number at least $2p+1$, first switch $v_{2p+1}$, and then
	consider the complete $(2p+1)$-colouring $\phi$ assigning colour $+i$ to $u_i$ for every $i \in \{1,\dots,p\}$, colour $+1$ to $v_1$, colour $-i$ to $v_i$ for every $i \in \{2,\dots,p\}$, colour $+i$ to $v_{p+i-1}$ for every $i \in \{2,\dots,p\}$, and colour $\pm 0$ to both $v_{2p}$ and $v_{2p+1}$.
	To see that this colouring is indeed complete, remark that there are n-edges of type $(i,i)$ for every $i \in \{1,\dots,p\}$ ($u_1v_1$ for $i=1$, and $u_iv_{p+i-1}$ for $i>1$), n-edges of type $(i,j)$ for every $0<i<j \leq p$ ($u_iv_{p+j-1}$), p-edges of type $(i,j)$ for every $0<i<j \leq p$ ($u_iv_j$), and n-edges and p-edges of type $(0,i)$ for every $i \in \{1,\dots,p\}$ ($u_iv_{2p}$ and $u_iv_{2p+1}$).
	Note also that, apart from the edges incident to $v_{2p+1}$, all edges of the signed graph are negative, which means that any p-edge of type $(i,i)$ (with $i \neq 0$) would be an edge $u_iv_j$ with $\phi(u_i)=-\phi(v_j)$. Such an edge is not present in $G'$, by construction. Similarly, only $v_{2p}$ and $v_{2p+1}$ are assigned colour $\pm 0$, so there are neither n-edges nor p-edges of type $(0,0)$.	\qedhere
\end{itemize}
\end{proof}

\subsection{Homomorphisms of signed graphs} \label{subsection:homomorphisms-signed-graphs}

As mentioned earlier, graph colourings tend to be related to graph homomorphisms, in a more or less intricate way depending on the exact notions involved.
In this subsection, we explore, in the context of signed graphs, certain homomorphism notions that have been studied in connection with the achromatic number of unsigned graphs.

Recall that, for two graphs $G$ and $H$, a homomorphism $h: V(G) \rightarrow V(H)$ of $G$ to $H$ is a vertex-mapping that preserves the edges,
i.e., $h(u)h(v)$ is an edge in $H$ whenever $uv$ is an edge in $G$.
The \textit{homomorphic image} $h(G)$ of $G$ (by $h$) is the subgraph of $H$ to which $G$ gets mapped, through $h$, 
i.e., $V(h(G))=\{h(v):\ v\in V(G)\}$ and $E(h(G))=\{h(u)h(v):\ uv\in E(G)\}$.
We say that $h$ is \textit{elementary} if the homomorphic image of $G$ by $h$ can be obtained from $G$ by identifying two non-adjacent vertices.
Note that any homomorphism can be decomposed into a sequence of elementary homomorphisms.
An \textit{elementary homomorphic image} $G$ is a homomorphic image obtained through an elementary homomorphism.

In the context of complete colourings and the achromatic number of unsigned graphs,
the connection with homomorphic images was investigated notably by Harary, Hedetniemi and Prins.
More precisely, they proved the following:

\begin{theorem}[Harary, Hedetniemi, Prins~\cite{Harary1967}] 
Let $G$ be a graph, and $H$ be a homomorphic image of $G$. Then, we have
$$\chi(G) \le \chi(H) \le \psi(H) \le \psi(G).$$
\end{theorem}

These bounds were made more precise, in the context of elementary homomorphisms:

\begin{theorem}[Harary, Hedetniemi \cite{Harary1970}] \label{theorem:elementary-unsigned}
Let $G$ be a graph, and $H$ be an elementary homomorphic image of $G$. Then, we have
\begin{itemize}
	\item $\chi(G) \le \chi(H) \le \chi(G) +1$, and
	\item $\psi(G) -2 \le \psi(H) \le \psi(G)$.
\end{itemize}
\end{theorem}

In what follows, our main goal is to generalise the two items in Theorem~\ref{theorem:elementary-unsigned} to our notions in signed graphs.
Recall that a homomorphism $h$ of a signed graph $(G,\sigma)$ to a signed graph $(H,\pi)$ is a vertex-mapping $h$ preserving both the edges and their signs.
As above, the \textit{homomorphic image} of $(G,\sigma)$ (by $h$) is the signed subgraph of $(H,\pi)$ induced by the subset of vertices to which the vertices of $(G,\sigma)$ get mapped.
To introduce a notion of elementary homomorphisms of signed graphs, we first need to make clear which pairs of vertices can be identified.
Note indeed that it would not make any sense identifying two vertices $u$ and $v$ with a common neighbour $w$ such that the sign of $uw$ is not the same as that of $vw$.
Consequently, we say that two vertices $u$ and $v$ of $(G,\sigma)$ are \textit{identifiable} if $u$ and $v$ are not adjacent,
and they are connected (up to possibly switching $u$ or $v$) in a similar way to each of their common neighbours (i.e., $\sigma(uw)=\sigma(vw)$ for every common neighbour $w$ of $u$ and $v$).
Now, we say that $h$ is \textit{elementary} if it identifies 
two identifiable vertices of $(G,\sigma)$,
and, in that case, the homomorphic image of $(G,\sigma)$ is called an \textit{elementary homomorphic image}.
%(that is, two vertices being not adjacent and being connected to their common neighbours via edges of the same type
%not belonging to a negative 4-cycle),

We start by adapting the first item of Theorem~\ref{theorem:elementary-unsigned} to those notions.
The next result we give is actually the only one, 
out of all results given in this paper, 
for which the bounds we provide are exactly the same as in the unsigned case.

\begin{theorem}\label{theorem:elementary-signed-1}
Let $(G,\si)$ be a signed graph, and $(H,\pi)$ be an elementary homomorphic image of $(G,\sigma)$, Then, we have $$\chi(G, \sigma) \leq  \chi(H,\pi) \leq \chi(G,\sigma) +1.$$
\end{theorem}

\begin{proof}
To avoid any confusion, let us emphasise that the statement is about proper colourings, not complete colourings.
Now, let $u$ and $v$ denote the two non-adjacent vertices of $(G,\sigma)$ which got identified by the homomorphism to result in $(H,\pi)$.
Let $w$ denote the vertex of $(H,\pi)$ resulting from the identification.
We note first that $\chi(G, \sigma) \leq \chi(H,\pi)$.
Indeed, if $\phi$ is a proper $k$-colouring of $(H,\pi)$, then note that it extends naturally to $(G,\sigma)$ by assigning colour $\phi(w)$ to both $u$ and $v$.
Particularly, the properness of the so-obtained $k$-colouring of $(G,\sigma)$ results from the fact that $(H,\pi)$ is a homomorphic image of $(G,\sigma)$.
We prove now that $\chi(H,\pi) \leq \chi(G,\sigma) +1$. 
Let $\phi$ be a proper $k$-colouring of $(G,\sigma)$.
If $\phi(u)=\phi(v)$, then note that $\phi$ can be extended directly to a proper $k$-colouring of $(H,\pi)$ by simply assigning colour $\phi(u)$ to $w$.
Thus, assume $\phi(u) \neq \phi(v)$.
Now, if, on the one hand, $k$ is even, then, the $(k+1)$-colouring of $(H,\pi)$ obtained by assigning colour $\pm 0$ to $w$ (and keeping all other colours by $\phi$)
is clearly proper, since $\phi$ is proper and $w$ is the sole vertex assigned colour~$\pm 0$.
If, on the other hand, $k$ is odd, then note that, upon switching vertices of $(G,\sigma)$ being assigned colour $\pm 0$ by $\phi$ (which, recall, results in a proper $k$-colouring),
we get to a signed graph defined over $G$ that is equivalent to $(G,\sigma)$ in terms of proper colourings, and in which all edges incident to $u$ and $v$ going to vertices assigned colour $\pm 0$ are negative.
We may thus assume $(G,\sigma)$ (and, thus, $(H,\pi)$ regarding its vertex $w$) has this property.
Now, because two adjacent vertices of $(G,\sigma)$ cannot be assigned colour $\pm 0$ by $\phi$,
when extending $\phi$ to $(G,\pi)$ by assigning colour $+(k+1)$ to $w$ and all other vertices assigned colour $\pm 0$ by $\phi$ in $(G,\sigma)$,
we deduce that the resulting $(k+1)$-colouring of $(H,\pi)$ is proper, essentially because any two vertices assigned colour $+(k+1)$ are joined by a negative edge.
Thus, $\chi(H,\pi) \leq \chi(G,\sigma) +1$.
\end{proof}

Let us now focus on the second item of Theorem~\ref{theorem:elementary-unsigned}.

\begin{theorem} \label{theorem:homomorphism-signed2}
If $(H,\pi)$ is an elementary homomorphic image of a signed graph $(G,\sigma)$, then
$$ \psi(G,\sigma) -4 \leq \psi(H,\pi) \leq \psi(G,\sigma).$$ 
Furthermore, there are contexts where the lower bound can be attained.
Particularly, for every $\psi(G,\sigma) \geq 8$, the lower bound can be attained.
\end{theorem}

\begin{proof}
Assume $u$ and $v$ are the non-adjacent vertices of $(G,\sigma)$ that were identified to form $(H,\pi)$, and denote by $w$ the resulting identified vertex.
We observe first that if $\phi$ is a complete $k$-colouring of $(H,\pi)$, 
then the $k$-colouring of $(G,\sigma)$ obtained from $\phi$ by assigning colour $\phi(w)$ to both $u$ and $v$ (and keeping all other colours by $\phi$) is also complete.
This is because the edges incident to $u$ and $v$ in $(G,\sigma)$ essentially induce, by the resulting colouring, the same edge types as $w$ in $(H,\pi)$ by $\phi$.
This shows that $\psi(H,\pi) \leq \psi(G,\sigma)$.
We now focus on the bound on the left-hand side of the inequality.
Assume that $\psi(G,\sigma)=k$, and let $\phi$ be a complete $k$-colouring of $(G,\sigma)$.
Denote by $(G',\sigma)$ the signed subgraph obtained from $(G,\sigma)$ by removing all vertices assigned colour $\pm \phi(u)$ or $\pm \phi (v)$.
In $(G',\sigma)$, note that $\phi$ induces a complete $k'$-colouring with $k-4 \leq k' \leq k-1$, by Observation~\ref{observation:remove-class}.
Because $(G',\sigma)$ is an induced signed subgraph of $(H,\pi)$, by Corollary~\ref{corollary:achromatic-signed-monotonous-subgraphs} we have $ \psi(G,\sigma) -4 \leq \psi(G',\sigma)\leq \psi(H,\pi)$.

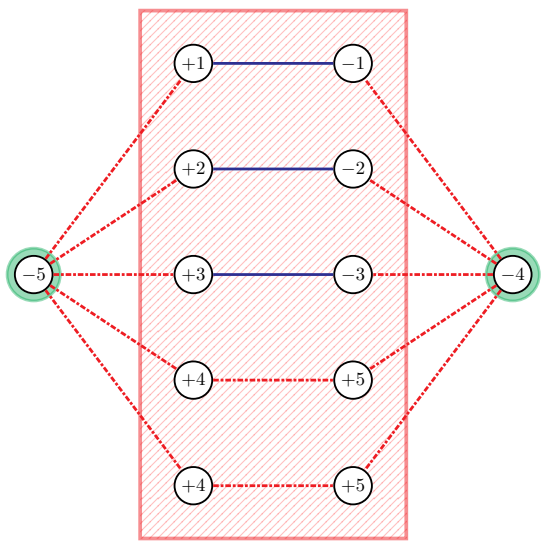
\begin{figure}[!t]
 	\centering
    \scalebox{0.7}{
	\begin{tikzpicture}[inner sep=0.7mm]	
	
	\draw [-, line width=2pt, Red, pattern=north east lines, pattern color=Red, opacity=0.5]  (-1,1) -- (4,1) -- (4,-9) -- (-1,-9) -- cycle;

	\node[draw, circle, black, line width=1pt,fill=white](u1) at (0,0)[]{$+1$};
	\node[draw, circle, black, line width=1pt,fill=white](u2) at (0,-2)[]{$+2$};
	\node[draw, circle, black, line width=1pt,fill=white](u3) at (0,-4)[]{$+3$};
	\node[draw, circle, black, line width=1pt,fill=white](u4) at (0,-6)[]{$+4$};
	\node[draw, circle, black, line width=1pt,fill=white](u5) at (0,-8)[]{$+4$};
	
	\node[draw, circle, black, line width=1pt,fill=white](v1) at (3,0)[]{$-1$};
	\node[draw, circle, black, line width=1pt,fill=white](v2) at (3,-2)[]{$-2$};
	\node[draw, circle, black, line width=1pt,fill=white](v3) at (3,-4)[]{$-3$};
	\node[draw, circle, black, line width=1pt,fill=white](v4) at (3,-6)[]{$+5$};
	\node[draw, circle, black, line width=1pt,fill=white](v5) at (3,-8)[]{$+5$};
	
	 \draw[Green, line width=1.5pt,opacity=0.4, fill] (-3,-4) circle (0.5cm);
	 \draw[Green, line width=1.5pt,opacity=0.4, fill] (6,-4) circle (0.5cm);

	\node[draw, circle, black, line width=1pt,fill=white](w1) at (-3,-4)[]{$-5$};
	\node[draw, circle, black, line width=1pt,fill=white](w2) at (6,-4)[]{$-4$};

    \draw [-, line width=1.5pt, Red, densely dashdotted] (w1) -- (u1);
    \draw [-, line width=1.5pt, Red, densely dashdotted] (w1) -- (u2);
    \draw [-, line width=1.5pt, Red, densely dashdotted] (w1) -- (u3);
    \draw [-, line width=1.5pt, Red, densely dashdotted] (w1) -- (u4);
    \draw [-, line width=1.5pt, Red, densely dashdotted] (w1) -- (u5);

    \draw [-, line width=1.5pt, Red, densely dashdotted] (w2) -- (v1);
    \draw [-, line width=1.5pt, Red, densely dashdotted] (w2) -- (v2);
    \draw [-, line width=1.5pt, Red, densely dashdotted] (w2) -- (v3);
    \draw [-, line width=1.5pt, Red, densely dashdotted] (w2) -- (v4);
    \draw [-, line width=1.5pt, Red, densely dashdotted] (w2) -- (v5);
    
    \draw  [-, line width=1.5pt, Blue] (u1) -- (v1);
    \draw  [-, line width=1.5pt, Blue] (u2) -- (v2);
    \draw  [-, line width=1.5pt, Blue] (u3) -- (v3);
    \draw [-, line width=1.5pt, Red, densely dashdotted] (u4) -- (v4);
    \draw [-, line width=1.5pt, Red, densely dashdotted] (u5) -- (v5);
	\end{tikzpicture}
    }

\caption{Part of the signed graphs attesting of the tightness of the lower bound in Theorem~\ref{theorem:homomorphism-signed2}.
Vertices highlighted in green are non-adjacent vertices which, when identified, make the achromatic number decrease by~$4$.
In the figure is also depicted a complete $10$-colouring.
Dashed red edges are negative edges, while solid blue edges are positive edges.  
\label{figure:homomorphism}}
\end{figure}

We now prove the tightness of the lower bound.
We, again, consider two constructions, the first one being illustrated in Figure~\ref{figure:homomorphism},
to deal with the two possible parities for the number $k$ of colours.
Precisely, for every $k \geq 4$, 
we prove that there is a signed graph $(G,\sigma)$ with an elementary homomorphic image $(H,\pi)$ such that $\psi(G,\sigma) \geq k+4$ and $\psi(H,\pi) \leq k$,
thus getting the desired result due to the lower bound proved earlier.

\begin{itemize}
	\item For $k=2p$ even at least~$4$, consider the following construction.
	We start from a complete graph $K_{2p+2}$ on $2p+2$ vertices with all edges negative.
	Next, consider a matching $M=\{e_1,\dots,e_{p+1}\}$ of cardinality $p+1$ in $K_{2p+2}$,
	and change to positive all the edges $e_1,\dots,e_{p-1}$.
	Finally, add two vertices $v$ and $w$, and, setting $e_i=x_iy_i$ for every $i \in \{1,\dots,p+1\}$,
	 join $v$ and all $x_i$'s through negative edges,
	and similarly join $w$ and all $y_i$'s through negative edges.
	The resulting signed graph is $(G,\sigma)$, 
	while the signed graph $(H,\pi)$ is the elementary homomorphic image of $(G,\sigma)$ obtained when identifying $v$ and $w$.
	
	We observe, first, that $\psi(H,\pi) \leq k$.
	This is because $(H,\pi)$ can be obtained from a negative complete graph by changing the sign of $p-1$ independent edges.
	Indeed, by Theorems~\ref{theorem:kn-negative} and~\ref{theorem:change-edge-sign}, we deduce that $\psi(H,\pi) \leq 2 + 2(p-1) = 2p = k$.
	We claim, now, that $\psi(G, \sigma) \geq k+4$.
	To see this is true, consider the $(k+4)$-colouring $\phi$ of $(G,\sigma)$ obtained as follows.
	For every positive edge $e_i=x_iy_i$ (i.e., with $i \in \{1,\dots,p-1\}$) of $M$, we assign, to the two ends $x_i$ and $y_i$ (where, recall, $x_i$ is adjacent to $v$ and $y_i$ is adjacent to $w$), colour $+i$ to $x_i$ and colour $-i$ to $y_i$.
	For each $e_i=x_iy_i$ of the two negative edges $e_p$ and $e_{p+1}$, we assign colour $+p$ to $x_i$ and colour $+(p+1)$ to $y_i$.
	Finally, we assign colour $-(p+1)$ to $v$ and colour $-p$ to $w$.
	It can be checked that this results in $\phi$ being complete, all required edge types being realised by $\phi$. Indeed, for every $i \in \{1,\dots,p-1\}$ the n-edge of type $(i,i)$ is realised through the edge $e_i$, while the n-edge of type $(p,p)$ (n-edge of type $(p+1,p+1)$, respectively) is realised through the edge $x_px_{p+1}$ ($y_py_{p+1}$, respectively). Furthermore, for every $i<j \le p$, the n-edge of type $(i,j)$ is realised through the edge $x_ix_j$, and the p-edge of type $(i,j)$ is realised through the edge $x_jy_i$. Also, for every $i \in \{1, \dots, p\}$, the n-edge of type $(i,p+1)$ (p-edge of type $(i,p+1)$, respectively) is realised through the edge $x_iy_p$ ($x_iv$, respectively). Also, due to how the colours were assigned, there are no p-edges of type $(i,i)$ for any $i \in \{1,\dots,p+1\}$.

	%Indeed, every edge $e_i$ with $1 \leq i \leq p-1$ is an n-edge of type $(i,i)$, the edges $x_px_{p+1}$ and $y_py_{p+1}$ are n-edges of type $(p,p)$ and $(p+1,p+1)$, respectively, every edge $x_ix_j$ with $i<j \leq p$ is an n-edge of type $(i,j)$, every edge $x_iy_p$ with $i \leq p$ is an n-edge of type $(i,p+1)$, every edge $x_iy_j$ with $i \leq p$ and $i<j \leq p-1$ is a p-edge of type $(i,j)$, and every edge $vx_i$ with $i \leq p$ is a p-edge of type $(i,p+1)$. Also, due to how the colours were assigned, there are no p-edges of type $(i,i)$ for any $i \in \{1,\dots,p+1\}$.
	
	\item For $k=2p+1$ odd at least~$5$, we consider a quite similar construction. Let $(G,\sigma)$ be the signed graph obtained as follows. Start from the signed graph obtained in the previous case for the parameter $k-1$ (which signed graph is well defined since $k-1$ is even at least~$4$), and, to obtain $(G,\sigma)$, just add two non-adjacent dominating vertices $x$ and $y$ joined to all other vertices through negative edges. The signed graph $(H,\pi)$ considered here is again the elementary image of $(G,\sigma)$ obtained when identifying $v$ and $w$.
	
	In particular, $(H,\pi)$ is here obtained from a negative complete graph with an edge missing by turning to positive $p-1$ independent edges, for which we again have $\psi(H,\pi) \leq 3+2(p-1)=2p+1=k$ by Theorems~\ref{theorem:kn-negative} and~\ref{theorem:change-edge-sign}.
	Now, to obtain a complete $(k+4)$-colouring of $(G,\sigma)$, it suffices to assign colour $\pm 0$ to both $x$ and $y$, and to colour the other vertices as in the previous case. It can be checked that this results in a complete colouring, as $x$ and $y$ are not adjacent (thus there are neither n-edges nor p-edges of type $(0,0)$), while, due to all the edges incident to $x$ and $y$, there is at least one n-edge and at least one p-edge of type $(0,i)$ for every $i \in \{1,\dots,p+1\}$, as desired. \qedhere
\end{itemize}
\end{proof}

Due to the connection between homomorphic images and elementary homomorphic images, 
note that the upper bound in Theorem~\ref{theorem:homomorphism-signed2} yields the following as a corollary:

\begin{corollary}
If $(H,\pi)$ is a homomorphic image of a signed graph $(G,\sigma)$, then
$$ \psi(H,\pi) \leq \psi(G,\sigma).$$
\end{corollary}

%%%%%%%%%%%%%%%%%%%%%%%%%%%%%%%%%%%%%%%%%%%%%%%%%%%%%%
%%%%%%%%%%%%%%%%%%%%%%%%%%%%%%%%%%%%%%%%%%%%%%%%%%%%%%
%%%%%%%%%%%%%%%%%%%%%%%%%%%%%%%%%%%%%%%%%%%%%%%%%%%%%%
%%%%%%%%%%%%%%%%%%%%%%%%%%%%%%%%%%%%%%%%%%%%%%%%%%%%%%
%%%%%%%%%%%%%%%%%%%%%%%%%%%%%%%%%%%%%%%%%%%%%%%%%%%%%%
%%%%%%%%%%%%%%%%%%%%%%%%%%%%%%%%%%%%%%%%%%%%%%%%%%%%%%
%%%%%%%%%%%%%%%%%%%%%%%%%%%%%%%%%%%%%%%%%%%%%%%%%%%%%%
%%%%%%%%%%%%%%%%%%%%%%%%%%%%%%%%%%%%%%%%%%%%%%%%%%%%%%
%%%%%%%%%%%%%%%%%%%%%%%%%%%%%%%%%%%%%%%%%%%%%%%%%%%%%%
%%%%%%%%%%%%%%%%%%%%%%%%%%%%%%%%%%%%%%%%%%%%%%%%%%%%%%
%%%%%%%%%%%%%%%%%%%%%%%%%%%%%%%%%%%%%%%%%%%%%%%%%%%%%%
%%%%%%%%%%%%%%%%%%%%%%%%%%%%%%%%%%%%%%%%%%%%%%%%%%%%%%
%%%%%%%%%%%%%%%%%%%%%%%%%%%%%%%%%%%%%%%%%%%%%%%%%%%%%%
%%%%%%%%%%%%%%%%%%%%%%%%%%%%%%%%%%%%%%%%%%%%%%%%%%%%%%

\section{Complexity aspects} \label{section:complexity}

In this section, we investigate the complexity of the problem of determining the achromatic number of a given signed graph.
More particularly, we focus on two more specific decision problems, being defined as

\medskip

\noindent \SAN\\
\textbf{Input:} A signed graph $(G,\sigma)$ and an integer $k \geq 1$.\\
\textbf{Question:} Do we have $\psi(G,\sigma) \geq k$?

\medskip

\noindent and

\medskip

\noindent \kSAN\\
\textbf{Input:} A signed graph $(G,\sigma)$.\\
\textbf{Question:} Do we have $\psi(G,\sigma) \geq k$?

\medskip

\noindent where, in the latter problem, $k \geq 1$ is a fixed integer.
Recall that, in the unsigned context, the two corresponding problems, \AN and \kAN,
are respectively \np-complete and linear-time solvable for every $k \geq 1$ (see~\cite{FHHM86}).
In this section, we essentially prove that these complexity results adapt to the signed context.

\subsection{Polynomial cases} \label{subsection:polynomial-cases}

We start by considering the \kSAN problem.
We show that this problem can be solved in polynomial time for every $k \geq 1$.
Quite similarly as in the unsigned case~\cite{FHHM86}, 
this follows mainly from the fact that the number of vertex-critical signed graphs $(G,\sigma)$ (i.e., signed graphs $(G,\sigma)$ with $\psi(G-v, \si) \le \psi(G, \si)$ for every vertex $v \in V(G)$) with $\psi(G,\sigma) \leq k$
is a function of $k$ only. That is:

\begin{theorem}\label{theorem:number-graphs-psi-k}
%Let $k \geq 1$. Every signed graph $(G, \sigma)$ with $\psi(G,\sigma) = k $ has an induced signed subgraph of size at most $k^2$ and achromatic number $k$.
Let $k \geq 1$. 
There is a polynomial function $f(k)$ such that every signed graph $(G, \sigma)$ with $\psi(G,\sigma) = k $ has an induced signed subgraph with achromatic number $k$ and size at most $f(k)$.
\end{theorem}

\begin{proof}
Let $(G,\sigma)$ be a signed graph such that $\psi(G,\sigma) = k$. 
Consider a complete $k$-colouring $\phi$ of  $(G,\sigma)$.
By definition, for every two colours $\pm i, \pm j$ assigned by $\phi$,
then, in $(G,\sigma)$, we must have 1) both a p-edge of type $(i,j)$ and an n-edge of type $(i,j)$ if $i \neq j$,
and 2) an n-edge of type $(i,j)$ if $i=j$ and $0 \not \in \{i,j\}$.
%Let $(G',\sigma)$ be a signed subgraph obtained from $(G,\sigma)$ by keeping only one edge for each of these possible types (and deleting all other edges).
For every such n-edge of type $(i,j)$ and every such p-edge of type $(i,j)$, choose any one edge $(G,\sigma)$ of that type, and add that edge to a set $F$. Then $|F|$ is the size of $K_k^*$, which is quadratic in $k$.
Now let $(G',\sigma)$ be the signed subgraph of $(G,\sigma)$ induced by the edges of $F$. Since $|F|$ is quadratic in $k$, the order of $(G',\sigma)$ is also quadratic in $k$, and thus the size of $(G',\sigma)$ is a quartic function of $k$. Also, due to how $(G',\sigma)$ was obtained from $(G,\sigma)$, note that $\psi(G',\sigma)=k$.
The claim now follows from the existence of $(G',\sigma)$.
\end{proof}

\begin{corollary}\label{corollary:polynomiality}
%Let $k \geq 1$ be fixed. There is a family $\mathcal{H}_k$ of signed graphs, all of which have size at most~$k^2$,
Let $k \geq 1$ be fixed. 
There is a family $\mathcal{H}_k$ of signed graphs, all of which have size at most some polynomial function $f(k)$, such that a given signed graph $(G,\sigma)$ verifies $\psi(G,\sigma) \leq k$ if and only if $(G,\sigma)$ is $\mathcal{H}_k$-free.
Thus, \kSAN can be solved in polynomial time.
\end{corollary}

\begin{proof}
The existence of $\mathcal{H}_k$ follows from Theorem~\ref{theorem:number-graphs-psi-k}.
Now, to decide whether $\psi(G,\sigma) \geq k$, it suffices to check whether $(G,\sigma)$ has a member of $\mathcal{H}_k$ as an induced subgraph.
This can be done in polynomial time, 
the size of the members of $\mathcal{H}_k$ being at most $f(k)$, which is a polynomial function of $k$ by assumption.
%the size of the members of $\mathcal{H}_k$ being at most $k^2$.
\end{proof}

\subsection{\np-completeness} \label{subsection:npc}

We now turn to the \SAN problem, which we show is \np-complete.
Note that this problem is clearly in \np, as, for a signed graph $(G,\sigma)$ and an integer $k \geq 1$,
assuming we are given both a colouring $\phi$ of the vertices of $G$ and a set $S$ of vertices of $G$,
we can check, in polynomial time, whether $\phi$ is a complete $k'$-colouring of the signed graph defined over $G$ obtained by switching $S$ in $(G,\sigma)$,
where $k' \geq k$.
In other words, the fact that an instance of the problem is positive can be attested easily.
Thus, we can narrow our attention down to proving only the \np-hardness of \SAN.

\begin{figure}[!t]
 	\centering
 	    \scalebox{0.4}{
      	\begin{tikzpicture}[inner sep=0.7mm,scale=.9]
      	
      	\draw [-, line width=2pt, Blue, pattern=dots, pattern color=Blue]  (2,0) -- (4,-2) -- (4,-4) -- (2,-6) -- (0,-6) -- (-2,-4) -- (-2,-2) -- (0,0) -- cycle;
		\node at (1,0.5) {\textcolor{Blue}{\LARGE $G$}};
		
      	\draw [-, line width=2pt, Red, pattern=north east lines, pattern color=Red, densely dashdotted]  (12,0) -- (14,-2) -- (14,-4) -- (12,-6) -- (10,-6) -- (8,-4) -- (8,-2) -- (10,0) -- cycle;
		\node at (11,0.5) {\textcolor{Red}{\LARGE $K_k$}};
		
		\draw [-, line width=1.5pt, Red, densely dashdotted]  ([xshift=5pt]4,-2) -- ([xshift=-5pt]8,-2);
		\draw [-, line width=1.5pt, Red, densely dashdotted]  ([xshift=5pt]4,-4) -- ([xshift=-5pt]8,-4);
		\draw [-, line width=1.5pt, Red, densely dashdotted]  ([xshift=5pt,yshift=-5pt]4,-2) -- ([xshift=-5pt,yshift=5pt]8,-4);
		\draw [-, line width=1.5pt, Red, densely dashdotted]  ([xshift=5pt,yshift=5pt]4,-4) -- ([xshift=-5pt,yshift=-5pt]8,-2);
		
      	\draw [-, line width=2pt, Red, pattern=north east lines, pattern color=Red, densely dashdotted]  (19,7) -- (21,5) -- (21,3) -- (19,1) -- (17,1) -- (15,3) -- (15,5) -- (17,7) -- cycle;
		\node at (18,7.5) {\textcolor{Red}{\LARGE $K_N^-$}};
		
      	\draw [-, line width=2pt, Blue, pattern=north east lines, pattern color=Blue]  (19,-7) -- (21,-9) -- (21,-11) -- (19,-13) -- (17,-13) -- (15,-11) -- (15,-9) -- (17,-7) -- cycle;
		\node at (18,-13.5) {\textcolor{Blue}{\LARGE $K_N^+$}};
		
		\draw [-, line width=1.5pt, Red, densely dashdotted]  ([xshift=5pt]12,0) -- ([yshift=-5pt]15,3);
		\draw [-, line width=1.5pt, Red, densely dashdotted]  ([yshift=5pt]14,-2) -- ([xshift=-5pt]17,1);
		\draw [-, line width=1.5pt, Red, densely dashdotted]  ([xshift=8pt,yshift=-2pt]12,0) -- ([xshift=-8pt,yshift=2pt]17,1);
		\draw [-, line width=1.5pt, Red, densely dashdotted]  ([xshift=-2pt,yshift=8pt]14,-2) -- ([xshift=2pt,yshift=-8pt]15,3);
		
		\draw [-, line width=1.5pt, Blue]  ([xshift=5pt]12,-6) -- ([yshift=5pt]15,-9);
		\draw [-, line width=1.5pt, Blue]  ([yshift=-5pt]14,-4) -- ([xshift=-5pt]17,-7);
		\draw [-, line width=1.5pt, Blue]  ([xshift=8pt,yshift=2pt]12,-6) -- ([xshift=-8pt,yshift=-2pt]17,-7);
		\draw [-, line width=1.5pt, Blue]  ([xshift=-2pt,yshift=-8pt]14,-4) -- ([xshift=2pt,yshift=8pt]15,-9);
		
		\draw [-, line width=1.5pt, Red, densely dashdotted]  ([xshift=5pt,yshift=-5pt]21,5) to[out=0,in=0,bend left=40] ([xshift=5pt,yshift=5pt]21,-11);
		\draw [-, line width=1.5pt, Red, densely dashdotted]  ([xshift=5pt,yshift=4pt]21,2.8) to[out=0,in=0,bend left=22] ([xshift=5pt,yshift=-4pt]21,-8.8);
		\draw [-, line width=1.5pt, Red, densely dashdotted]  ([xshift=5pt,yshift=-5pt]21,4.1) to[out=0,in=0,bend left=34] ([xshift=5pt,yshift=5pt]21,-10.1);
     	\end{tikzpicture}
     }

\caption{Illustration of the reduction in the proof of Theorem~\ref{theorem:SAN-npc}.
Red shapes represent signed graphs with all edges negative, while blue shapes represent signed graphs with all edges positive.
Dashed red edges are negative edges, while solid blue edges are positive edges.  
\label{figure:npc}}
\end{figure}
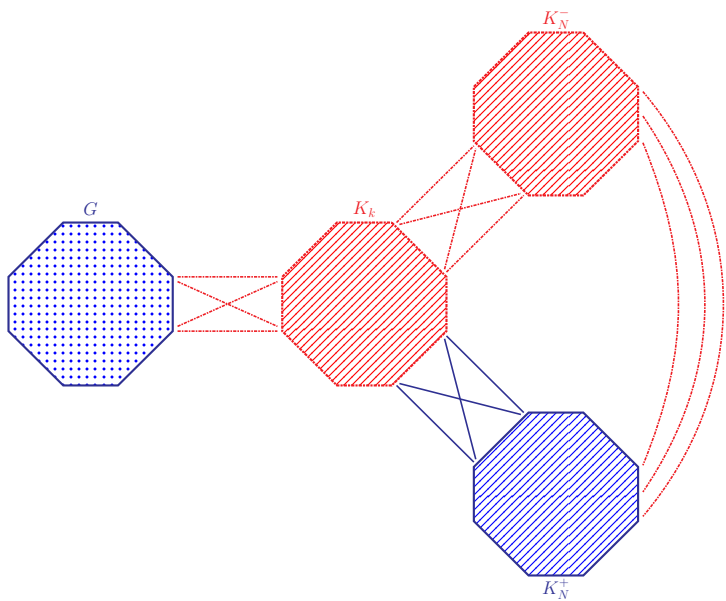

\begin{theorem}\label{theorem:SAN-npc}
\SAN is \np-hard.
\end{theorem}

\begin{proof}
The proof is by reduction from \AN, the problem of deciding whether, for an input unsigned graph $G$ and an input integer $k \geq 1$,
we have $\psi(G) \geq k$, which is known to be \np-hard (see e.g.~\cite{FHHM86}).
Let $G$ and $k$ form an instance of that problem, where we set $n=|V(G)|$.
From $G$ and $k$, we construct, in polynomial time, a signed graph $(G',\sigma)$ such that $\psi(G) \geq k$ if and only if $\psi(G',\sigma) \geq 2N+2k$,
where $N = (n+k)^2$.
Note that, $2N+2k$ being even, we have $\pm 0 \not \in M_{2N+2k}$.

The construction of $(G',\sigma)$ goes as follows (see Figure~\ref{figure:npc}).
We start from the signed graph containing a copy of $G$ with all edges positive.
Let us denote by $g_1,\dots,g_n$ the vertices of that copy.
We then add to the signed graph a disjoint complete graph $K_k$ of order~$k$ with vertices $c_1,\dots,c_k$ and all edges negative.
We also add all possible negative edges between the $g_i$'s and $c_i$'s,
i.e., we add the negative edge $g_ic_j$ for every $i \in \{1,\dots,n\}$ and $j \in \{1,\dots,k\}$.
Note that $N$ is strictly larger than the total number of edges we have added to the signed graph so far.
Next, we add, to the signed graph, two disjoint complete graphs $K^-_N$ and $K^+_N$, both of order $N$ and
having their edges signed so that all edges of $K^-_N$ are negative while all edges of $K^+_N$ are positive.
We denote by $v_1^-, \dots v_N^-$ and $v_1^+, \dots, v_N^+$ their vertices, respectively.
Lastly, we add a negative edge $v_i^-v_i^+$ for every $i \in \{1,\dots, N\}$,
all possible negative edges between the $v_i^-$'s and the $c_i$'s,
and all possible positive edges between the $v_i^+$'s and the $c_i$'s.
This achieves the construction of $(G',\sigma)$, which is clearly performed in polynomial time.

As mentioned earlier, we claim that $\psi(G) \geq k$ if and only if $\psi(G',\sigma) \geq 2N+2k$.
We prove the two directions of this equivalence.
Assume first that $G$ admits a complete $k'$-colouring $\phi$, where $k' \geq k$.
We claim that colouring $(G',\sigma)$ as follows, results in a complete $(2N+2k)$-colouring $\phi'$.
For every vertex $g_i$ corresponding to a vertex $v$ in $G$, we set $\phi'(g_i)=+\phi(v)$.
We next set $\phi'(c_i)=+i$ for every $i \in \{1,\dots,k\}$.
Due to the signs of the edges of $(G',\sigma)$ joining the $g_i$'s and the $c_i$'s, note that, already,
for every two $i,j \in \{1,\dots,k\}$,
we have the required edges of type $(i,j)$, i.e., both a p-edge and an n-edge of type $(i,j)$ if $i \neq j$,
and an n-edge of type $(i,i)$ for $i=j$ (but no p-edge of type $(i,i)$).
Finally set $\phi'(v_i^-)=\phi'(v_i^+)=+(k+i)$ for every $i \in \{1,\dots,N\}$.
Due to the signs of the edges joining the vertices of $K_k$, $K^-_N$ and $K^+_N$,
and the fact that $K^-_N$ and $K^+_N$ have all their edges being negative and positive, respectively,
we get that we also have the desired types of edges for pairs of colours in $\{k+1,\dots,N+k\}$,
and for every combination of a colour in $\{1,\dots,k\}$ and a colour in $\{k+1,\dots,N+k\}$.
Thus, $\phi'$ is complete, and we are done.

Now assume that $(G',\sigma)$ can have its vertices switched, so that the resulting signed graph admits a complete colouring $\phi'$ with at least $2N+2k$ colours.
Due to Observation~\ref{observation:large-matching}, because of the number of vertices of $G'$,
note that $\phi'$ must actually be a $(2N+2k)$-colouring.
Free to consider signed colours, we may instead work in $(G',\sigma)$ directly, with the complete $(2N+2k)$-colouring $\gamma'$ inferred from $\phi'$.
Note first that all vertices in $V(K_k) \cup V(K^-_N)$ must be assigned distinct colours by $\gamma'$.
Indeed, assume the contrary. Because the $c_i$'s and the $v_i^-$'s induce a negative complete graph,
then two such vertices $x$ and $y$ assigned the same colour can be assumed, by Lemma~\ref{lemma:inferred-force-colours}, to be assigned colour $1^+$ by $\gamma'$.
Since $x$ and $y$ have at least $N$ common neighbours, each time joined via negative edges,
we get that there are at least two p-edges and n-edges of type $(1,i)$.
From this, we deduce that the total number of different types of edges by $\phi'$,
is at most $$(n+k)^2 + N^2 - N + 2Nk < (N+k)^2.$$
Since $(N+k)^2$ is the number of edges of $K_{(2N+2k)}^*$, it follows that $\phi'$ cannot realise all required types of edges,
and thus cannot be complete, a contradiction.

We may thus assume that every two vertices in $V(K_k) \cup V(K^-_N)$ are assigned distinct colours by $\gamma'$,
and, by Lemma~\ref{lemma:inferred-force-colours}, we may assume that each of their colours is of the form $i^+$.
Since all edges joining vertices in $V(K_k) \cup V(K^-_N)$ are negative,
all these edges yield n-edges.
Because the vertices not in $K_k$ and $K^-_N$ are the $g_i$'s and the $c^+_i$'s,
which form a set of $N+k$ vertices,
then,
due to the signs of the edges in $(G',\sigma)$,
to have all required types of p-edges, 
the $g_i$'s must be assigned $k$ distinct colours by $\gamma'$, 
and these $k$ colours must thus form distinct types of p-edges.
In other words, these $k$ colours induce a complete $k$-colouring of $G$.
Hence, we have $\psi(G) \geq k$.
\end{proof}

%%%%%%%%%%%%%%%%%%%%%%%%%%%%%%%%%%%%%%%%%%%%%%%%%%%%%%
%%%%%%%%%%%%%%%%%%%%%%%%%%%%%%%%%%%%%%%%%%%%%%%%%%%%%%
%%%%%%%%%%%%%%%%%%%%%%%%%%%%%%%%%%%%%%%%%%%%%%%%%%%%%%
%%%%%%%%%%%%%%%%%%%%%%%%%%%%%%%%%%%%%%%%%%%%%%%%%%%%%%
%%%%%%%%%%%%%%%%%%%%%%%%%%%%%%%%%%%%%%%%%%%%%%%%%%%%%%
%%%%%%%%%%%%%%%%%%%%%%%%%%%%%%%%%%%%%%%%%%%%%%%%%%%%%%
%%%%%%%%%%%%%%%%%%%%%%%%%%%%%%%%%%%%%%%%%%%%%%%%%%%%%%
%%%%%%%%%%%%%%%%%%%%%%%%%%%%%%%%%%%%%%%%%%%%%%%%%%%%%%
%%%%%%%%%%%%%%%%%%%%%%%%%%%%%%%%%%%%%%%%%%%%%%%%%%%%%%
%%%%%%%%%%%%%%%%%%%%%%%%%%%%%%%%%%%%%%%%%%%%%%%%%%%%%%
%%%%%%%%%%%%%%%%%%%%%%%%%%%%%%%%%%%%%%%%%%%%%%%%%%%%%%
%%%%%%%%%%%%%%%%%%%%%%%%%%%%%%%%%%%%%%%%%%%%%%%%%%%%%%
%%%%%%%%%%%%%%%%%%%%%%%%%%%%%%%%%%%%%%%%%%%%%%%%%%%%%%
%%%%%%%%%%%%%%%%%%%%%%%%%%%%%%%%%%%%%%%%%%%%%%%%%%%%%%

\section{Some significant differences with unsigned graphs} \label{section:differences}

We here discuss three lines of research investigated for the achromatic number of unsigned graphs which,
when generalised in the most obvious way possible to the signed context,
lead to partial discrepancies. These three lines of investigations cover irreducible graphs,
perfect graphs, and the Homomorphism Interpolation Theorem.

\subsection*{Irreducible graphs}

In unsigned graphs, the notion of irreducible graphs arises from the observation that,
although a graph with large achromatic number must have lots of edges,
the contrary is not true, as there exist graphs with lots of edges but small achromatic number.
To be convinced of this statement, just note that $\psi(K_{n,n})=2$ for every $n \geq 2$.
Note, however, that such complete bipartite graphs are very pathological,
as they have many vertices interacting the same way, in terms of neighbourhood, with the rest of the graph (which equivalent vertices can thus be regarded as single vertices).

The irreducibility of graphs was introduced to measure such aspects.
The definitions are as follows.
Let $G$ be a graph. Two non-adjacent vertices $u$ and $v$ of $G$ are part of the same \textit{congruence class} $R$ if they have the same neighbourhood, i.e., $N(u)=N(v)$.
The \textit{reduced graph} $G / R$ is the graph whose vertices are the congruence classes of $G$,
and in which two vertices are joined by an edge if the two corresponding congruence classes of $G$ contain vertices that are adjacent in $G$.
In other words, $G/R$ is obtained from $G$ by contracting every congruence class to a single vertex (and keeping the graph simple).
We say that $G$ is \textit{irreducible} if $G/R$ is $G$, i.e., every congruence class is a singleton.

Irreducible graphs and their connection with the achromatic number led to several interesting investigations, see~\cite{HM97}.
Regarding our investigations in this paper, we will restrict our attention to two particular related results, being:

\begin{lemma}[Hell, Miller~\cite{HM76b}]\label{lemma:irreducible-unsigned}
If $G$ is an irreducible graph and $\psi(G) \leq k$, then:
\begin{itemize}
	\item $G$ has at most $k+1 \choose 2$ components, and
	\item for every component $G'$ of $G$, the diameter of $G'$ is at most $\left\lfloor \frac{k+3}{2} \right\rfloor \cdot (k-1)$.
\end{itemize}
\end{lemma}

\begin{theorem}[Hell, Miller~\cite{HM76b}]\label{theorem:irreducible-unsigned}
Let $k \geq 1$. There is a constant $K$ such that, for every irreducible graph $G$ with $\psi(G) \leq k$, we have $|V(G)| \leq K$.
\end{theorem}

One point of interest behind Theorem~\ref{theorem:irreducible-unsigned} is that it implies that the number of irreducible graphs with given achromatic number is finite.
In what follows, we essentially provide, in the context of signed graphs,
a result that is reminiscent of Lemma~\ref{lemma:irreducible-unsigned}, 
and, to the contrary, a proof that Theorem~\ref{theorem:irreducible-unsigned} does not adapt immediately.

We generalise the notions around irreducible unsigned graphs to our context,
in the following way.
Let $(G,\sigma)$ be a signed graph. For a vertex $u \in V(G)$, we define its \textit{signed neighbourhood} as the set $$\{ (v,s) \in V(G) \times \{+,-\} : uv \in E(G), s = \sigma(uv) \}.$$
Now, we say that two non-adjacent vertices $u$ and $v$ of $(G,\sigma)$ belong to the same \textit{congruence class}
if there is a signed graph equivalent to $(G,\sigma)$ in which $u$ and $v$ have the same signed neighbourhood.
We define then the \textit{reduced signed graph} $(G / R^\sigma, \sigma)$ of $(G,\sigma)$ as the signed graph obtained by contracting every congruence class in $(G,\sigma)$ to a single vertex (with keeping this signed graph simple).
Finally, we say that $(G,\sigma)$ is \textit{irreducible} if all its vertices have distinct signed neighbourhoods,
that is if $(G,\sigma)$ is $(G/R^\sigma,\sigma)$.

We start by proving a result being analogous to Lemma~\ref{lemma:irreducible-unsigned}, but in the signed context.

\begin{lemma}\label{lemma:irreducible-signed}
Let $(G,\sigma)$ be an irreducible signed graph. If, for some $p \geq 1$, we have
\begin{itemize}
	\item $\psi(G,\sigma) \leq 2p$, then
	\begin{itemize}
    	\item $G$ contains at most $(p+1)^2 -1$ components, and
    	\item for every induced path $P_n$ of order $n$ in $G$, we have $n \leq (p+1)^2 -2$;
	\end{itemize}
	
	\item $\psi(G,\sigma) \leq 2p-1$, then
	\begin{itemize}
    	\item $G$ contains at most $p^2$ components, and
    	\item for every induced path $P_n$ of order $n$ in $G$, we have $n \leq p^2 -1$.
	\end{itemize}
\end{itemize}
\end{lemma}

\begin{proof}
We prove the claim for the first item, the proof being identical for the second one.
First, let us suppose, towards a contradiction, that $G$ has at least $(p+1)^2$ components.
Because $(G,\sigma)$ is irreducible, at most one of these components is an isolated vertex.
Thus, $(G,\sigma)$ has at least one edge in each component and thus an induced matching of size at least $(p+1)^2 -1$, and so $\psi(G,\sigma) \geq 2p+1$ by Corollary~\ref{corollary:achromatic-signed-monotonous-subgraphs}, a contradiction.
Now,  assume $G$ has an induced path $P_n$ with $n \ge (p+1)^2-1$. Then, by Corollary~\ref{corollary:achromatic-signed-monotonous-subgraphs} and Theorem~\ref{theorem:pn}, we have $\psi(G,\sigma) \ge 2p+1$, another contradiction.
\end{proof}

We now observe that the statement that is analogous to that of Theorem~\ref{theorem:irreducible-unsigned}
does not hold immediately in our context. Namely:

\begin{theorem}\label{theorem:irreducible-signed}
Let $k \geq 1$. There exist arbitrarily large irreducible signed graphs $(G,\sigma)$ with $\psi(G,\sigma) \leq k$. 
\end{theorem}

\begin{proof}
As seen in the proof of Theorem~\ref{theorem:homomorphism-signed2}, any complete graph on at least $2p$ vertices,
signed so that all edges but those in a matching of size $p$ are negative, has achromatic number exactly $2p+2$.
Furthermore, by definition, every signed complete graph is irreducible.
Since these arguments hold for arbitrarily large such signed graphs, the result follows.
\end{proof}

%%%%%%%%%%%%%%%%%%%%%%%%%%%%%%%%
%%%%%%%%%%%%%%%%%%%%%%%%%%%%%%%%
%%%%%%%%%%%%%%%%%%%%%%%%%%%%%%%%
%%%%%%%%%%%%%%%%%%%%%%%%%%%%%%%%
%%%%%%%%%%%%%%%%%%%%%%%%%%%%%%%%
%%%%%%%%%%%%%%%%%%%%%%%%%%%%%%%%
%%%%%%%%%%%%%%%%%%%%%%%%%%%%%%%%
%%%%%%%%%%%%%%%%%%%%%%%%%%%%%%%%
%%%%%%%%%%%%%%%%%%%%%%%%%%%%%%%%
%%%%%%%%%%%%%%%%%%%%%%%%%%%%%%%%
%%%%%%%%%%%%%%%%%%%%%%%%%%%%%%%%
%%%%%%%%%%%%%%%%%%%%%%%%%%%%%%%%
%%%%%%%%%%%%%%%%%%%%%%%%%%%%%%%%
%%%%%%%%%%%%%%%%%%%%%%%%%%%%%%%%
%%%%%%%%%%%%%%%%%%%%%%%%%%%%%%%%
%%%%%%%%%%%%%%%%%%%%%%%%%%%%%%%%
%%%%%%%%%%%%%%%%%%%%%%%%%%%%%%%%
%%%%%%%%%%%%%%%%%%%%%%%%%%%%%%%%
%%%%%%%%%%%%%%%%%%%%%%%%%%%%%%%%
%%%%%%%%%%%%%%%%%%%%%%%%%%%%%%%%
%%%%%%%%%%%%%%%%%%%%%%%%%%%%%%%%

\subsection*{Perfect graphs}

We here deal with perfect graph colourings, as investigated first by Christen and Selkow in~\cite{CS79}.
Let us first recall what Grundy colourings are.
Let $G$ be a graph. A $k$-colouring $\phi$ of $G$ is called a \textit{Grundy colouring} if every vertex $v \in V(G)$ is adjacent to at least one vertex with colour $i$ for every $i<\phi(v)$.
The \textit{Grundy number} $\gamma(G)$ of $G$ is the smallest $k$ such that $G$ admits Grundy $k$-colourings.
Recall that $\omega(G)$ refers to the \textit{clique number} of $G$, which is the order of the largest complete subgraph of $G$.
Note that the four chromatic parameters $\omega$, $\chi$, $\gamma$ and $\psi$ are quite related,
as, for every graph $G$, we have $\omega(G) \leq \chi(G) \leq \gamma(G) \leq \psi(G)$.

For any two distinct parameters $\alpha,\beta \in \{\omega,\chi,\gamma,\psi\}$, 
we say that $G$ is \textit{$(\alpha,\beta)$-perfect} if $\alpha(H)=\beta(H)$ for every induced subgraph $H$ of $G$.
Note that this notion encapsulates the more famous notion of perfect graphs, which are precisely $\{\omega,\chi\}$-perfect graphs.
An interesting and legitimate question is about other combinations of parameters as $\alpha$ and $\beta$.
Regarding $(\psi,\omega)$-perfect graphs and $(\psi,\chi)$-perfect graphs, Christen and Selkow provided the following characterisation:

\begin{theorem}\label{theorem:perfect-unsigned}
Let $G$ be a graph. The following statements are equivalent:
\begin{enumerate}
    \item $G$ is $(\psi,\omega)$-perfect;
    
    \item $G$ is $(\psi,\chi)$-perfect;
    
    \item $G$ does not contain one of $P_4$, $P_3 + P_2$ and $P_2 + P_2 + P_2$ as an induced subgraph;
    
    \item no homomorphic image of $G$ contains an induced subgraph isomorphic to $P_4$.
\end{enumerate}
\end{theorem}

We prove that Theorem~\ref{theorem:perfect-unsigned} does not apply to the signed context,
at least not when derived in the way below.
Let $(G,\sigma)$ be a signed graph. 
In what follows, we define the \textit{clique number} $\omega(G,\sigma)$ of $(G,\sigma)$ as being $\omega(G)$, the clique number of the underlying graph $G$.
Now, for any two distinct parameters $\alpha, \beta \in \{\omega,\chi,\psi\}$, we say that $(G,\sigma)$ is \textit{$(\alpha,\beta)$-perfect} if $\alpha(H,\sigma)=\beta(H,\sigma)$ for every induced subgraph $H$ of $G$.

The result we prove is, essentially, that there is a signed graph fulfilling the third item in Theorem~\ref{theorem:perfect-unsigned},
but none of the first two items.

\begin{figure}[!t]
 	\centering
 	\subfloat[$\chi(G,\sigma) \leq 3$]{
    \scalebox{0.6}{
	\begin{tikzpicture}[inner sep=0.7mm]	
	
	\node[draw, circle, black, line width=1pt](u1) at (0,0)[]{$+1$};
	\node[draw, circle, black, line width=1pt](u2) at (0,-3)[]{$-1$};
	
	\node[draw, circle, black, line width=1pt](v1) at (3,0)[]{$\pm 0$};
	\node[draw, circle, black, line width=1pt](v2) at (3,-3)[]{$\pm 0$};
	
%	\node[draw, circle, black, line width=1pt](w1) at (6,1.5)[]{$+1$};
	\node[draw, circle, black, line width=1pt](w2) at (6,0)[]{$-1$};
	\node[draw, circle, black, line width=1pt](w3) at (6,-3)[]{$+1$};

    \draw [-, line width=1.5pt, Red, densely dashdotted] (u1) -- (v1);
    \draw [-, line width=1.5pt, Red, densely dashdotted] (u1) -- (v2);
    \draw [-, line width=1.5pt, Red, densely dashdotted] (u2) -- (v1);
    \draw [-, line width=1.5pt, Red, densely dashdotted] (u2) -- (v2);
%    \draw [-, line width=1.5pt, Red, densely dashdotted] (v2) -- (w1);
    \draw [-, line width=1.5pt, Red, densely dashdotted] (v2) -- (w2);
    \draw [-, line width=1.5pt, Red, densely dashdotted] (v2) -- (w3);
    
    \draw  [-, line width=1.5pt, Blue] (u1) -- (u2);
%    \draw  [-, line width=1.5pt, Blue] (v1) -- (w1);
    \draw  [-, line width=1.5pt, Blue] (v1) -- (w2);
    \draw  [-, line width=1.5pt, Blue] (v1) -- (w3);
%    \draw  [-, line width=1.5pt, Blue] (w1) -- (w2);
    \draw  [-, line width=1.5pt, Blue] (w2) -- (w3);
	\end{tikzpicture}
    }
    }
    \hspace{50pt}
 	\subfloat[$\psi(G,\sigma) \geq 4$]{
    \scalebox{0.6}{
	\begin{tikzpicture}[inner sep=0.7mm]	
	
	\node[draw, circle, black, line width=1pt](u1) at (0,0)[]{$+1$};
	\node[draw, circle, black, line width=1pt](u2) at (0,-3)[]{$-1$};
	
	\node[draw, circle, black, line width=1pt](v1) at (3,0)[]{$-2$};
	\node[draw, circle, black, line width=1pt](v2) at (3,-3)[]{$+2$};
	
%	\node[draw, circle, black, line width=1pt](w1) at (6,1.5)[]{$+1$};
	\node[draw, circle, black, line width=1pt](w2) at (6,0)[]{$+2$};
	\node[draw, circle, black, line width=1pt](w3) at (6,-3)[]{$+1$};

    \draw [-, line width=1.5pt, Red, densely dashdotted] (u1) -- (v1);
    \draw [-, line width=1.5pt, Red, densely dashdotted] (u1) -- (v2);
    \draw [-, line width=1.5pt, Red, densely dashdotted] (u2) -- (v1);
    \draw [-, line width=1.5pt, Red, densely dashdotted] (u2) -- (v2);
%    \draw [-, line width=1.5pt, Red, densely dashdotted] (v2) -- (w1);
    \draw [-, line width=1.5pt, Red, densely dashdotted] (v2) -- (w2);
    \draw [-, line width=1.5pt, Red, densely dashdotted] (v2) -- (w3);
    
    \draw  [-, line width=1.5pt, Blue] (u1) -- (u2);
%    \draw  [-, line width=1.5pt, Blue] (v1) -- (w1);
    \draw  [-, line width=1.5pt, Blue] (v1) -- (w2);
    \draw  [-, line width=1.5pt, Blue] (v1) -- (w3);
%    \draw  [-, line width=1.5pt, Blue] (w1) -- (w2);
    \draw  [-, line width=1.5pt, Blue] (w2) -- (w3);
	\end{tikzpicture}
    }
    }

\caption{The signed graph $(G,\sigma)$ mentioned in the proof of Theorem~\ref{theorem:perfect-signed}.
In (a) is depicted a proper $3$-colouring of $(G,\sigma)$, and in (b) a complete $4$-colouring.
Dashed red edges are negative edges, while solid blue edges are positive edges.  
\label{figure:perfect}}
\end{figure}
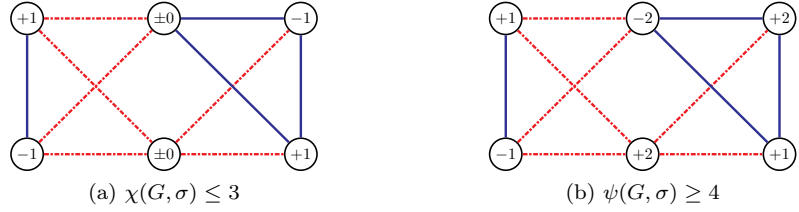

\begin{theorem}\label{theorem:perfect-signed}
There exists a signed graph $(G,\sigma)$ that does not have any of $P_4$, $P_3 + P_2$ and $P_2 + P_2 + P_2$ as an induced subgraph, 
but is neither $(\psi,\omega)$-perfect nor $(\psi,\chi)$-perfect.
\end{theorem}

\begin{proof}
This signed graph $(G,\sigma)$ is depicted in Figure~\ref{figure:perfect}.
Note that $G$ does not have any of $P_4$, $P_3 + P_2$ and $P_2 + P_2 + P_2$ as an induced subgraph.
Regarding the last part of the statement now, it can be checked that $\omega(G,\sigma)=3$, and,
as shown in Figure~\ref{figure:perfect}, we have $\chi(G,\sigma) \leq 3$ and $\psi(G,\sigma) \geq 4$.
Thus, $(G,\sigma)$ is neither $(\psi,\omega)$-perfect nor $(\psi,\chi)$-perfect.
\end{proof}

%%%%%%%%%%%%%%%%%%%%%%%%%%%%%%%%
%%%%%%%%%%%%%%%%%%%%%%%%%%%%%%%%
%%%%%%%%%%%%%%%%%%%%%%%%%%%%%%%%
%%%%%%%%%%%%%%%%%%%%%%%%%%%%%%%%
%%%%%%%%%%%%%%%%%%%%%%%%%%%%%%%%
%%%%%%%%%%%%%%%%%%%%%%%%%%%%%%%%
%%%%%%%%%%%%%%%%%%%%%%%%%%%%%%%%
%%%%%%%%%%%%%%%%%%%%%%%%%%%%%%%%
%%%%%%%%%%%%%%%%%%%%%%%%%%%%%%%%
%%%%%%%%%%%%%%%%%%%%%%%%%%%%%%%%
%%%%%%%%%%%%%%%%%%%%%%%%%%%%%%%%
%%%%%%%%%%%%%%%%%%%%%%%%%%%%%%%%
%%%%%%%%%%%%%%%%%%%%%%%%%%%%%%%%
%%%%%%%%%%%%%%%%%%%%%%%%%%%%%%%%
%%%%%%%%%%%%%%%%%%%%%%%%%%%%%%%%
%%%%%%%%%%%%%%%%%%%%%%%%%%%%%%%%
%%%%%%%%%%%%%%%%%%%%%%%%%%%%%%%%
%%%%%%%%%%%%%%%%%%%%%%%%%%%%%%%%
%%%%%%%%%%%%%%%%%%%%%%%%%%%%%%%%
%%%%%%%%%%%%%%%%%%%%%%%%%%%%%%%%
%%%%%%%%%%%%%%%%%%%%%%%%%%%%%%%%

\subsection*{The Homomorphism Interpolation Theorem}

In the line of our results in Subsection~\ref{subsection:homomorphisms-signed-graphs},
an interesting result on the achromatic number that we can try to bring from the unsigned context
is the so-called Homomorphism Interpolation Theorem~\cite{Harary1967}.
Roughly put, it says that unsigned graphs $G$ admit complete $k$-colourings for every theoretically legit value of $k$,
i.e., ranging in $\{\chi(G),\dots,\psi(G)\}$:

\begin{theorem}[Harary, Hedetniemi, Prins~\cite{Harary1967}]\label{theorem:interpo-unsigned}
For every graph $G$ and every integer $k$ such that $\chi(G) \leq k \leq \psi(G)$, there is an edge-surjective homomorphism of $G$ to $K_k$.
\end{theorem}

It turns out that Theorem~\ref{theorem:interpo-unsigned} does not adapt to our context, at least not in the most direct way possible.
Namely, the following holds true:

\begin{figure}[!t]
 	\centering
 	\subfloat[An inferred complete $6$-colouring]{
    \scalebox{0.75}{
	\begin{tikzpicture}[inner sep=0.7mm]	
		\node[draw, circle, black, line width=1pt](v1) at (1*360/6:2.8cm){$1^+$};	
		\node[draw, circle, black, line width=1pt](v2) at (2*360/6:2.8cm){$2^-$};	
		\node[draw, circle, black, line width=1pt](v3) at (3*360/6:2.8cm){$1^-$};	
		\node[draw, circle, black, line width=1pt](v4) at (4*360/6:2.8cm){$3^+$};	
		\node[draw, circle, black, line width=1pt](v5) at (5*360/6:2.8cm){$3^-$};	
		\node[draw, circle, black, line width=1pt](v6) at (6*360/6:2.8cm){$2^+$};	

    	\draw [-, line width=1.5pt, Red, densely dashdotted] (v1) -- (v2);		
    	\draw [-, line width=1.5pt, Red, densely dashdotted] (v3) -- (v4);		
    	\draw [-, line width=1.5pt, Red, densely dashdotted] (v5) -- (v6);		
		
    	\draw [-, line width=1.5pt, Blue] (v1) -- (v3);
    	\draw [-, line width=1.5pt, Blue] (v1) -- (v4);
    	\draw [-, line width=1.5pt, Blue] (v1) -- (v5);
    	\draw [-, line width=1.5pt, Blue] (v1) -- (v6);
    	
    	\draw [-, line width=1.5pt, Blue] (v2) -- (v3);
    	\draw [-, line width=1.5pt, Blue] (v2) -- (v4);
    	\draw [-, line width=1.5pt, Blue] (v2) -- (v5);
    	\draw [-, line width=1.5pt, Blue] (v2) -- (v6);
    	
    	\draw [-, line width=1.5pt, Blue] (v3) -- (v5);
    	\draw [-, line width=1.5pt, Blue] (v3) -- (v6);
    	
    	\draw [-, line width=1.5pt, Blue] (v4) -- (v5);
    	\draw [-, line width=1.5pt, Blue] (v4) -- (v6);
	\end{tikzpicture}
    }
    }
    \hspace{40pt}
 	\subfloat[An inferred complete $4$-colouring]{
    \scalebox{0.75}{
	\begin{tikzpicture}[inner sep=0.7mm]	
		\node[draw, circle, black, line width=1pt](v1) at (1*360/6:2.8cm){$1^+$};	
		\node[draw, circle, black, line width=1pt](v2) at (2*360/6:2.8cm){$1^+$};	
		\node[draw, circle, black, line width=1pt](v3) at (3*360/6:2.8cm){$2^+$};	
		\node[draw, circle, black, line width=1pt](v4) at (4*360/6:2.8cm){$2^+$};	
		\node[draw, circle, black, line width=1pt](v5) at (5*360/6:2.8cm){$1^-$};	
		\node[draw, circle, black, line width=1pt](v6) at (6*360/6:2.8cm){$1^-$};	

    	\draw [-, line width=1.5pt, Red, densely dashdotted] (v1) -- (v2);		
    	\draw [-, line width=1.5pt, Red, densely dashdotted] (v3) -- (v4);		
    	\draw [-, line width=1.5pt, Red, densely dashdotted] (v5) -- (v6);		
		
    	\draw [-, line width=1.5pt, Blue] (v1) -- (v3);
    	\draw [-, line width=1.5pt, Blue] (v1) -- (v4);
    	\draw [-, line width=1.5pt, Blue] (v1) -- (v5);
    	\draw [-, line width=1.5pt, Blue] (v1) -- (v6);
    	
    	\draw [-, line width=1.5pt, Blue] (v2) -- (v3);
    	\draw [-, line width=1.5pt, Blue] (v2) -- (v4);
    	\draw [-, line width=1.5pt, Blue] (v2) -- (v5);
    	\draw [-, line width=1.5pt, Blue] (v2) -- (v6);
    	
    	\draw [-, line width=1.5pt, Blue] (v3) -- (v5);
    	\draw [-, line width=1.5pt, Blue] (v3) -- (v6);
    	
    	\draw [-, line width=1.5pt, Blue] (v4) -- (v5);
    	\draw [-, line width=1.5pt, Blue] (v4) -- (v6);
	\end{tikzpicture}
    }
    }

\caption{One of the signed graphs mentioned in the proof of Theorem~\ref{theorem:no-interpolation-signed}.
In (a) is also depicted an inferred complete $6$-colouring, and in (b) an inferred complete $4$-colouring.
Dashed red edges are negative edges, while solid blue edges are positive edges. 
\label{figure:interpo}}
\end{figure}
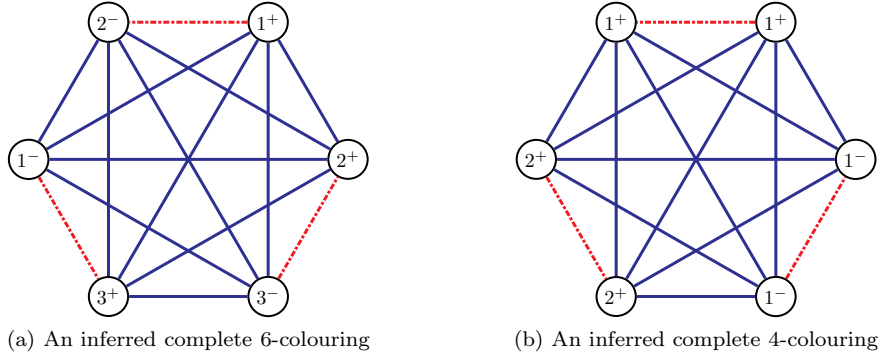

\begin{theorem}\label{theorem:no-interpolation-signed}
Let $k \ge 3$ be an odd integer. 
There exists a signed graph $(G,\sigma)$ such that $\chi(G, \sigma) < k < \psi(G, \sigma)$ and $(G, \sigma)$ admits no complete $k$-colourings.
\end{theorem}

\begin{proof}
The construction of $(G,\sigma)$ goes as follows (see Figure~\ref{figure:interpo}).
Start from a complete graph $K_k$ of order~$k$ with all edges positive.
Denote its vertices by $u_1,\dots,u_k$.
Now consider every vertex $u_i$, and replace $u_i$ with two new vertices $v_i$ and $w_i$ joined through a negative edge,
and connected to all other vertices through positive edges.
What results, as $(G,\sigma)$, is thus a positive complete graph of order $2k$, in which the edges of a perfect matching have been turned negative.

We claim that $\chi(G, \sigma) \leq 2 \lceil \frac{k}{2} \rceil $ and $\psi(G, \sigma) = 2k$,
but that $(G,\sigma)$ admits no complete colouring with an odd number of colours in-between $\chi(G,\sigma)$ and $\psi(G,\sigma)$.
On the one hand, to see first that $\chi(G, \sigma) \leq 2 \lceil \frac{k}{2} \rceil $, it suffices to note, by checking carefully the resulting edge types (remember we want the colouring to be proper, not necessarily complete),
that, by assigning colour $+i$ to every two vertices $v_{2i}$ and $w_{2i}$ with even index $2i$,
and colour $-i$ to every two vertices $v_{2i-1}$ and $w_{2i-1}$ with odd index $2i-1$, we obtain a proper $2 \lceil \frac{k}{2} \rceil$-colouring of $(G,\sigma)$.
On the other hand, to see that $\psi(G, \sigma) = 2k$ (note that we cannot have $\psi(G, \sigma) > 2k$ due to Observation~\ref{observation:large-matching}), it suffices to consider the colouring of $(G,\sigma)$ assigning colours $+i$ and $-(i+1)$ to every two vertices $v_i$ and $w_i$ with index $i \in \{1,\dots,k-1\}$, and colours $+k$ and $-1$ to $v_k$ and $w_k$. It can be checked that this yields a complete $2k$-colouring of $(G,\sigma)$.
Particularly, every colour $\pm i$ appears twice (both $+i$ and $-i$ are assigned once, to vertices joined by a positive edge) to produce the desired n-edge of type $(i,i)$, and, due to the graph's signature,
all other required types of p-edges and n-edges appear (without creating any p-edge of type $(i,i)$).

Towards a contradiction, assume, finally, that $(G,\sigma)$ 
can have its vertices switched, 
so that the resulting signed graph 
admits a complete colouring $\phi$ with an odd number of colours in the range mentioned above.
Because $k \geq 3$, note that the range $\{\chi(G,\sigma),\dots,\psi(G,\sigma)\}$ indeed contains odd integers.
Free to consider signed colours instead, we may consider that colouring in $(G,\sigma)$ directly, through the complete colouring $\gamma$ inferred from $\phi$.
Because $G$ is a complete graph, colour $\pm 0$ must be assigned only once by $\phi$. 
Assume $\phi(v_i)=\pm 0$ for some vertex $v_i$, and consider the colour $\phi(w_i)$, which, w.l.o.g., can be assumed to be $\pm 1$.
By Lemma~\ref{lemma:inferred-force-colours}, we may suppose that $\gamma(v_i)=0^+$ and $\gamma(w_i)=1^+$.
Then the edge $v_iw_i$ is an n-edge of type $(0,1)$.
To get a contradiction, we claim that there cannot be a p-edge of type $(0,1)$.
Indeed, because $v_i$ and $w_i$ are connected to all other vertices of the graph through positive edges,
so that such a p-edge exists, there must be another vertex $u$ of $(G,\sigma)$ being assigned colour $1^+$ by $\gamma$.
But that vertex, together with $w_i$, would form a p-edge of type $(1,1)$, contradicting that $\phi$ is proper. 
\end{proof}

%%%%%%%%%%%%%%%%%%%%%%%%%%%%%%%%%%%%%%%%%%%%%%%%%%%%%%
%%%%%%%%%%%%%%%%%%%%%%%%%%%%%%%%%%%%%%%%%%%%%%%%%%%%%%
%%%%%%%%%%%%%%%%%%%%%%%%%%%%%%%%%%%%%%%%%%%%%%%%%%%%%%
%%%%%%%%%%%%%%%%%%%%%%%%%%%%%%%%%%%%%%%%%%%%%%%%%%%%%%
%%%%%%%%%%%%%%%%%%%%%%%%%%%%%%%%%%%%%%%%%%%%%%%%%%%%%%
%%%%%%%%%%%%%%%%%%%%%%%%%%%%%%%%%%%%%%%%%%%%%%%%%%%%%%
%%%%%%%%%%%%%%%%%%%%%%%%%%%%%%%%%%%%%%%%%%%%%%%%%%%%%%
%%%%%%%%%%%%%%%%%%%%%%%%%%%%%%%%%%%%%%%%%%%%%%%%%%%%%%
%%%%%%%%%%%%%%%%%%%%%%%%%%%%%%%%%%%%%%%%%%%%%%%%%%%%%%
%%%%%%%%%%%%%%%%%%%%%%%%%%%%%%%%%%%%%%%%%%%%%%%%%%%%%%
%%%%%%%%%%%%%%%%%%%%%%%%%%%%%%%%%%%%%%%%%%%%%%%%%%%%%%
%%%%%%%%%%%%%%%%%%%%%%%%%%%%%%%%%%%%%%%%%%%%%%%%%%%%%%
%%%%%%%%%%%%%%%%%%%%%%%%%%%%%%%%%%%%%%%%%%%%%%%%%%%%%%
%%%%%%%%%%%%%%%%%%%%%%%%%%%%%%%%%%%%%%%%%%%%%%%%%%%%%%

\section{Conclusion and discussion} \label{section:ccl}

Our main intent in this work was to introduce a line of research on the achromatic number of signed graphs
that is parallel to that initiated recently by Lajou in~\cite{Laj19}.
While Lajou based his definitions on colourings of signed graphs as introduced by Guenin~\cite{Gue05},
ours are based on colouring notions introduced by Zaslavsky~\cite{Zas82}.
As mentioned earlier, Guenin's and Zaslavsky's notions hardly compare in general,
and, as a consequence, our results should be regarded as independent from those obtained by Lajou.

To guide our investigations, we chose to stick to approaches and directions considered for the unsigned version of the problem,
since, as highlighted through numerous references of the literature, one of the main interesting aspects, when generalising a problem from graphs to signed graphs,
are the similarities and discrepancies between the original problem and the generalised one.
Consequently, we picked several results from~\cite{CZ09,HM97}, and tried to come up with possible ways to generalise them.
For most of our attempts, we have observed that the original problems for unsigned graphs tend to behave in a reminiscent way in the signed context,
modulo some slight modifications. As an illustration, note that most bounds we have exhibited through Sections~\ref{section:operations} and~\ref{section:complexity}
are very close to their corresponding unsigned counterparts. For other attempts, we observed more or less significant differences with the unsigned context; this is illustrated by our results from Section~\ref{section:differences}.

Throughout this work, we also did our best to take into account the very peculiarities of Zaslavsky's proper colourings of signed graphs.
Notably, a very particular of their subtleties lies in the way the number of assigned colours behaves, depending on the parity of this number.
This explains why, notably when studying the tightness of some bounds in Section~\ref{section:operations}, 
we did our best to provide constructions for all numbers of colours.
This aspect seems of interest to us, as it does not concern complete colourings of signed graphs only, 
but their proper colourings as well.

\medskip

The investigations initiated in this work open the way for way more research on the topic,
some of which would be particularly interesting to consider further.
Let us start by mentioning, as a very general perspective, that many aspects related to complete colourings of unsigned graphs
have not been covered by our results in this work,
and, thus, that many such aspects, mentioned e.g. in~\cite{CZ09,HM97}, would deserve to be brought to our context.

Regarding the results we have exhibited in this work, a few of them remain with some main or side aspects open,
which could be subject to further work. For instance:

\begin{itemize}
	\item Regarding complexity questions, 
	it would be interesting to investigate whether Theorem~\ref{theorem:SAN-npc} remains true when restricted to particular families of signed graphs.
	In the unsigned context, the corresponding problem was shown to remain \np-hard when restricted to very particular classes of graphs,
	including cographs and interval graphs~\cite{Bod89}. Note that the reduced signed graphs we produce in our reduction
	do not allow to get such an interesting strengthening of our result right away.
	
	\item Regarding perfect graphs, note that Theorem~\ref{theorem:perfect-signed}, 
	while it does refute a naive generalisation of Theorem~\ref{theorem:perfect-unsigned}, 
	would deserve to be pushed further. For instance, we are still not sure, in our context, 
	of whether the first two items in the statement of Theorem~\ref{theorem:perfect-unsigned} are equivalent or not.
	More generally speaking, it would be interesting to study which items of that theorem are equivalent in our context.
	We would also be interested, regarding Theorem~\ref{theorem:perfect-signed}, 
	in having a generalisation of the signed graph we have provided to prove this result.
	
	\item Regarding Theorem~\ref{theorem:no-interpolation-signed}, 
	we would be interested in knowing whether there exist similar signed graphs admitting no complete $k$-colouring for even values of $k$.
	In the same spirit, regarding the proof of Theorem~\ref{theorem:irreducible-signed}, 
	we would be interested in having similar families of arbitrarily large signed graphs with given fixed odd achromatic number.
	
	\item Recall that, when dealing with perfect graphs, we mentioned, along the way, the notion of Grundy colourings of unsigned graphs,
	which, at this point, was eluded from our investigations. One possible direction for further work could be to introduce a corresponding notion in our context,
	and then to investigate its connection with the analogues of the chromatic parameters that take part to Theorem~\ref{theorem:perfect-signed}.
	
	\item Regarding irreducible graphs, we wonder whether there is an alternative way to extend that notion to signed graphs,
	for which an analogue of Theorem~\ref{theorem:irreducible-unsigned} would hold.
\end{itemize}

Apart from those concerns, there are also interesting questions we ran into, 
which, as far as we know, were not even considered in the unsigned context.
One such example is the following.
While trying to come up with constructions of signed graphs showing the tightness of several of our bounds, 
one approach we considered was to come up with signed graphs, which we call \textit{uniquely $k$-colourable},
that have the very peculiar property that all their complete $k$-colourings are equivalent, up to permuting the colour classes.
While this topic has a well-studied equivalent one for proper colourings of unsigned graphs, we are not aware of any similar notion for complete colourings of unsigned graphs.
Another direction for research could thus be to investigate this notion, either for unsigned graphs or signed graphs,
and, in particular, to investigate whether it can be of any use towards some of the questions we leave open.

\end{document}